\documentclass{article}
\usepackage[a4paper, total={6in, 8in}]{geometry}
\usepackage[utf8]{inputenc}
\usepackage{graphicx} 
\usepackage{tikz}
\usepackage{caption}
\usepackage{subcaption}
\usepackage{dsfont}
\usetikzlibrary{calc}
\usepackage{graphicx}
\usepackage{lipsum}
\usepackage{float}
\usepackage{mwe}
\usepackage[rightcaption]{sidecap}
\graphicspath{ {./images/} }
\usepackage[utf8]{inputenc}
\usepackage{amsmath,amssymb,amsthm}
\usepackage{lmodern,stmaryrd}
\usepackage{wasysym}
\usepackage{sgame}
\usepackage[T1]{fontenc}
\usepackage[export]{adjustbox}
\newtheorem{theorem}{Theorem}
\newtheorem{definition}{Definition}
\newtheorem{lemma}{Lemma}

\newtheorem{example}{Example}
\newtheorem{proposition}{Proposition}

\counterwithin*{equation}{section}
\counterwithin*{equation}{subsection}
\usepackage[pdftex]{hyperref}
\usepackage{chicago}
\usepackage{enumerate,tikz,arydshln,ragged2e,wasysym}

\usepackage{apacite}
\usepackage{chngcntr}
\counterwithin*{equation}{section}
\counterwithin*{equation}{subsection}

\usepackage{bbm}

\begin{document}

\thispagestyle{plain}
\begin{center}
    \Large
    \textbf{Rationalizable Behavior in the Hotelling Model with Waiting Costs \footnote{ I would like to acknowledge Andrés Perea for his valuable contributions as my supervisor. His guidance, availability
for discussions, and willingness to provide constructive criticism have been greatly appreciated} }
        
   \vspace{0.4cm}
   \large

   \vspace{0.4cm}
    \text{Joep van Sloun} \footnote{Department of Quantitative Economics, School of Business and Economics, Maastricht University, 6200 MD
Maastricht, THE NETHERLANDS, Email: j.vansloun@maastrichtuniversity.nl
}
       \large

    \vspace{0.4cm}
    \text{\today}
    \vspace{0.9cm}

    \end{center}

\begin{abstract}
This paper revisits the Hotelling model with waiting costs \cite{kohlberg1983equilibrium}, focusing on two specific settings where pure Nash equilibria do not exist: the asymmetric model with two firms and the symmetric model with three firms. In the asymmetric two-firm model, we show that the weaker concept of point rationalizability has strong predictive power, as it selects exactly two locations for both firms. As the two firms become more similar in their efficiency in handling queues of consumers, the two point rationalizable locations converge towards the center of the line. In the symmetric three-firm model, the set of point rationalizable choices forms an interval. This interval is shrinking in the inefficiency levels of the firms in handling queues of consumers. 
\end{abstract}
\bigskip

            \vspace{0in}

\newpage
\section{Introduction}
The influential paper by \citeA{hotbllino1929stability} explores spatial competition among firms. This model has been widely utilized to analyze various economic scenarios, including political competition \cite{downs1957economic}. Several variations of the original model exist, incorporating factors such as more than two firms, quality differences between products, and pricing behaviors. In the original model, two firms must simultaneously choose a location on a line, which represents the market. Consumers are uniformly distributed along this line and incur transportation costs when traveling to buy a product from one of the firms.
\newline
\newline
A key assumption in the classical model is that consumers will choose the nearest firm, given the cost of traveling to obtain the product. \citeA{kohlberg1983equilibrium} observed a concerning aspect of the original Hotelling model: the market shares of the firms are discontinuous in the locations, meaning that a firm could drastically increase or decrease its market share by making a small change in its chosen location. To address this issue, \citeA{kohlberg1983equilibrium} proposed a variation of the original model. In this model, consumers not only care about travel costs but also about the (average) waiting time for service from the firm, assumed to be an increasing function of the market share of the firm. He demonstrated that the inclusion of this extra cost for consumers turns the market shares of the firms into continuous functions of their locations. Lastly, he assumes that each firm is equally efficient at handling queues of consumers.
\newline
\newline
\citeA{kohlberg1983equilibrium} claimed that pure Nash equilibria only exist when the model consists of two firms. With two firms, the unique pure Nash equilibrium is where both firms locate at the center of the line. However, \citeA{peters2018hotelling} showed that this claim is not true and proved that under certain conditions, pure Nash equilibria also exist in Kohlberg's model for 4 firms and 6 firms. Nevertheless, \citeA{peters2018hotelling} verified that the claim holds for 3 and 5 firms. Furthermore, they relaxed the assumption of firms being equally efficient at dealing with queues of consumers. They showed that with two asymmetric firms in terms of their inefficiency level, no pure Nash equilibrium exists.
\newline
\newline
 This paper focuses on two variations of the Hotelling model with waiting costs. The first variation is of \citeA{peters2018hotelling}, where two asymmetric firms compete for consumers. The second is of \citeA{kohlberg1983equilibrium}, with three symmetric firms. Since no Nash equilibrium exists for these two settings, we propose to use a different solution concept called point rationalizability. \citeA{pearce1984rationalizable} and \citeA{bernheim1984rationalizable} criticized Nash equilibrium and independently developed rationalizability as an alternative solution concept. Their main critique of Nash equilibrium was that it can be too restrictive in its assumptions. This could partially explain why no solution is found in the Hotelling model with waiting costs, even for basic cases such as two asymmetric firms or three symmetric firms. Rationalizability requires that each player believes his opponents choose rationally, believes that his opponents believe that the other players choose rationally, and so on. Nash equilibrium additionally imposes a correct beliefs assumption, stating that each player must believe his opponents are correct about his own beliefs and that his opponents share his beliefs about other players \cite{aumann1995epistemic, perea2007one,spohn1982make}. However, rationalizability does not impose this constraint, making it a more natural and permissive concept compared to Nash equilibrium.
\newline
\newline
This paper characterizes the set of point rationalizable choices, as opposed to the set of rationalizable choices. The primary distinction between rationalizability and point rationalizability lies in the beliefs held by the firms. Point rationalizability solely considers point beliefs, where each belief assigns probability 1 to exactly one of the opponent's pure choices. In contrast, rationalizability encompasses probabilistic beliefs as well. In the context of variations of the Hotelling model, much emphasis has been placed on pure Nash equilibria, where probabilistic beliefs are not taken into account. Given the aim to compare our results with pure Nash equilibria or their absence, point rationalizability emerges as the preferred solution concept over rationalizability.
\newline
\newline
 The first main result of this paper is a characterization of the set of point rationalizable locations for two firms in the asymmetric Hotelling model with waiting costs. We show that this set only consists of two locations for both firms. These two locations are further from the center as their difference in inefficiency level increases, which implies that a greater difference in inefficiency levels between the firms can lead to a greater level of differentiation. A larger inefficiency gap incentivizes the less efficient firm to locate further away, targeting consumers who prefer it due to lower transportation costs, despite its inefficiency. The second main result is a characterization of the set of point rationalizable locations of each firm in the symmetric Hotelling model with three firms present. In this setting, the set of point rationalizable choices of a firm is an interval that is symmetric around the center of the line. Furthermore, the greater the inefficiency level of the firms, the smaller the interval becomes. Higher inefficiency levels lead consumers to care more about waiting costs than transportation costs, reducing the advantage of clustering near competitors and prompting firms to locate more centrally to capture a broader market.
\newline
\newline
The structure of the paper is as follows: Section 2 formally introduces the Hotelling model with waiting costs. Section 3 defines point rationalizability. The set of rationalizable locations for two firms at any inefficiency level, is covered in Section 4. Section 5 covers the case with three firms and symmetric inefficiency levels. Section 6 consists of some concluding remarks. All proofs are collected in the appendix.

\section{The Hotelling Model with Waiting Costs}
Consider a market with homogeneous products and regulated prices where $n \geq 2$ firms are competing for market shares on a line segment $[0,1]$. We assume that consumers are uniformly distributed on the line segment. Each firm $i \in N$, where $N = \{1,...,n\}$, simultaneously chooses a location $c_i \in C_i=[0,1]$. After the firms have chosen a location on the line, each consumer will buy one product, from one firm. If a consumer $x \in [0,1]$ buys from firm $i$, positioned at $c_i \in [0,1]$, who has a market share of $s_i(c_i,c_{-i})$, where $ c_{-i} = (c_j)_{j \neq i}$, then the total cost of this consumer is
$$\lvert c_i - x  \rvert + a_i \cdot s_i(c_i,c_{-i}).$$
Similar to the classical model of \citeA{hotbllino1929stability}, the first term corresponds to the distance that consumer $x$ has to travel along the line to arrive at the location $c_i$. The second term is due to \citeA{kohlberg1983equilibrium}, and corresponds to additional costs it has to pay, which depend on the market share of firm $i$. For example, the higher the market share of firm $i$, the greater the queue in front of the store, incurring waiting costs for consumer $x$. The coefficient $a_i >0$ corresponds to the inefficiency level of firm $i$. The lower the value of $a_i$, the more efficient firm $i$ is at handling queues of consumers and reducing waiting times. The sum of all markets shares should add up to one, so we have that
$$\sum_{i=1}^{n}s_i(c_i,c_{-i})=1.$$
A discerning reader may question what factors determine market shares initially. Market shares are determined by a choice function, which assigns a firm to each consumer on the line based on the locations chosen by all firms. This choice function should possess the property that, once assignments are made, no consumer has an incentive to deviate from their designated firm. \citeA{peters2018hotelling} showed the existence of a unique cost-minimizing choice function that assigns intervals of consumers to firms in a monotonic manner. To elaborate on this, consider a setting where the locations are chosen by the firms such that $c_1 \leq c_2 \leq ... \leq c_n$. In this context, there exist unique values $0 < x_1 < x_2 < ... < x_{n-1} < 1$ such that all consumers in $[0, x_1]$ buy from firm 1, all consumers in $(x_1, x_2]$ buy from firm 2, and so on until all consumers in $(x_{n-1}, 1]$ buy from firm $n$. This specific choice function aligns with the approach adopted in this paper and is also implied by \citeA{kohlberg1983equilibrium}. Let $x_0=0$ and $x_n=1$. To compute the values of $x_1$ until $x_{n-1}$, the following system of equations has to be solved, which has a unique solution. For $i=1,...n-1$,
\begin{equation}\label{BIGN}
\lvert c_i-x_i \rvert + a_i \cdot (x_i-x_{i-1})=\lvert c_{i+1}-x_i \rvert +a_{i+1} \cdot (x_{i+1}-x_i).
\end{equation}
For instance, with three firms and $c_1 \leq c_2 \leq c_3$, this cost-minimizing choice function assigns firm 1 to consumers in the interval $[0, x_1]$, firm 2 to consumers in the interval $(x_1,x_2]$, with the remaining interval of consumers $(x_2,1]$ assigned to firm 3. Then, we have that $s_1(c_1,c_2,c_3)=x_1$, $s_2(c_2,c_1,c_3)=x_2-x_1$ and $s_3(c_3,c_1,c_2)=1-x_2$. 

\section{Point Rationalizability}
In this section, we define the concept of point rationalizability \cite{bernheim1984rationalizable}. 
\newline
\newline 
Each firm $i \in N$ can motivate his choice (location) by forming a belief about his opponent's choice in $C_j$. We only consider point beliefs. Hence, firm $i$ believes with probability 1 that his opponents choose some choice combination $c_{-i} \in C_{-i}=\prod_{j \neq i} C_j$. Thus, in our model, a choice is optimal for a firm if it maximizes his market share for some point belief. This is equivalent to a choice being optimal for a firm if this choice maximizes his market share given that his opponents choose some choice combination $c_{-i}$. This leads to the following definition:
\begin{definition}
 A choice $c_i\in C_i$ is optimal for firm $i$ given a choice combination $c_{-i}$ if $\forall \ c_i^{*}\in C_i$, 

$$s_i(c_i,c_{-i}) \geq s_i(c_i^{*},c_{-i}).$$

If $c_i \in A_i \subseteq C_i$ and the above inequality holds for every $c_i^{*}\in A_i$, then $c_i$ is optimal in $A_i$ given $c_{-i}$.
\end{definition}
Next, we define an iterative procedure to characterize the set of point rationalizable choices for both firms. The following inductive procedure resembles the method introduced by \citeA{pearce1984rationalizable} but is adapted for point beliefs in pure choices.
\begin{definition} \label{set of point rationalizable choices}
Let $P_i(0)=C_i$ for all $i \in N$. Then $P_i(k)$ is inductively defined for $k=1,2,...$ by $P_i(k)=\{c_i \in P_i(k-1):\ \text{there exists}$ a $c_{-i}$ in $\prod_{j \neq i} P_{j}(k-1)$ such that $c_i$ is optimal in $P_i(k-1)$ given $c_{-i} \}$. The set of point rationalizable choices for firm $i$ is then $P_i={\bigcap}_{k=1}^\infty P_i(k)$. 
\end{definition}

\section{Point Rationalizable Locations with Two Firms}
In this section, we restrict our analysis to a setting with two firms. Given $c_1 \leq c_2$, equation (\ref{BIGN}) simplifies to 
$$\lvert c_1-x_1 \rvert + a_1 \cdot x_1=\lvert c_2-x_1 \rvert +a_2 \cdot (1-x_1).$$
This equation can be solved for $x_1$. Although $x_1$ depends on $c_1$ and $c_2$, we write $x_1$ instead of $x_1(c_1,c_2)$ for ease of notation. In this case, firm 1's market share is $x_1$ and firm 2's market share is $1-x_1$. Conversely, if $c_2 \leq c_1$, equation (\ref{BIGN}) simplifies to  
$$\lvert c_2-x_1 \rvert + a_2 \cdot x_1=\lvert c_1-x_1 \rvert +a_1 \cdot (1-x_1).$$
Here, firm 2's market share is $x_1$ and firm 1's market share is $1-x_1$. After defining all relevant concepts, we make use of a Lemma by \citeA{peters2018hotelling}, which helps us characterize the choices that are optimal for firm $i$ given certain point beliefs.
\begin{lemma}\label{marc1}
    If $c_i < c_j$ is optimal for firm $i$ given $c_j$, then $c_i=x_1$. Similarly, if $c_i>c_j$ is optimal for firm i given $c_j$, then $c_i=1-x_1$.
\end{lemma}
To provide intuition for this result, consider the following argument, similar to \citeA{kohlberg1983equilibrium}. Suppose that $c_1 < x_1$, where the latter represents firm 1's market share, as depicted in Figure 1a. Now, let firm 1 reposition slightly to the right, so $c_1' > c_1$, resulting in a new market share for firm 1, denoted as $x_1'$. We will now show that $x_1' > x_1$. By contradiction, assume that $x_1' \leq x_1$. Consider a consumer at $x_1'$. This consumer is indifferent between buying from firm 1 and firm 2, despite firm 1 being strictly closer to him than before and having a lower or equal market share than before, both of which reduce his costs. Hence, the only possible reason for this consumer to be indifferent between firm 1 and 2 is if the market share of firm 2 has strictly decreased, which is a contradiction, as we started by assuming that $x_1' \leq x_1$. Therefore, increasing $c_1$ as long as $c_1 < x_1$ leads to a greater market share for firm 1.
\newline
\newline
If $ x_1 < c_1 < c_2$, then all consumers between $0$ and $x_1$ will buy from firm 1, and all consumers between $x_1$ and $1$ will buy from firm 2. This scenario has been depicted in Figure 1b. Now let firm 1 locate slightly to the left, to $c_1' < c_1$, resulting in a new market share for firm 1, denoted as $x_1'$. We will now show that $x_1' > x_1$. By contradiction, assume that $x_1' \leq x_1$. Consider a consumer at $x_1'$. This consumer is indifferent between buying from firm 1 and firm 2, despite firm 1 being strictly closer to him than before and having a lower or equal market share than before. Therefore, the same argument as before implies that the market share of firm 2 has strictly decreased, which is a contradiction as we started by assuming that $x_1' \leq x_1$. Hence, decreasing $c_1$ as long as $c_1 > x_1$ leads to a greater market share for firm 1.
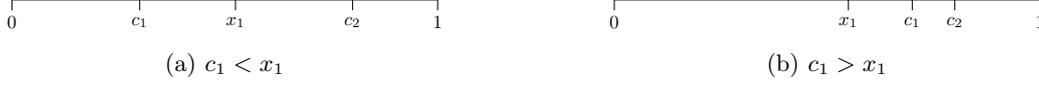
\begin{figure}
     \centering
     \begin{subfigure}[b]{0.48\textwidth}
        \centering
     \scalebox{0.7} {
    \begin{tikzpicture}
[scale=0.8]
   \draw[thick] (0,0) -- (10,0);
\draw (0,-0.25) node [anchor=north] {$0$};
\draw[black] (0,0) -- (0, -0.25);
             \draw[black] (10,0) -- (10, -0.25);
             \draw (10,-0.25) node [anchor=north] {$1$};
\draw[black] (3,0) -- (3, -0.25);
 \draw (3,-0.25) node [anchor=north] {$c_1$};
\draw[black] (8,0) -- (8, -0.25);
 \draw (8,-0.25) node [anchor=north] {$c_2$};
 \draw[black] (5.25,0) -- (5.25, -0.25);
 \draw (5.25,-0.25) node [anchor=north] {$x_1$};

\end{tikzpicture}
}
   \caption{$c_1 < x_1$}
    \label{fig:my_label}
     \end{subfigure}
     \hfill
     \begin{subfigure}[b]{0.48\textwidth}
         \centering
     \scalebox{0.7} {
   \begin{tikzpicture}
[scale=0.8]
    \draw[thick] (0,0) -- (10,0);
\draw (0,-0.25) node [anchor=north] {$0$};
\draw[black] (0,0) -- (0, -0.25);
             \draw[black] (10,0) -- (10, -0.25);
             \draw (10,-0.25) node [anchor=north] {$1$};

\draw[black] (7,0) -- (7, -0.25);
 \draw (7,-0.25) node [anchor=north] {$c_1$};
\draw[black] (8,0) -- (8, -0.25);
 \draw (8,-0.25) node [anchor=north] {$c_2$};
 \draw[black] (5.5,0) -- (5.5, -0.25);
 \draw (5.5,-0.25) node [anchor=north] {$x_1$};

\end{tikzpicture}
}
   \caption{$c_1 > x_1$}
    \label{fig:my_label}
     \end{subfigure}
        \caption{Intuition for Lemma 1}
        \label{fig:three graphs}
\end{figure}
\newline
\newline
For the analysis that follows, we assume that $a_1 < a_2$, meaning that firm 1 is more efficient than firm 2. Let us first consider cases where $c_1 < c_2$ is optimal for firm 1 given $c_2$. Then Lemma \ref{marc1} implies that $c_1 = x_1$ and equality (\ref{BIGN}) simplifies to
$$a_1 \cdot c_1 = c_2-c_1+a_2(1-c_1),$$
which yields $c_1= s_1(c_1,c_2) = \frac{c_2+a_2}{1+a_1+a_2}$. In this case, where $c_1 < c_2$, we need $c_1=\frac{c_2+a_2}{1+a_1+a_2} < c_2$, implying $c_2 > \frac{a_2}{a_1+a_2}$. Similarly, if $c_1 > c_2$ is optimal for firm 1 given $c_2$, we obtain $c_1= \frac{c_2+a_1}{1+a_1+a_2}$, $s_1(c_1,c_2)=1-c_1=\frac{1+a_2-c_2}{1+a_1+a_2}$, and $c_2 < \frac{a_1}{a_1+a_2}$.
\newline
\newline
Lastly, if $c_1=c_2$, then equality (\ref{BIGN}) simplifies to
$$\lvert c_1 - x_1  \rvert + a_1 \cdot x_1=\lvert c_1 - x_1  \rvert + a_2 \cdot (1-x_1),$$
and hence
$$ a_1 \cdot x_1= a_2 \cdot (1-x_1).$$ 
This leads to $s_1(c_1,c_2)=x_1=\frac{a_2}{a_1+a_2}$. As a result, no matter what location firm 2 chooses, firm 1 can obtain a market share of $\frac{a_2}{a_1+a_2}$ by choosing $c_1=c_2$. Firm 1's choice $c_1=c_2$ is optimal given $c_2$ if this gives him a market share that is at least as high as positioning to the left of $c_2$ or to the right of $c_2$. This implies that $\frac{a_2}{a_1+a_2} \geq \frac{c_2+a_2}{1+a_1+a_2}$ and $\frac{a_2}{a_1+a_2} \geq \frac{1+a_2-c_2}{1+a_1+a_2}$. These two inequalities hold if $c_2 \in [\frac{a_1}{a_1+a_2},\frac{a_2}{a_1+a_2}]$. Similarly, firm 1's choice $c_1=\frac{c_2+a_1}{1+a_1+a_2}$ is optimal for firm 1 given $c_2$ if $c_2\in[0,\frac{a_1}{a_1+a_2}]$, and $c_1=\frac{c_2+a_2}{1+a_1+a_2}$ is optimal for firm 1 given $c_2$ if $c_2\in [\frac{a_2}{a_1+a_2},1]$. This leads to the following increasing reaction function of firm 1, depicted in Figure 2:

\begin{figure}[t]
     \centering
     \scalebox{0.8} {
    \begin{tikzpicture}
[scale=0.8]
   \draw[thick] (0,0) -- (10,0);
\draw[thick] (0,0) -- (0,10);
\draw (0,0) node [anchor=east] {$0$};
\draw (-0.7,8.5) node [anchor=east] {$c_1$};

\draw (9,-0.7) node [anchor=north] {$c_2$};

\draw[ultra thick] (0,3.03) -- (4.348,4.438);
\draw[ultra thick]  (4.348,4.438) -- (5.652,5.652);
\draw[ultra thick] (5.652,5.652) -- (10,6.97);
\draw[dotted]  (4.348,4.438) --  (4.348,0);
\draw[dotted] (5.652,5.652) -- (5.652,0);
\draw[dotted]  (10,0) --  (10,6.97);

\draw[dotted] (5.652,5.652) -- (0,5.652);
\draw[dotted]  (4.348,4.438) --  (0,4.348);
\draw[dotted]  (0,6.97) --  (10,6.97);

\draw[black](0,3.03) -- (-0.25,3.03);
\draw (-0.25,3.03) node [anchor=east] {$\frac{a_1}{1+a_1+a_2}$};
\draw[black](0,4.348) -- (-0.25,4.348);
\draw (-0.25,4.348) node [anchor=east] {$\frac{a_1}{a_1+a_2}$};
\draw[black](4.348,0) -- (4.348,-0.25);
\draw (4.348,-0.25) node [anchor=north] {$\frac{a_1}{a_1+a_2}$};
\draw[black](0,5.652) -- (-0.25,5.652);
\draw (-0.25,5.652) node [anchor=east] {$\frac{a_2}{a_2+a_2}$};
\draw[black](5.652,0) -- (5.652,-0.25);
\draw (5.652,-0.25) node [anchor=north] {$\frac{a_2}{a_1+a_2}$};
\draw[black](0,6.97) -- (-0.25,6.97);
\draw (-0.25,6.97) node [anchor=east] {$\frac{1+a_2}{1+a_1+a_2}$};

           \draw[black] (0,0) -- (0, -0.25);
             \draw[black] (10,0) -- (10, -0.25);
             \draw (10,-0.25) node [anchor=north] {$1$};

             \draw[black] (0,10) -- (-0.25,10);
             \draw (-0.25,10) node [anchor=east] {$1$};

\end{tikzpicture}
}
   \caption{Reaction function of firm 1}
    \label{fig:my_label}
    
\end{figure}
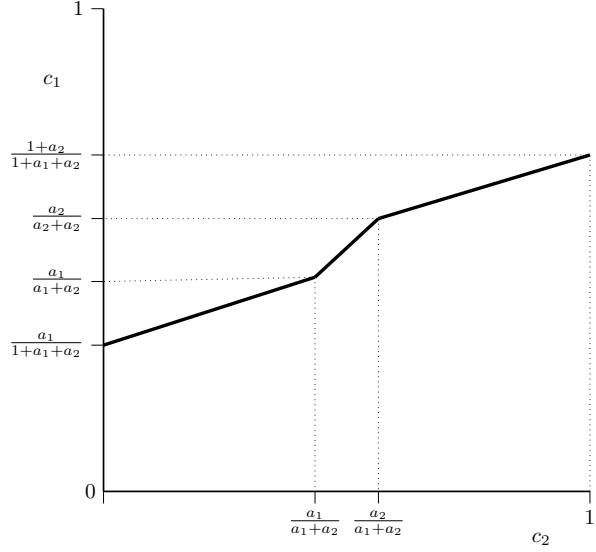

 \[   
c_1(c_2)=
     \begin{cases}
      \frac{c_2+a_1}{1+a_1+a_2}  & \text{if} \ c_2 \in [0,\frac{a_1}{a_1+a_2}) \\
     c_2  & \text{if} \  c_2 \in [\frac{a_1}{a_1+a_2},\frac{a_2}{a_1+a_2}]  \\
      \frac{c_2+a_2}{1+a_1+a_2}  & \text{if} \  c_2 \in (\frac{a_2}{a_1+a_2},1] \\
     \end{cases}
\]

Hence, the set of choices of firm 1 that survive round 1 of the iterative procedure are $P_1^1=[\frac{a_1}{1+a_1+a_2},\frac{1+a_2}{1+a_1+a_2}]$.
\begin{figure}[t]
     \centering
     \begin{subfigure}[b]{0.48\textwidth}
         \centering
     \scalebox{0.8} {
    \begin{tikzpicture}
[scale=0.8]
   \draw[thick] (0,0) -- (10,0);
\draw[thick] (0,0) -- (0,10);
\draw (0,0) node [anchor=east] {$0$};
\draw (-0.7,8.5) node [anchor=east] {$c_2$};

\draw (9,-0.7) node [anchor=north] {$c_1$};

\draw[ultra thick] (0,3.939) -- (5,5.455);
\draw[ultra thick]  (5,4.545) -- (10,6.061);
\draw[dotted]  (5,5.455) -- (5,0);
\draw[dotted]  (5,5.455) -- (0,5.455);
\draw[dotted]  (5,4.545) -- (0,4.545);
\draw[dotted]  (10,6.061) -- (0,6.061);

  \draw[black] (0,0) -- (0, -0.25);
             \draw[black] (10,0) -- (10, -0.25);
             \draw (10,-0.25) node [anchor=north] {$1$};
             \draw[black] (0,10) -- (-0.25,10);
             \draw (-0.25,10) node [anchor=east] {$1$};
             
              \draw[black] (0,3.939) -- (-0.25,3.939);
             \draw (-0.25,3.939) node [anchor=east] {$\frac{a_2}{1+a_1+a_2}$};
             \draw[black] (0,5.455) -- (-0.25,5.455);
             \draw (-0.25,5.455) node [anchor=east] {$\frac{\frac{1}{2}+a_2}{1+a_1+a_2}$};
              \draw[black] (0,4.545) -- (-0.25,4.545);
             \draw (-0.25,4.545) node [anchor=east] {$\frac{\frac{1}{2}+a_1}{1+a_1+a_2}$};
             \draw[black] (0,6.061) -- (-0.25,6.061);
             \draw (-0.25,6.061) node [anchor=east] {$\frac{1+a_1}{1+a_1+a_2}$};
              \draw[black] (5,0) -- (5,-0.25);
             \draw (5,-0.25) node [anchor=north] {$\frac{1}{2}$};

\end{tikzpicture}
}
   \caption{$a_2-a_1 \leq \frac{1}{2}$}
    \label{fig:my_label}
     \end{subfigure}
     \hfill
     \begin{subfigure}[b]{0.48\textwidth}
         \centering
     \scalebox{0.8} {
      \begin{tikzpicture}
[scale=0.8]
   \draw[thick] (0,0) -- (10,0);
\draw[thick] (0,0) -- (0,10);
\draw (0,0) node [anchor=east] {$0$};
\draw (-0.7,8.5) node [anchor=east] {$c_2$};

\draw (9,-0.7) node [anchor=north] {$c_1$};

\draw[ultra thick] (0,4.595) -- (5,5.946);
\draw[ultra thick]  (5,4.054) -- (10,5.405);
\draw[dotted]  (5,5.946) -- (5,0);
\draw[dotted]  (5,5.946) -- (0,5.946);
\draw[dotted]  (5,4.054) -- (0,4.054);
\draw[dotted]  (10,5.405) -- (0,5.405);

 \draw[black] (0,4.595) -- (-0.25,4.595);
             \draw (-0.25,4.595) node [anchor=east] {$\frac{a_2}{1+a_1+a_2}$};
             \draw[black] (0,5.946) -- (-0.25,5.946);
             \draw (-0.25,5.946) node [anchor=east] {$\frac{\frac{1}{2}+a_2}{1+a_1+a_2}$};
              \draw[black] (0,4.054) -- (-0.25,4.054);
             \draw (-0.25,4.054) node [anchor=east] {$\frac{\frac{1}{2}+a_1}{1+a_1+a_2}$};
             \draw[black] (0,5.405) -- (-0.25,5.405);
             \draw (-0.25,5.405) node [anchor=east] {$\frac{1+a_1}{1+a_1+a_2}$};

  \draw[black] (0,0) -- (0, -0.25);
            
             \draw[black] (10,0) -- (10, -0.25);
             \draw (10,-0.25) node [anchor=north] {$1$};
             \draw[black] (0,10) -- (-0.25,10);
             \draw (-0.25,10) node [anchor=east] {$1$};            
             \draw[black] (5,0) -- (5,-0.25);
             \draw (5,-0.25) node [anchor=north] {$\frac{1}{2}$};

\end{tikzpicture}
    
}
   \caption{$\frac{1}{2} \leq  a_2 -a_1 \leq 1$}
    \label{fig:my_label}
     \end{subfigure}
        \caption{Reaction correspondence of firm 2}
        \label{fig:three graphs}
\end{figure}
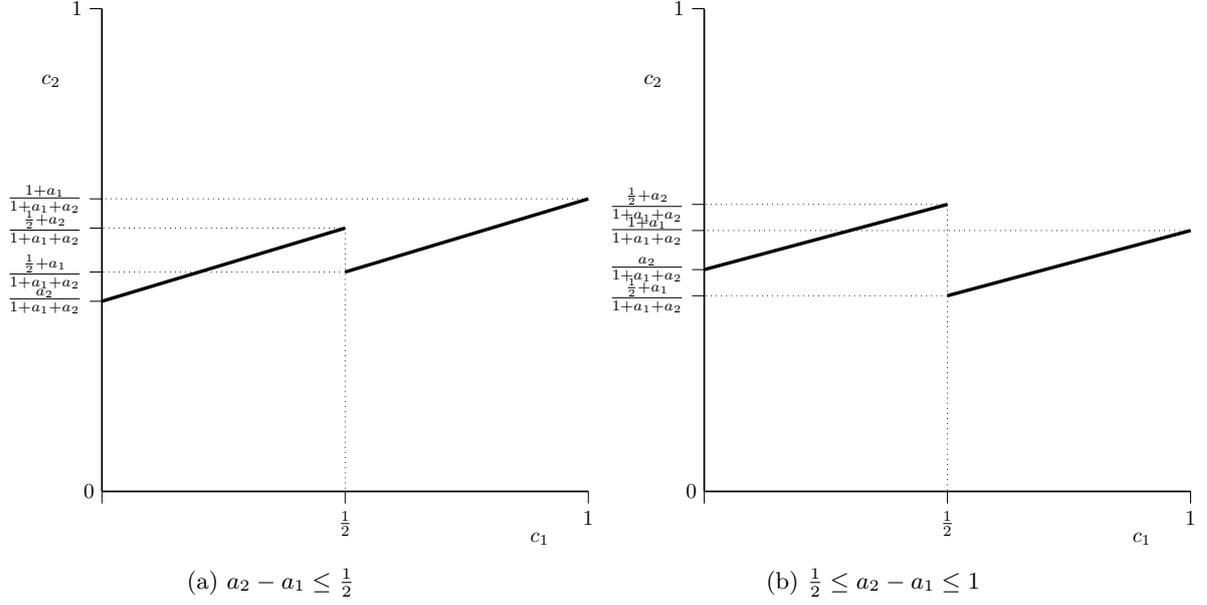
\newline
\newline
Now, consider firm 2. If $c_2 < c_1$ is optimal for firm 2 given $c_1$, then Lemma \ref{marc1} implies that $c_2 = x_1$, and equality (\ref{BIGN}) simplifies to 
$$a_2 \cdot c_2= c_1 - c_2 + a_1 \cdot (1-c_2)$$
Solving for $c_2$ leads to $c_2 = s_2(c_1,c_2) = \frac{c_1 + a_1}{1 + a_1 + a_2}$. In this case, where $c_2 < c_1$, we also have $\frac{c_1 + a_1}{1 + a_1 + a_2} < c_1$, implying $c_1 > \frac{a_1}{a_1 + a_2}$. Similarly, if $c_2 > c_1$ is optimal for firm 2, then we obtain $c_2 = \frac{c_1 + a_2}{1 + a_1 + a_2}$, $s_2(c_1,c_2) = 1 - c_2 = \frac{1 + a_1 - c_1}{1 + a_1 + a_2}$, and $c_1 < \frac{a_2}{a_1 + a_2}$.
\newline
\newline
Next, note that, similarly to the analysis for firm 1, if $c_2=c_1$, then $s_2(c_1,c_2)=\frac{a_1}{a_1+a_2}$. However, choosing $c_2=c_1$ is never an optimal choice for firm 2 given $c_1$. In other words, choosing either some $c_2 < c_1$ or $c_2 > c_1$ always yields strictly greater market share. Specifically, if $c_1 \leq \frac{1}{2}$, then choosing $c_2=\frac{c_1+a_2}{1+a_1+a_2}$ yields a strictly greater market share than choosing $c_2=c_1$. Conversely, if $c_1 \geq \frac{1}{2}$, then choosing $c_2=\frac{c_1+a_1}{1+a_1+a_2}$ yields a strictly greater market share than choosing $c_2=c_1$.
\newline
\newline
As $c_2=c_1$ is never optimal for firm 2 given $c_1$, it will either choose to position to the left of firm 1, which leads to a market share of $s_2(c_1,c_2)=c_2=x_1=\frac{c_1+a_2}{1+a_1+a_2}$, or to the right of firm 1, which leads to a market share of $s_2(c_1,c_2)=\frac{1+a_1-c_1}{1+a_1+a_2}$. We can observe that choosing $c_2 = \frac{c_1+a_1}{1+a_1+a_2}$, which is to the left of firm 1, is optimal for firm 2 given $c_1$ if $\frac{c_1+a_1}{1+a_1+a_2} \geq \frac{1+a_1-c_1}{a+a_1+a_2}$, implying $c_1 \geq \frac{1}{2}$. Similarly, $c_2 = \frac{c_1+a_2}{1+a_1+a_2}$ is optimal for firm 2 given $c_1$ if $c_1 \leq \frac{1}{2}$. This leads to the following reaction correspondence depicted below.

\[   
c_2(c_1)=
     \begin{cases}
      \frac{c_1+a_2}{1+a_1+a_2}  & \text{if} \ c_1 \in [0,\frac{1}{2}) \\
       \frac{\frac{1}{2}+a_2}{1+a_1+a_2} \ \text{or} \ \frac{\frac{1}{2}+a_1}{1+a_1+a_2}  & \text{if} \ c_1 =\frac{1}{2} \\
     \frac{c_1+a_1}{1+a_1+a_2}  & \text{if} \ c_1 \in (\frac{1}{2},1] \\
    
     \end{cases}
\]
If the difference in inefficiency levels is small $(a_2 -a_1 \leq \frac{1}{2})$, then the choices of firm 2 that survive round 1 of the iterative procedure are in $P_2^1=[\frac{a_2}{1+a_1+a_2},\frac{1+a_1}{1+a_1+a_2}]$. For intermediate levels $(\frac{1}{2} \leq  a_2 -a_1 \leq 1)$, it is $P_2^1=[\frac{\frac{1}{2} + a_1}{1+a_1+a_2},\frac{\frac{1}{2} + a_2}{1+a_1+a_2}]$, and for large levels $(a_2 - a_1 > 1)$, it is $P_2^1=[\frac{\frac{1}{2} + a_1}{1+a_1+a_2},\frac{1+a_1}{1+a_1+a_2}] \cup [\frac{a_2}{1+a_1+a_2},\frac{\frac{1}{2} + a_2}{1+a_1+a_2}]$. Figures 3a, 3b, and 4a depict this.
\newline
\newline
It has been shown by \citeA{peters2018hotelling} that if $a_1 \neq a_2$, then no pure Nash equilibrium exists. Visually, this is also an immediate consequence when we look at the reaction function of firm 1 and the reaction correspondence of firm 2 together in one Figure. Figure 4b shows that there is no point where these intersect.
\newline
\newline
\begin{figure}
     \centering
     \begin{subfigure}[b]{0.48\textwidth}
        \centering
     \scalebox{0.7} {
      \begin{tikzpicture}
[scale=0.8]
   \draw[thick] (0,0) -- (10,0);
\draw[thick] (0,0) -- (0,10);
\draw (0,0) node [anchor=east] {$0$};
\draw (-0.7,8.5) node [anchor=east] {$c_2$};

\draw (9,-0.7) node [anchor=north] {$c_1$};

\draw[ultra thick] (0,6) -- (5,7);
\draw[ultra thick]  (5,3) -- (10,4);
\draw[dotted]  (5,7) -- (5,0);
\draw[dotted]  (5,7) -- (0,7);
\draw[dotted]  (5,3) -- (0,3);
\draw[dotted]  (10,4) -- (0,4);

  \draw[black] (0,0) -- (0, -0.25);
            
             \draw[black] (10,0) -- (10, -0.25);
             \draw (10,-0.25) node [anchor=north] {$1$};
             \draw[black] (0,10) -- (-0.25,10);
             \draw (-0.25,10) node [anchor=east] {$1$};

              \draw[black] (0,6) -- (-0.25,6);
             \draw (-0.25,6) node [anchor=east] {$\frac{a_2}{1+a_1+a_2}$};
             \draw[black] (0,7) -- (-0.25,7);
             \draw (-0.25,7) node [anchor=east] {$\frac{\frac{1}{2}+a_2}{1+a_1+a_2}$};
              \draw[black] (0,3) -- (-0.25,3);
             \draw (-0.25,3) node [anchor=east] {$\frac{\frac{1}{2}+a_1}{1+a_1+a_2}$};
             \draw[black] (0,4) -- (-0.25,4);
             \draw (-0.25,4) node [anchor=east] {$\frac{1+a_1}{1+a_1+a_2}$};
             \draw[black] (0,5) -- (-0.25,5);
             \draw (-0.25,5) node [anchor=east] {$\frac{1}{2}$};
             \draw[black] (5,0) -- (5,-0.25);
             \draw (5,-0.25) node [anchor=north] {$\frac{1}{2}$};

\end{tikzpicture}
}
   \caption{$ a_2 - a_1 > 1$}
    \label{fig:my_label}
     \end{subfigure}
     \hfill
     \begin{subfigure}[b]{0.48\textwidth}
         \centering
     \scalebox{0.7} {
   \begin{tikzpicture}
[scale=0.8]
   \draw[thick] (0,0) -- (10,0);
\draw[thick] (0,0) -- (0,10);
\draw (0,0) node [anchor=east] {$0$};
\draw (-0.7,8.5) node [anchor=east] {$c_2$};

\draw (9,-0.7) node [anchor=north] {$c_1$};

\draw[ultra thick] (0,3.939) -- (5,5.455);
\draw[ultra thick]  (5,4.545) -- (10,6.061);

\draw[ultra thick]  (3.03,0) -- (4.348,4.348);
\draw[ultra thick]  (4.348,4.348) -- (5.652,5.652);
\draw[ultra thick]  (5.652,5.652) -- (6.97,10);

           \draw[black] (0,0) -- (0, -0.25);
             \draw[black] (10,0) -- (10, -0.25);
             \draw (10,-0.25) node [anchor=north] {$1$};

             \draw[black] (0,10) -- (-0.25,10);
             \draw (-0.25,10) node [anchor=east] {$1$};

\end{tikzpicture}
}
   \caption{No Nash equilibrium}
    \label{fig:my_label}
     \end{subfigure}
        \caption{Reaction correspondence of firm 2}
        \label{fig:three graphs}
\end{figure}
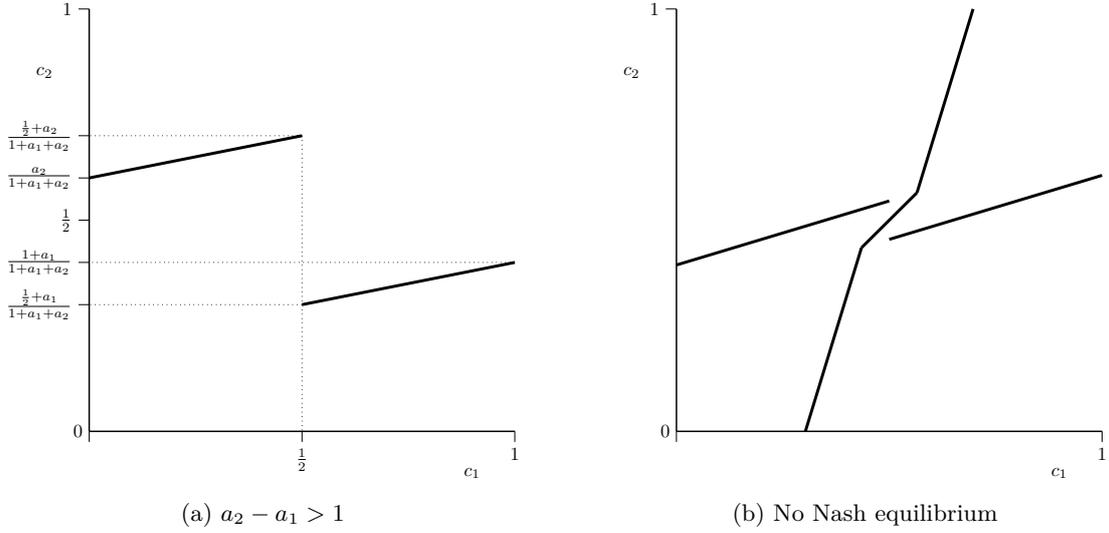
The next theorem characterizes the set of point rationalizable choices for both firms. 
\begin{theorem}
Let there be two firms, where $a_1 < a_2$. The set of point rationalizable choices of both firms is given by $P_1=P_2=\{\frac{a_1+1}{a_1+a_2+2},\frac{a_2+1}{a_1+a_2+2}\}.$
\end{theorem}
The set of point rationalizable choices consists of only two locations. This suggests a high predictive power for this solution concept. If the set of point rationalizable choices only had one option for each firm, then it would follow that these choices together would form a Nash equilibrium. From \citeA{peters2018hotelling}, we know that this cannot be the case. Hence, in some sense, this is the most predictive solution we could have hoped for.
\newline
\newline
We observe that as the inefficiency levels of the firms become more similar, the two point rationalizable choices are closer to the center of the line. Conversely, the greater the difference in inefficiency levels of the firms, the further away the two point rationalizable choices are from the center of the line. The greater the difference in inefficiency levels between the firms, the more incentive the less efficient firm has to choose a location further from the more efficient firm. This is because a larger inefficiency gap makes consumers increasingly prefer the more efficient firm. By increasing the distance from the more efficient firm, the less efficient firm can target a segment of consumers who will choose it due to lower transportation costs, despite its greater inefficiency.
\newline
\newline
The outcome of the game depends on which of the two locations both firms choose if they opt for rationalizable locations. Note that firm $1$'s most favorable outcome, when he chooses $c_1= \frac{a_1+1}{a_1+a_2+2}$, is that firm 2 also chooses this location. His least favorable outcome is that firm 2 then chooses $c_2= \frac{a_2+1}{a_1+a_2+2}$. In other words, given that both firms choose a point rationalizable location, firm 1 prefers the minimum differentiation outcome, and firm 2 prefers the maximum differentiation outcome.
\newline
\newline
The following example is added to provide some intuition about the first main result.

\begin{example}
Consider firm 1 and firm 2, where $a_1=1$ and $a_2=3$. The difference between the inefficiency levels is large ($a_2 > a_1+1$), and the reaction function/correspondence of both firms then imply that $P_1^1=[\frac{1}{5},\frac{4}{5}]$ and $P_2^1=[\frac{3}{10},\frac{4}{10}] \cup [\frac{6}{10},\frac{7}{10}]$. Figure 5a visually represents this. Next, from the reaction function of firm $1$, it can be seen that he will want to choose $c_1=c_2$ whenever $c_2 \in [\frac{a_1}{a_1+a_2},\frac{a_2}{a_1+a_2}]=[\frac{1}{4},\frac{3}{4}]$. As $P_2^1 \subseteq [\frac{1}{4},\frac{3}{4}]$, this implies that $P_1^2= P_2^1$. Similarly for firm 2, applying the reaction correspondence implies that $P_2^2=[\frac{30}{100},\frac{36}{100}] \cup [\frac{64}{100},\frac{70}{100}]$. This has been depicted by figure 5b. By the same reasoning as before, we again have that $P_1^3=P_2^2$.
\newline
\newline
Note that to calculate $P_2^3$, we are interested in $P_1^2=P_2^1$. As $P_2^1=[\frac{3}{10},\frac{4}{10}] \cup [\frac{6}{10},\frac{7}{10}]$, applying the reaction correspondence of firm 2 leads to $P_2^3=[\frac{32}{100},\frac{34}{100}] \cup [\frac{66}{100},\frac{68}{100}]$. Hence, for round $k \geq 3$, we observe an interesting pattern: $P_1^k=P_2^{k-1}$ and $P_2^k$ can be calculated based on $P_1^{k-1}=P_2^{k-2}$. Continuing in this fashion eventually leads to $P_1=P_2=\{{\frac{1}{3},\frac{2}{3}}\}$.    
\end{example}

\begin{figure}
     \centering
     \begin{subfigure}[b]{0.48\textwidth}
        \centering
     \scalebox{0.7} {
      \begin{tikzpicture}
[scale=0.8]
  \draw[thick] (0,0) -- (10,0);
\draw[thick] (0,0) -- (0,10);
\draw (0,0) node [anchor=east] {$0$};
\draw (-0.7,8.5) node [anchor=east] {$c_2$};

\draw (9,-0.7) node [anchor=north] {$c_1$};

\draw[dotted]  (2,0) -- (2,10);
\draw[dotted]  (8,0) -- (8,10);
 
             \draw (5,-0.3) node [anchor=north] {$P_1^1$};
\draw [<->] (2,-0.2) -- (8,-0.2);
\draw (2,-0.3) node [anchor=north] {$0.2$};
\draw (8,-0.3) node [anchor=north] {$0.8$};

\draw[ultra thick] (0,6) -- (5,7);
\draw[ultra thick]  (5,3) -- (10,4);

\draw[ultra thick]  (2,0) -- (2.5,2.5);
\draw[ultra thick]  (2.5,2.5) -- (7.5,7.5);
\draw[ultra thick]  (7.5,7.5) -- (8,10);

           \draw[black] (0,0) -- (0, -0.25);
             \draw[black] (10,0) -- (10, -0.25);
             \draw (10,-0.25) node [anchor=north] {$1$};

             \draw[black] (0,10) -- (-0.25,10);
             \draw (-0.25,10) node [anchor=east] {$1$};
\draw[dotted]  (0,3)  -- (10,3);
\draw[dotted]  (0,4) -- (10,4);
\draw[dotted]  (0,6)  -- (10,6);
\draw[dotted]  (0,7) -- (10,7);
 
             \draw (-0.2,5) node [anchor=east] {$P_2^1$};
\draw [<->] (-0.2,3) -- (-0.2,4);
\draw [<->] (-0.2,6) -- (-0.2,7);
\draw (-0.2,3) node [anchor=east] {$0.3$};
\draw (-0.2,4) node [anchor=east] {$0.4$};
\draw (-0.2,6) node [anchor=east] {$0.6$};
\draw (-0.2,7) node [anchor=east] {$0.7$};

\end{tikzpicture}
}
   \caption{Round 1}
    \label{fig:my_label}
     \end{subfigure}
     \hfill
     \begin{subfigure}[b]{0.48\textwidth}
         \centering
     \scalebox{0.7} {
   \begin{tikzpicture}
[scale=0.8]
  \draw[thick] (0,0) -- (10,0);
\draw[thick] (0,0) -- (0,10);
\draw (0,0) node [anchor=east] {$0$};
\draw (-0.7,8.5) node [anchor=east] {$c_2$};

\draw (9,-0.7) node [anchor=north] {$c_1$};

\draw[dotted]  (3,0) -- (3,10);
\draw[dotted]  (4,0) -- (4,10);
\draw[dotted]  (6,0) -- (6,10);
\draw[dotted]  (7,0) -- (7,10);
 
             \draw (5,-0.3) node [anchor=north] {$P_1^2$};
\draw [<->] (3,-0.25) -- (4,-0.25);
\draw [<->] (6,-0.25) -- (7,-0.25);
\draw (3,-0.3) node [anchor=north] {$0.3$};
\draw (4,-0.3) node [anchor=north] {$0.4$};
\draw (6,-0.3) node [anchor=north] {$0.6$};
\draw (7,-0.3) node [anchor=north] {$0.7$};

\draw[ultra thick] (2,6.4) -- (5,7);
\draw[ultra thick]  (5,3) -- (8,3.6);

\draw[ultra thick]  (3,3) -- (4,4);
\draw[ultra thick]  (6,6) -- (7,7);

           \draw[black] (0,0) -- (0, -0.25);
             \draw[black] (10,0) -- (10, -0.25);
             \draw (10,-0.25) node [anchor=north] {$1$};

             \draw[black] (0,10) -- (-0.25,10);
             \draw (-0.25,10) node [anchor=east] {$1$};

\draw[dotted]  (0,3)  -- (10,3);
\draw[dotted]  (0,3.6) -- (10,3.6);
\draw[dotted]  (0,6.4)  -- (10,6.4);
\draw[dotted]  (0,7) -- (10,7);
 
             \draw (-0.2,5) node [anchor=east] {$P_2^2$};
\draw [<->] (-0.2,3) -- (-0.2,3.6);
\draw [<->] (-0.2,6.4) -- (-0.2,7);
\draw (-0.2,3) node [anchor=east] {$0.3$};
\draw (-0.2,3.6) node [anchor=east] {$0.36$};
\draw (-0.2,6.4) node [anchor=east] {$0.64$};
\draw (-0.2,7) node [anchor=east] {$0.7$};

\end{tikzpicture}
}
   \caption{Round 2}
    \label{fig:my_label}
     \end{subfigure}
        \caption{Example of iterative procedure where $a_1=1$ and $a_2=3$}
        \label{fig:three graphs}
\end{figure}

If firm 1 and 2 have two identical inefficiency coefficients, meaning that $a_1=a_2=a$, then the analysis simplifies, and both firms have the following reaction function:
$$c_i(c_j)= \frac{c_j+a}{1+2a} \ \text{for all} \ c_j \in [0,1]$$
This leads us to the following result in the case of the firms having the same ineffiency level. 
\begin{proposition}
Let there be two firms, and let $a_1=a_2$. Then the set of point rationalizable choices is $P_1=P_2=\{\frac{1}{2}\}$.    
\end{proposition}
It has been shown by \citeA{peters2018hotelling} and \citeA{kohlberg1983equilibrium} that there exists a unique pure Nash equilibrium where both firms choose to locate at the center of the line.

\section{Three Symmetric Firms}
Consider three symmetric firms with the same inefficiency coefficient. That is, we assume that $a_1=a_2=a_3=a$. In this section, we take the perspective of a firm $i \in N$, who considers his opponent firms $l$ and $r$. As $a_l=a_r$, we observe that if a choice $c_i$ is optimal for some belief $(c_l,c_r)$, then it is also optimal for the belief $(c_r,c_l)$. Hence, in each round of the iterative procedure, it is sufficient to only consider beliefs $(c_l,c_r)$ that satisfy $c_l \leq c_r$.
\newline
\newline
If firm $i$ chooses $c_i \leq c_l$, then the system of equations (\ref{BIGN}) can be simplified to
$$\lvert c_i -x_1 \rvert +a x_1 = \lvert c_l-x_1 \rvert +a(x_2-x_1) \  \text{and} \ \lvert c_l -x_2 \rvert +a(x_2-x_1) = \lvert c_r-x_2 \rvert +a(1-x_2)$$
and firm $i$'s market share is $x_1$. If firm $i$ chooses $c_l \leq c_i \leq c_r$, then the system of equations (\ref{BIGN}) can be simplified to
$$\lvert c_l -x_1 \rvert +a x_1 = \lvert c_i-x_1 \rvert +a(x_2-x_1) \  \text{and} \ \lvert c_i -x_2 \rvert +a(x_2-x_1) = \lvert c_r-x_2 \rvert +a(1-x_2)$$
and firm $i$'s market share is $x_2-x_1$. Lastly, if firm $i$ chooses $c_i \geq c_r$, then the system of equations (\ref{BIGN}) can be simplified to
$$\lvert c_l -x_1 \rvert +a x_1 = \lvert c_r-x_1 \rvert +a(x_2-x_1) \  \text{and} \ \lvert c_r -x_2 \rvert +a(x_2-x_1) = \lvert c_i-x_2 \rvert +a(1-x_2)$$
and firm $i$'s market share is $1-x_2$. Furthermore, \citeA{kohlberg1983equilibrium} has proven that $x_1$ and $x_2$ are continuous in $c_i$. 
\newline
\newline
Figure 6 is a visual representation of the reaction correspondence of a firm $i$, where firm $i$ believes that $1-c_r \leq c_l \leq c_r$. In Lemma \ref{29juli3}, we show that considering this set of beliefs is sufficient to characterize the set of point rationalizable choices of a firm. In region 1, both firms $l$ and $r$ select locations close to the right edge of the line. Firm $i$'s optimal choice is then to choose $c_i$ such that $c_i = x_1$. In fact, firms $l$ and $r$ are so close to the right edge of the line that we have $x_2 \leq c_l$. Similarly, in regions 2 and 3, the optimal choice $c_i$ of firm $i$ satisfies $c_i = x_1$, where $c_l \leq x_2 \leq c_r$ in region 2 and $x_2 \geq c_r$ in region 3. Finally, in region 4, firm $l$ selects a location close to the left side of the line, whereas firm $r$ selects a location close to the right side of the line. As a result, firm $i$'s optimal choice is to position itself between these two firms, where any choice in the interval $[c_i^*, c_i^{**}]$ is optimal, where $c_i^*$ satisfies $c_i^*=x_1$ and $c_i^{**}$ satisfies $c_i^{**}=x_2$.  \newline

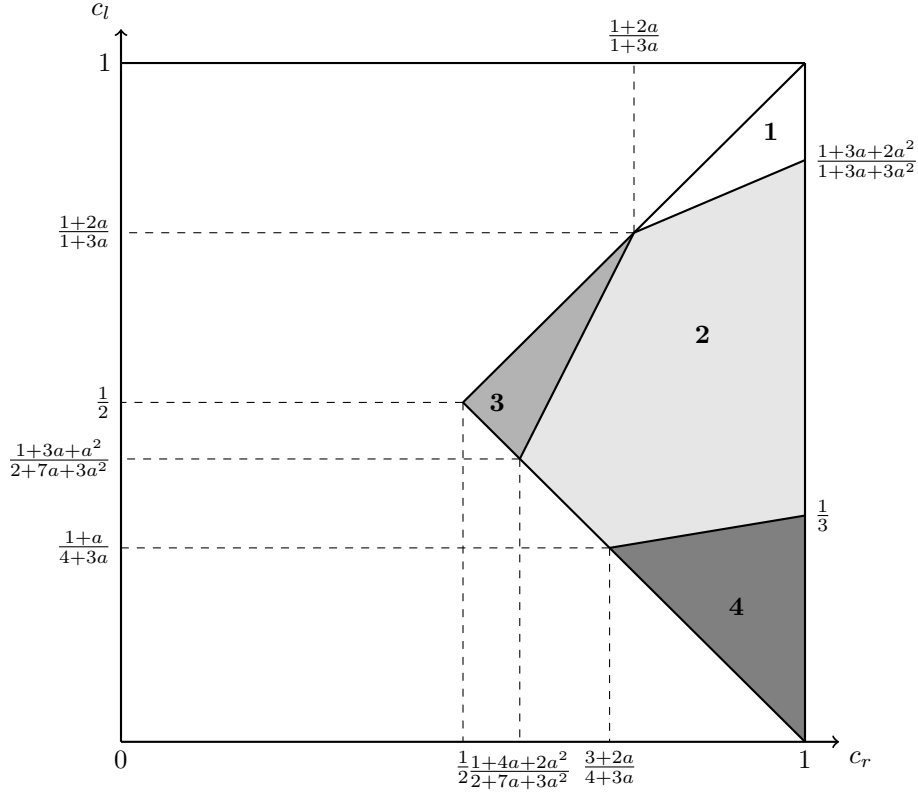
\begin{figure}[t]
    \centering
        \begin{tikzpicture} [scale=0.9]
        
        \fill[black!0!white]  (10,10) -- (10,8.57) -- (7.5,7.5);
        \fill[black!10!white]  (7.5,7.5) -- (10,8.57) -- (10,3.33333)--(7.143,2.857)--(5.83,4.167);
        \fill[black!30!white]  (5,5) -- (7.5,7.5) -- (5.83,4.167);
        \fill[black!50!white]  (7.143,2.857) -- (10,3.333333) -- (10,0);
        
        \draw[thick,->] (0,0) -- (10.5,0) node[anchor=north west] {$c_r$};
        \draw[thick,->] (0,0) -- (0,10.5) node[anchor= south east] {$c_l$};
        \draw[thick] (10,0) -- (5,5);
        \draw[thick] (5,5) -- (10,10);
        \draw[thick] (10,0) -- (10,10);
        \draw[thick] (0,10) -- (10,10);
        \draw[dashed] (5,5) -- (5,0);
        \draw[dashed] (5,5) -- (0,5);

        \draw[thick] (7.5,7.5) -- (10,8.57);
        \draw[thick] (5.83,4.167) -- (7.5,7.5);
        \draw[thick] (7.143,2.857) -- (10,3.333333);

        \node at (0,0) [anchor=north] {0};
        \node at (10,0) [anchor=north] {1};
        \node at (0,10) [anchor=east] {1};
        \node at (5,0) [anchor=north] {$\frac{1}{2}$};
        \node at (0,5) [anchor=east] {$\frac{1}{2}$};

        \node at (9.5,9) {\textbf{1}};
        \node at (8.5,6) {\textbf{2}};
        \node at (5.5,5) {\textbf{3}};
        \node at (9,2) {\textbf{4}};

        \draw[dashed] (7.5,7.5) -- (7.5,10) node[anchor=south] {$\frac{1+2a}{1+3a}$};
        \draw[dashed] (7.5,7.5) -- (0,7.5) node[anchor=east] {$\frac{1+2a}{1+3a}$};
        \draw[dashed] (5.83,4.167) -- (5.83,0) node[anchor=north] {$\frac{1+4a+2a^2}{2+7a+3a^2}$};
         \draw[dashed] (5.83,4.167) -- (0,4.167) node[anchor=east] {$\frac{1+3a+a^2}{2+7a+3a^2}$};
         \draw[dashed] (7.143,2.857) -- (7.143,0) node[anchor=north] {$\frac{3+2a}{4+3a}$};
         \draw[dashed] (7.143,2.857) -- (0,2.857) node[anchor=east] {$\frac{1+a}{4+3a}$};
          \node at (10,8.57) [anchor=west] {$\frac{1+3a+2a^2}{1+3a+3a^2}$};
           \node at (10,3.33333) [anchor=west] {$\frac{1}{3}$};

        \end{tikzpicture}
\caption{reaction correspondence of firm $i$}
\label{a=1}
\begin{tabular}{r@{: }l r@{: }l}
\textbf{1} & $c_i=\frac{c_l+c_r+a}{2+3a}$ & \textbf{2} & $  c_i=\frac{3ac_l+2c_l+ac_r+a^2}{2+6a+3a^2}$\\
\textbf{3}& $c_i= \frac{a+3c_l-c_r}{2+3a}$ & \textbf{4}& $c_i \in [\frac{2c_l+2ac_l+ac_r+a^2}{2+5a+3a^2},\frac{ac_l+2ac_r+2c_r+2a+2a^2}{2+5a+3a^2}]$
\end{tabular}
\label{legend a=1}
\end{figure}
After characterizing the reaction correspondence of firm $i$, we can apply the iterative procedure to compute the set of point rationalizable choices for a firm. This leads us to the second main result of this paper.
\begin{theorem}
Let there be three firms, where $a_1=a_2=a_3=a$. The set of point rationalizable choices of each firm $i \in N$ is $P_i=[\frac{1+a}{4+3a},\frac{3+2a}{4+3a}]$.
\end{theorem}
The set found in Theorem 2 is decreasing in $a$. As $a$ approaches $0$, which implies very low waiting costs for the consumers, the set of point rationalizable choices approaches $[\frac{1}{4},\frac{3}{4}]$. This is in line with \citeA{van2023rationalizable}, where it was shown that the set of point rationalizable choices in the classical Hotelling model without waiting costs is $[\frac{1}{4},\frac{3}{4}]$. On the other hand, as $a$ approaches infinity, which implies extremely high waiting costs for the consumers, the set of point rationalizable choices approaches $[\frac{1}{3},\frac{2}{3}]$. Hence, a greater level of inefficiency of the firms reduces the potential for product differentiation. As the inefficiency levels of all firms increase, consumers prioritize waiting costs over transportation costs. This reduces the benefit of stealing customers by locating close to competitors. Consequently, firms are incentivized to locate more centrally to capture a broader market.
\newline
\newline
The next example illustrates how the iterative procedure is applied for three firms.

\begin{example}
Consider three firms, where $a_1 = a_2 = a_3 = 1$, and take the perspective of firm $i$. For every possible $(c_l, c_r)$ such that $1 - c_r \leq c_l \leq c_r$, we can compute the optimal choice or choices of firm $i$ by using the reaction correspondence. In the proof of Theorem 1, we show that the greatest choice that is optimal for a belief is obtained for the belief that selects the maximum values of $c_r$ and $c_l$ such that positioning between $c_l$ and $c_r$ is optimal for firm $i$. The reaction correspondence shows that this corresponds to $c_r = 1$ and $c_l = \frac{1}{3}$. This has been depicted in Figure 7(a). Any choice in $[\frac{1}{3}, \frac{5}{6}]$ is optimal for this belief. By symmetry, the choices in $[\frac{1}{6}, \frac{2}{3}]$ are then optimal for the belief where $c_l=0$ and $c_r=\frac{2}{3}$. Hence, we have that $P_i(1) = [\frac{1}{6}, \frac{5}{6}]$, and similarly for the other two firms. 
\newline
\newline
In the next round, we can again find the greatest choice that is optimal for a belief that survived round 1 of the procedure by considering the maximum values of $c_r$ and $c_l$ such that positioning between $c_l$ and $c_r$ is optimal. This corresponds to $c_r = \frac{5}{6}$ and $c_l = \frac{11}{36}$. This has been depicted in Figure 7(b). The reaction correspondence shows that any choice in $[\frac{11}{36}, \frac{55}{72}]$ is optimal for this belief. Using a similar symmetry argument as above, this implies that $P_i(2) = [\frac{17}{72}, \frac{55}{72}]$ for all $i \in N$. Continuing in this fashion leads to $P_1 = P_2 = P_3 = [\frac{2}{7}, \frac{5}{7}]$.
\end{example}

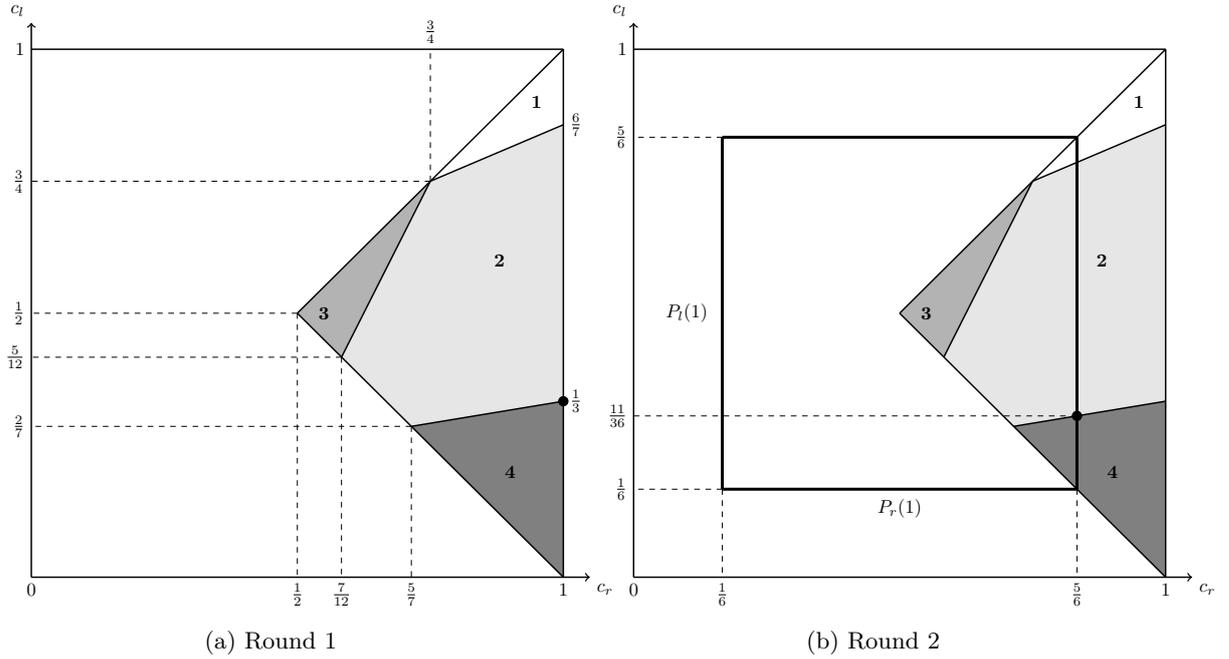
\begin{figure}
     \centering
     \begin{subfigure}[b]{0.48\textwidth}
        \centering
     \scalebox{0.7} {
      \begin{tikzpicture} [scale=1]
        
        \fill[black!0!white]  (10,10) -- (10,8.57) -- (7.5,7.5);
        \fill[black!10!white]  (7.5,7.5) -- (10,8.57) -- (10,3.33333)--(7.143,2.857)--(5.83,4.167);
        \fill[black!30!white]  (5,5) -- (7.5,7.5) -- (5.83,4.167);
        \fill[black!50!white]  (7.143,2.857) -- (10,3.333333) -- (10,0);
        
        \draw[thick,->] (0,0) -- (10.5,0) node[anchor=north west] {$c_r$};
        \draw[thick,->] (0,0) -- (0,10.5) node[anchor= south east] {$c_l$};
        \draw[thick] (10,0) -- (5,5);
        \draw[thick] (5,5) -- (10,10);
        \draw[thick] (10,0) -- (10,10);
        \draw[thick] (0,10) -- (10,10);
        \draw[dashed] (5,5) -- (5,0);
        \draw[dashed] (5,5) -- (0,5);

        \draw[thick] (7.5,7.5) -- (10,8.57);
        \draw[thick] (5.83,4.167) -- (7.5,7.5);
        \draw[thick] (7.143,2.857) -- (10,3.333333);

        \node at (0,0) [anchor=north] {0};
        \node at (10,0) [anchor=north] {1};
        \node at (0,10) [anchor=east] {1};
        \node at (5,0) [anchor=north] {$\frac{1}{2}$};
        \node at (0,5) [anchor=east] {$\frac{1}{2}$};

        \node at (9.5,9) {\textbf{1}};
        \node at (8.8,6) {\textbf{2}};
        \node at (5.5,5) {\textbf{3}};
        \node at (9,2) {\textbf{4}};

        \draw[dashed] (7.5,7.5) -- (7.5,10) node[anchor=south] {$\frac{3}{4}$};
        \draw[dashed] (7.5,7.5) -- (0,7.5) node[anchor=east] {$\frac{3}{4}$};
        \draw[dashed] (5.83,4.167) -- (5.83,0) node[anchor=north] {$\frac{7}{12}$};
         \draw[dashed] (5.83,4.167) -- (0,4.167) node[anchor=east] {$\frac{5}{12}$};
         \draw[dashed] (7.143,2.857) -- (7.143,0) node[anchor=north] {$\frac{5}{7}$};
         \draw[dashed] (7.143,2.857) -- (0,2.857) node[anchor=east] {$\frac{2}{7}$};
          \node at (10,8.57) [anchor=west] {$\frac{6}{7}$};
           \node at (10,3.33333) [anchor=west] {$\frac{1}{3}$};

          \filldraw (10,3.333) circle (2.5pt);

        \end{tikzpicture}
}
   \caption{Round 1}
    \label{fig:my_label}
     \end{subfigure}
     \hfill
     \begin{subfigure}[b]{0.48\textwidth}
         \centering
     \scalebox{0.7} {
   \begin{tikzpicture} [scale=1]
        
        \fill[black!0!white]  (10,10) -- (10,8.57) -- (7.5,7.5);
        \fill[black!10!white]  (7.5,7.5) -- (10,8.57) -- (10,3.33333)--(7.143,2.857)--(5.83,4.167);
        \fill[black!30!white]  (5,5) -- (7.5,7.5) -- (5.83,4.167);
        \fill[black!50!white]  (7.143,2.857) -- (10,3.333333) -- (10,0);
        
        \draw[thick,->] (0,0) -- (10.5,0) node[anchor=north west] {$c_r$};
        \draw[thick,->] (0,0) -- (0,10.5) node[anchor= south east] {$c_l$};
        \draw[thick] (10,0) -- (5,5);
        \draw[thick] (5,5) -- (10,10);
        \draw[thick] (10,0) -- (10,10);
        \draw[thick] (0,10) -- (10,10);

        \draw[thick] (7.5,7.5) -- (10,8.57);
        \draw[thick] (5.83,4.167) -- (7.5,7.5);
        \draw[thick] (7.143,2.857) -- (10,3.333333);

        \node at (0,0) [anchor=north] {0};
        \node at (10,0) [anchor=north] {1};
        \node at (0,10) [anchor=east] {1};

        \node at (9.5,9) {\textbf{1}};
        \node at (8.8,6) {\textbf{2}};
        \node at (5.5,5) {\textbf{3}};
        \node at (9,2) {\textbf{4}};

           \draw[ultra thick] (1.667,1.667) -- (1.667,8.3333);
           \draw[ultra thick] (1.667,1.667) -- (8.3333,1.667);
           \draw[ultra thick] (8.3333,1.667) -- (8.3333,8.3333);
           \draw[ultra thick] (1.667,8.3333) -- (8.3333,8.3333);

            \draw[dashed] (8.3333,1.667) -- (8.3333,0) node[anchor=north] {$\frac{5}{6}$};
                   \draw[dashed] (1.667,1.667) -- (1.667,0) node[anchor=north] {$\frac{1}{6}$};
                   \draw[dashed] (1.667,1.667) -- (0,1.667) node[anchor=east] {$\frac{1}{6}$};
                   \draw[dashed] (1.667,8.3333) -- (0,8.3333) node[anchor=east] {$\frac{5}{6}$};
            \draw[dashed] (8.3333,3.055) -- (0,3.055) node[anchor=east] {$\frac{11}{36}$};

           \filldraw (8.3333,3.055) circle (2.5pt);

 \node at (5,1.3) {\textbf{$P_r(1)$}};
  \node at (1,5) {\textbf{$P_l(1)$}};
          
        \end{tikzpicture}
}
   \caption{Round 2}
    \label{fig:my_label}
     \end{subfigure}
     \caption{Example of iterative procedure where $a=1$}
     \begin{tabular}{r@{: }l r@{: }l}
\textbf{1} & $\frac{c_l+c_r+1}{5}$ & \textbf{2} & $  \frac{5c_l+c_r+1}{11}$\\
\textbf{3}& $ \frac{1+3c_l-c_r}{5}$ & \textbf{4}& $ [\frac{4c_l+c_r+1}{10},\frac{c_l+4c_r+4}{10}]$
\end{tabular}
        \label{fig:three graphs}
        
\end{figure}

\section{Conclusion}
The first main result of this paper is a characterization of the set of point rationalizable choices for the Hotelling model with waiting costs involving two firms and any inefficiency level for both firms. This set consists of only two locations when the inefficiency levels of the firms are asymmetric and a single location when the firms share the same inefficiency level. Furthermore, we find that as the firms become more similar in terms of their inefficiency level, the two point rationalizable locations are positioned closer to the center of the line. Despite the lack of a pure Nash equilibrium solution for the former case, the use of point rationalizability does provide a solution with high predictive power, involving only two locations.
\newline
\newline
We also characterized the set of point rationalizable choices when three symmetric firms are present in the model. Instead of two points, the set of point rationalizable choices for a firm becomes the interval $[\frac{1+a}{4+3a}, \frac{3+2a}{4+3a}]$. As the inefficiency level $a$ of the firms increases, the interval becomes smaller. Despite the simple assumptions of the model, such as uniformly distributed consumers and linear costs, a different approach will have to be taken to find results for more than three firms.

\section{Appendix}
\subsection{Two firms and asymmetric inefficiency levels}
For convenience, let $\gamma=1+a_1+a_2$. Lemma 2 until 9 are needed to prove the first main result of the paper (Theorem 1). Lemma \ref{symmetry} shows that the set of choices that survives the iterative procedure in each round is symmetric.
\begin{lemma}\label{symmetry}
    $P_1^k$ and $P_2^k$ are symmetric for every $k$. That is, if a choice $c \in P_1^k$, then $1-c \in P_1^k$ and similarly for $P_2^k$.
\end{lemma}

\begin{proof}
Consider firm $i$'s reaction function/correspondence. Note that if a choice $c_i$ is optimal for firm $i$ given $c_j$, then the choice $1-c_i$ is optimal for firm 1 given $1-c_j$. This is because a choice $c_i$ given $c_j$ yields the same market share as a choice $1-c_i$ given $1-c_j$.   
\newline
\newline
Note that $P_1^0=P_2^0=[0,1]$. By induction, assume that $P_i^{k-1}$ and $P_j^{k-1}$ are symmetric, where $k \geq 1$. The choices of firm $i$ that survive round $k$ are the choices that are optimal for firm $i$ in $P_i^{k-1}$ given a $c_j \in P_j^{k-1}$. Consider any choice $c_i \in P_i^{k-1}$ that is optimal for firm $i$ in $P_i^{k-1}$ given $c_j \in P_j^{k-1}$. Then the choice $1-c_i \in P_i^{k-1}$ is optimal for firm $i$ in $P_i^{k-1}$ given $1-c_j \in P_j^{k-1}$. Hence, $P_i^k$ is also symmetric.
\end{proof}
The next Lemma uses the fact that firm 1's reaction function implies that he will choose exactly the same location of firm 2, if firm 2 positions too close to the center of the line.
\begin{lemma}\label{copylemma}
    If $P_2^{k'-1} \subseteq [\frac{a_1}{a_1+a_2},\frac{a_2}{a_1+a_2}]$, then $P_1^{k}=P_2^{k-1}$ for all $k \geq k'$.
\end{lemma}
\begin{proof}
    From the reaction function of firm 1, given any choice $c_2 \in  [\frac{a_1}{a_1+a_2},\frac{a_2}{a_1+a_2}]$, the optimal choice of firm 1 is $c_1=c_2$.
\end{proof}
The next two Lemmas characterize $P_i^k$, if it is given that $P_j^{k-1}$ is an interval.
\begin{lemma}\label{player1lemma}
    If $P_2^{k-1}$ is an interval that can be written as $[l_2^{k-1},1-l_2^{k-1}]$, then $P_1^{k}$ is an interval that can be written as either 

    $$P_1^{k}=[\frac{l_2^{k-1}+a_1}{\gamma},\frac{1-l_2^{k-1}+a_2}{\gamma}] \ \text{if} \ l_2^{k-1} < \frac{a_1}{a_1+a_2}$$ or
    $$P_1^{k}=[l_2^{k-1},1-l_2^{k-1}] \ \text{if} \ l_2^{k-1} \geq \frac{a_1}{a_1+a_2}.$$

\end{lemma}
\begin{proof}
    If  $[l_2^{k-1},1-l_2^{k-1}] \subseteq [\frac{a_1}{a_1+a_2},\frac{a_2}{a_1+a_2}]$, then Lemma \ref{copylemma} implies that $P_1^{k}=P_2^{k-1}$, which is an interval. If $l_2^{k-1} < \frac{a_1}{a_1+a_2}$, then $1-l_2^{k-1} > \frac{a_2}{a_1+a_2}$ and $P_1^k=[\frac{l_2^{k-1}+a_1}{\gamma},\frac{1-l_2^{k-1}+a_2}{\gamma}]$, which is an interval.
\end{proof}

\begin{lemma}\label{Player2Lemma}
     If $P_1^{k-1}$ is an interval that can be written as $[l_1^{k-1},1-l_1^{k-1}]$, then $P_2^{k}$ can be written as either 
\begin{equation}\label{1}
P_2^{k}=[\frac{l_1^{k-1}+a_2}{\gamma},\frac{1-l_1^{k-1}+a_1}{\gamma}] \ \text{if} \ l_1^{k-1} <\frac{1}{2} - (a_2-a_1) 
\end{equation}
   or
   \begin{equation}\label{2}
P_2^{k}=[\frac{\frac{1}{2}+a_1}{\gamma},\frac{\frac{1}{2}+a_2}{\gamma}] \ \text{if} \ l_1^{k-1} \in [\frac{1}{2} - (a_2-a_1),\frac{1}{2}- \frac{a_2-a_1}{2}]
\end{equation}
    or
      \begin{equation}\label{3}
P_2^{k}=[\frac{\frac{1}{2}+a_1}{\gamma},\frac{1-l_1^{k-1}+a_1}{\gamma}] \cup [\frac{l_1^{k-1}+a_2}{\gamma},\frac{\frac{1}{2}+a_2}{\gamma}] \ \text{if} \ l_1^{k-1} > \frac{1}{2}- \frac{a_2-a_1}{2}
\end{equation}
\end{lemma}
\begin{proof}
    As  $P_1^{k-1}=[l_1^{k-1},1-l_1^{k-1}]$, we have that $P_2^{k}=[\frac{\frac{1}{2}+a_1}{\gamma},\frac{1-l_1^{k-1}+a_1}{\gamma}] \cup [\frac{l_1^{k-1}+a_2}{\gamma},\frac{\frac{1}{2}+a_2}{\gamma}]$. We will have two disjoint intervals if $\frac{l_1^{k-1}+a_2}{\gamma} > \frac{1-l_1^{k-1}+a_1}{\gamma}$, which implies that $l_1^{k-1} > \frac{1}{2}- \frac{a_2-a_1}{2}$. Hence in that case, we have $P_2^{k}=[\frac{\frac{1}{2}+a_1}{\gamma},\frac{1-l_1^{k-1}+a_1}{\gamma}]\cup[\frac{l_1^{k-1}+a_2}{\gamma},\frac{\frac{1}{2}+a_2}{\gamma}]$.
    \newline
    \newline
    Similarly, if $\frac{l_1^{k-1}+a_2}{\gamma} \leq \frac{1-l_1^{k-1}+a_1}{\gamma}$ and $\frac{l_1^{k-1}+a_2}{\gamma} \geq \frac{\frac{1}{2}+a_1}{\gamma}$, we have that $P_2^{k}=[\frac{\frac{1}{2}+a_1}{\gamma},\frac{\frac{1}{2}+a_2}{\gamma}]$. Rewriting these inequalities leads to $l_1^{k-1} \leq \frac{1}{2}- \frac{a_2-a_1}{2}$ and $l_1^{k-1} \geq \frac{1}{2} - (a_2-a_1)$.
    \newline
    \newline
    Last, if $\frac{l_1^{k-1}+a_2}{\gamma} < \frac{\frac{1}{2}+a_1}{\gamma}$, we have that $P_2^{k}=[\frac{l_1^{k-1}+a_2}{\gamma},\frac{1-l_1^{k-1}+a_1}{\gamma}]$. Rewriting this inequality leads to $l_1^{k-1} < \frac{1}{2} - (a_2-a_1)$
\end{proof}
Note that an immediate consequence from Lemma \ref{Player2Lemma} is that $P_2^k$   only becomes two disjoint intervals if $l_1^{k-1} > \frac{1}{2} - \frac{a_2-a_1}{2}$ in some round $k$. The next Lemma formally states and proves this result.

\begin{lemma}\label{Andreslemma}
    If $\min P_1^{k'-1} \leq \frac{1}{2}- \frac{a_2-a_1}{2}$ , then both $P_1^k$ and $P_2^{k}$ are all intervals for all $k\in \{0,...,k'\}$
\end{lemma}
\begin{proof}
Let $k'=1$. We have that $P_1^0=P_2^0=[0,1]$, which are intervals. Lemma \ref{Player2Lemma} implies that if $\min P_1^{0} \leq \frac{1}{2}- \frac{a_2-a_1}{2}$, then $P_2^1$ is an interval. Because $P_2^0$ is an interval, Lemma \ref{player1lemma} implies that $P_1^1$ is an interval.
\newline
\newline
Now let $k' \geq 2$ and assume that $\min P_1^{k'-2} \leq \frac{1}{2}- \frac{a_2-a_1}{2}$ and that $P_1^{k}$ and $P_2^{k}$ are an interval for all $k \in \{0,...,k'-1\}$ . We will now prove that if $\min P_1^{k'-1} \leq \frac{1}{2}- \frac{a_2-a_1}{2}$, then $P_1^{k'}$ and $P_2^{k'}$ are an interval. As $P_1^{k'-1}$ is an interval and $\min P_1^{k'-1} \leq \frac{1}{2}- \frac{a_2-a_1}{2}$, Lemma \ref{Player2Lemma} implies that $P_2^{k'}$ is an interval. As $P_2^{k'-1}$ is an interval, Lemma \ref{player1lemma} implies that $P_1^{k'}$ is an interval.
\end{proof}
The next three Lemma's are very similar in terms of structure and result. In summary, we show that there exists a $k$ such that $\min P_1^{k} > \frac{1}{2} - \frac{a_2-a_1}{2}$ and $\min P_2^{k}  \geq \frac{a_1}{a_1+a_2}$.
\begin{lemma}\label{firstpuzzlepiece}
    There exists a $k_1 $ such that $\min P_2^{k_1}  \geq \frac{a_1}{a_1+a_2}$
\end{lemma}
\begin{proof}
We prove the statement by contradiction, that is, assume that $\min P_2^{k}  < \frac{a_1}{a_1+a_2}$ for all $k$. As $P_1^k$ and $P_2^k$ both start as an interval $[0,1]$ when $k=0$, Lemma \ref{Player2Lemma} implies that if $P_2^k$ ever has the form (\ref{2}) or (\ref{3}), then we have that $\min P_2^k= \frac{\frac{1}{2}+a_1}{\gamma}$. Note that $\frac{\frac{1}{2}+a_1}{\gamma} > \frac{a_1}{a_1+a_2}$ as $a_2 > a_1$. Hence, this would imply that $\min P_2^k > \frac{a_1}{a_1+a_2}$ for some $k$, which is a contradiction.
\newline
\newline
From now on, we assume that $P_2^k$ is in form $(\ref{1})$ for every $k$. This leads to the additional restriction that $l_1^k < \frac{1}{2}-(a_2-a_1)$ for all $k$. Furthermore, we have $\min P_2^{k}=l_2^k < \frac{a_1}{a_1+a_2}$ for all $k$ by assumption. Hence, Lemma \ref{player1lemma} and \ref{Player2Lemma} imply that for all $k \in \{1,2,...\}$
$$P_1^k=[\frac{l_2^{k-1}+a_1}{\gamma},\frac{1-l_2^{k-1}+a_2}{\gamma}]$$ 
and
$$P_2^k=[\frac{l_1^{k-1}+a_2}{\gamma},\frac{1-l_1^{k-1}+a_1}{\gamma}]$$ 
Let us then investigate the convergence of $l_2^k$. We have that $l_2^k=\frac{l_1^{k-1}+a_2}{\gamma}=\frac{\frac{l_2^{k-2}+a_1}{\gamma}+a_2}{\gamma}$.
\newline
\newline
As $l_2^k$ and $l_2^{k-2}$ converge to the same point $p$, we have $p=\frac{\frac{p+a_1}{\gamma}+a_2}{\gamma}=\frac{a_2+1}{a_1+a_2+2} > \frac{1}{2}> \frac{a_1}{a_1+a_2}$. This is a contradiction to the assumption that $\min P_2^k < \frac{a_1}{a_1+a_2}$ for all $k$. Hence, there exists a $k_1$ such that $\min P_2^{k_1} \geq \frac{a_1}{a_1+a_2}$ 
  \end{proof}

\begin{lemma}\label{secondpuzzlepiece}
    There exists a $k_2$ such that $\min P_1^{k_2} \geq \frac{1}{2} - (a_2-a_1)$
\end{lemma}
\begin{proof}
We prove the statement by contradiction. Hence, we assume that $\min P_1^k < \frac{1}{2} - (a_2-a_1)$ for all $k$. Lemma \ref{Andreslemma} then implies that $P_1^k$ and $P_2^k$ are intervals for all $k$.
\newline
\newline
Lemma \ref{Player2Lemma} implies that it then must be that $P_2^k=[\frac{l_1^{k-1}+a_2}{\gamma},\frac{1-l_1^{k-1}+a_1}{\gamma}]$ for all $k$. Furthermore, Lemma \ref{firstpuzzlepiece} and Lemma \ref{copylemma} imply that there exists a $k_1$ such that $P_1^{k}=P_2^{k-1}$ for all $k \geq k_1$. Hence, for all $k \geq k_1$ we have

$$P_2^{k}=[\frac{l_1^{k-1}+a_2}{\gamma},\frac{1-l_1^{k-1}+a_1}{\gamma}]$$ and

$$P_1^{k}=[l_2^{k-1},1-l_2^{k-1}]$$

We have that  $l_1^{k+2}=l_2^{k+1}=\frac{l_1^{k}+a_2}{\gamma}$. Investigating the convergence of $l_1^{k}$ then leads to the limit point $p= \frac{a_2}{a_1+a_2} > \frac{1}{2}> \frac{1}{2}- (a_2-a_1)$. This is a contradiction to the assumption that $\min P_1^k < \frac{1}{2} - (a_2-a_1)$ for all $k$. Hence, there exists a $k_2$ such that $\min P_1^{k_2} \geq \frac{1}{2} - (a_2-a_1)$. 
\end{proof}

\begin{lemma}\label{thirdpuzzlepiece}
    There exists a $k_3$ such that $\min P_1^{k_3} > \frac{1}{2} - \frac{a_2-a_1}{2}$
\end{lemma}
\begin{proof}
    By contradiction, assume that $\min P_1^k \leq \frac{1}{2} - \frac{a_2-a_1}{2}$ for all $k$. Lemma \ref{Andreslemma} implies that $P_1^k$ and $P_2^k$ are then an interval for all $k$. Furthermore, Lemma \ref{firstpuzzlepiece} and Lemma \ref{secondpuzzlepiece} imply that there exists a $k^*= \max (k_1,k_2)$ such that $\min P_2^{k^*-1}  \geq \frac{a_1}{a_1+a_2}$ and $\min P_1^{k^*} \geq \frac{1}{2} - (a_2-a_1)$. Then for all $k \geq k^*$

    $$P_1^{k}=[l_2^{k-1},1-l_2^{k-1}]$$ and
    $$P_2^{k}=[\frac{\frac{1}{2}+a_1}{\gamma},\frac{\frac{1}{2}+a_2}{\gamma}]$$

    However, then $l_1^{k+2}=l_2^{k+1}=\frac{\frac{1}{2}+a_1}{\gamma}> \frac{1}{2}- \frac{a_2-a_1}{2}$, which is a contradiction. Hence there exists a $k_3$ such that $\min P_1^{k_3} > \frac{1}{2} - \frac{a_2-a_1}{2}$.
\end{proof}

\subsubsection*{Proof of Theorem 1}
 \begin{proof}
    By Lemma \ref{thirdpuzzlepiece}, there exists a $k_3$ such that $l_1^{k_3-1}> \frac{1}{2} - \frac{a_2-a_1}{2}$. In particular, let $k_3$ be the minimum $k$ such that $l_1^{k_3-1}> \frac{1}{2} - \frac{a_2-a_1}{2}$. Lemma \ref{Andreslemma} then implies that $P_1^k$ and $P_2^k$ are intervals for all $k \in \{0,...,k_3-1\}$. Lemma \ref{Player2Lemma} then implies that

    $$P_2^{k_3}=[\frac{\frac{1}{2}+a_1}{\gamma},\frac{1-l_1^{k_3-1}+a_1}{\gamma}]\cup[\frac{l_1^{k_3-1}+a_2}{\gamma},\frac{\frac{1}{2}+a_2}{\gamma}]$$

    Furthermore by Lemma \ref{player1lemma}, $P_1^{k_3}$ is an interval $[l_1^{k_3},1-l_1^{k_3}]$, as $P_2^{k_3-1}$ is an interval. As $\min P_2^{k_3}=\frac{\frac{1}{2}+a_1}{\gamma} > \frac{a_1}{a_1+a_2}$, Lemma \ref{copylemma} implies that $P_1^{k+1}=P_2^{k}$ for all $k \geq k_3$. For convenience, it is useful to write this as $P_2^{k_3}=[d_2^{l,k_3},d_2^{r,k_3}]\cup[u_2^{l,k_3},u_2^{r,k_3}]$. As $l_1^{k_3-1}> \frac{1}{2} - \frac{a_2-a_1}{2}$, we have that $l_1^{k_3}> \frac{1}{2} - \frac{a_2-a_1}{2}$. Hence, Lemma \ref{Player2Lemma} implies that

    $$P_2^{k_3+1}=[\frac{\frac{1}{2}+a_1}{\gamma},\frac{1-l_1^{k_3}+a_1}{\gamma}]\cup[\frac{l_1^{k_3}+a_2}{\gamma},\frac{\frac{1}{2}+a_2}{\gamma}]=[d_2^{l,k_3+1},d_2^{r,k_3+1}]\cup[u_2^{l,k_3+1},u_2^{r,k_3+1}]$$
Lemma \ref{copylemma} implies that
$$P_1^{k_3+2}=P_2^{k_3+1}$$ 
Now consider $P_2^{k_3+2}$. The choices that are in $P_2^{k_3+2}$ depend on $P_1^{k_3+1}=P_2^{k_3}=[d_2^{l,k_3},d_2^{r,k_3}]\cup[u_2^{l,k_3},u_2^{r,k_3}]$. Now consider the reaction function of firm 2. As all the choices $c_1 \in [d_2^{l,k_3},d_2^{r,k_3}]$ are less than $\frac{1}{2}$, the optimal choice for firm 2 for these choices is $c_2=\frac{c_1+a_2}{\gamma}$. Similarly all the choices $c_1 \in [u_2^{l,k_3},u_2^{r,k_3}]$ are more than $\frac{1}{2}$, and the optimal choice for firm 2 for these choices is $c_2=\frac{c_1+a_1}{\gamma}$. Hence, we have that

$$P_2^{k_3+2}= [\frac{u_2^{l,k_3}+a_1}{\gamma},\frac{u_2^{r,k_3}+a_1}{\gamma}]\cup[\frac{d_2^{l,k_3}+a_2}{\gamma},\frac{d_2^{r,k_3}+a_2}{\gamma}]$$
Hence, for all $k \geq k_3+2$ we have 

$$P_2^{k}= [\frac{u_2^{l,k-2}+a_1}{\gamma},\frac{u_2^{r,k-2}+a_1}{\gamma}]\cup[\frac{d_2^{l,k-2}+a_2}{\gamma},\frac{d_2^{r,k-2}+a_2}{\gamma}]$$
and
$$P_1^{k+1}=P_2^{k}$$
We have that $d_2^{l,k+4}=\frac{u_2^{l,k+2}+a_1}{\gamma}=\frac{\frac{d_2^{l,k}+a_2}{\gamma}+a_1}{\gamma}$. As $d_2^{l,k+4}$ and $d_2^{l,k}$ converge to the same point $p$, we have that $p=\frac{\frac{p+a_2}{\gamma}+a_1}{\gamma}$. Similarly to Lemma \ref{firstpuzzlepiece}, solving for $p$ leads to $p = \frac{a_1+1}{a_1+a_2+2}$.  Similarly for $d_2^{k,r}$ we obtain $d_2^{r,k+4}=\frac{d_2^{r,k}+a_2}{\gamma^2}+\frac{a_1}{\gamma}$. Hence, $d_2^{r,k}$ also converges to the same point $\frac{a_1+1}{a_1+a_2+2}$. By Lemma \ref{symmetry} , the points $u_2^{l,k}$ and $u_2^{r,k}$ converge to $\frac{a_2+1}{a_1+a_2+2}$. As $P_1^k=P_2^{k-1}$ for all $k \geq k_3+2$, firm 1 has the same set of point rationalizable choices.    
    \end{proof}
  
 \subsubsection*{Proof of Proposition 1}
\begin{proof}
    We will prove by induction that $P_1^k=P_2^k=[l^k,u^k]$, where $[l^k,u^k]=[\frac{l^{k-1}+a}{1+2a},\frac{u^{k-1}+a}{1+2a}]$. For the base case $k=1$, we indeed have that $P_1^1=P_2^1=[\frac{0+a}{1+2a},\frac{1+a}{1+2a}]$, as $[l^k,u^k]=[0,1]$.
    \newline
    \newline
    Now assume that $P_1^{k-1}=P_2^{k-1}=[l^{k-1},u^{k-1}]$. As the problem is symmetric, using the reaction function of firm 1 leads to $P_1^{k}=[\frac{l_2^{k-1}+a}{1+2a},\frac{u_2^{k-1}+a}{1+2a}]=[\frac{l^{k-1}+a}{1+2a},\frac{u^{k-1}+a}{1+2a}]$.
    \newline
    \newline
    As $l_1^k$ and $l_1^{k-1}$ converge to the same point $p$, we have that $p=\frac{p+a}{1+2a}$, where solving for $p$ leads to $p=\frac{1}{2}$. Similarly, $u_1^k$ converges to $\frac{1}{2}$. Hence, we have that $P_1=P_2=\{\frac{1}{2}\}$.
\end{proof}
\newpage
\subsection{Three firms and Symmetric Inefficiency Levels}

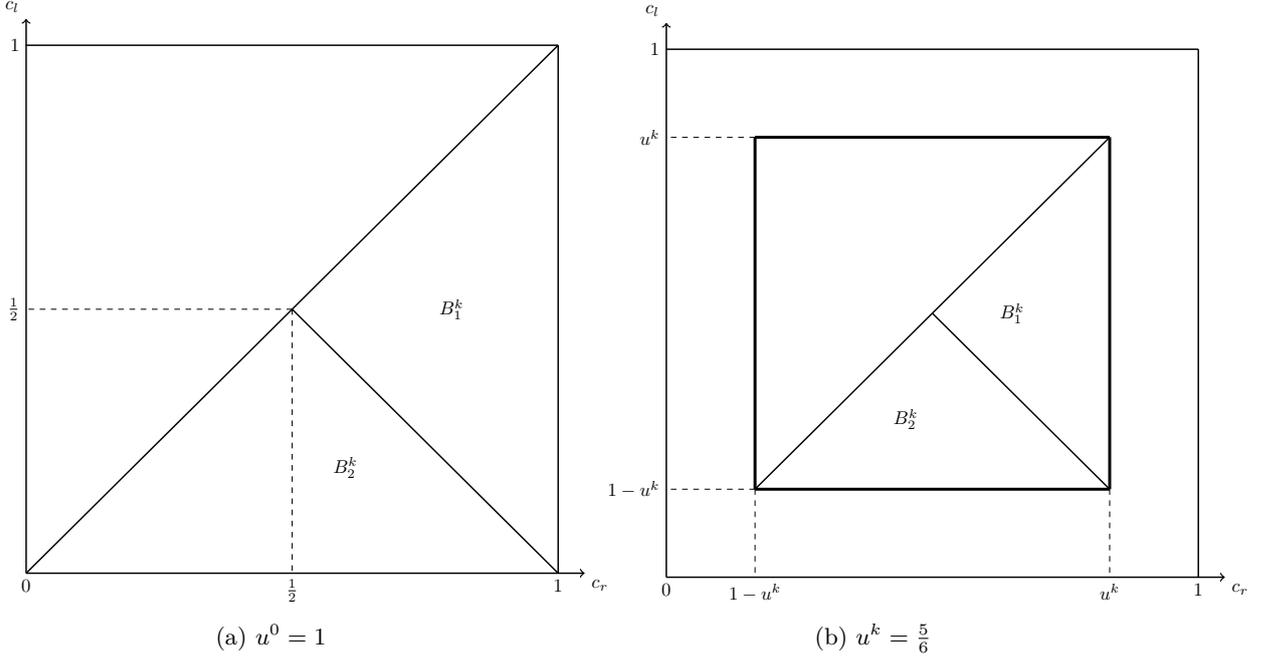
\begin{figure}
     \centering
     \begin{subfigure}[b]{0.48\textwidth}
        \centering
     \scalebox{0.7} {
      \begin{tikzpicture} [scale=1]

        \draw[thick,->] (0,0) -- (10.5,0) node[anchor=north west] {$c_r$};
        \draw[thick,->] (0,0) -- (0,10.5) node[anchor= south east] {$c_l$};
        \draw[thick] (10,0) -- (5,5);
        \draw[thick] (0,0) -- (10,10);
        \draw[thick] (10,0) -- (10,10);
        \draw[thick] (0,10) -- (10,10);
        \draw[dashed] (5,5) -- (5,0);
        \draw[dashed] (5,5) -- (0,5);

        \node at (0,0) [anchor=north] {0};
        \node at (10,0) [anchor=north] {1};
        \node at (0,10) [anchor=east] {1};
        
\node at (8,5) {\textbf{$B_1^k$}};
       \node at (6,2) {\textbf{$B_2^k$}};

          \node at (0,5) [anchor=east] {$\frac{1}{2}$};
             \node at (5,0) [anchor=north] {$\frac{1}{2}$};

        \end{tikzpicture}
}
   \caption{$u^0=1$}
    \label{fig:my_label}
     \end{subfigure}
     \hfill
     \begin{subfigure}[b]{0.48\textwidth}
         \centering
     \scalebox{0.7} {
   \begin{tikzpicture} [scale=1]

        \draw[thick,->] (0,0) -- (10.5,0) node[anchor=north west] {$c_r$};
        \draw[thick,->] (0,0) -- (0,10.5) node[anchor= south east] {$c_l$};
        \draw[thick] (8.3333,1.667) -- (5,5);
        \draw[thick] (1.667,1.667) -- (8.3333,8.3333);
        \draw[thick] (10,0) -- (10,10);
        \draw[thick] (0,10) -- (10,10);

        \node at (6.5,5) {\textbf{$B_1^k$}};
       \node at (4.5,3) {\textbf{$B_2^k$}};

        \node at (0,0) [anchor=north] {0};
        \node at (10,0) [anchor=north] {1};
        \node at (0,10) [anchor=east] {1};

           \draw[ultra thick] (1.667,1.667) -- (1.667,8.3333);
           \draw[ultra thick] (1.667,1.667) -- (8.3333,1.667);
           \draw[ultra thick] (8.3333,1.667) -- (8.3333,8.3333);
           \draw[ultra thick] (1.667,8.3333) -- (8.3333,8.3333);

            \draw[dashed] (8.3333,1.667) -- (8.3333,0) node[anchor=north] {$u^k$};
                   \draw[dashed] (1.667,1.667) -- (1.667,0) node[anchor=north] {$1-u^k$};
                   \draw[dashed] (1.667,1.667) -- (0,1.667) node[anchor=east] {$1-u^k$};
                   \draw[dashed] (1.667,8.3333) -- (0,8.3333) node[anchor=east] {$u^k$};

        \end{tikzpicture}
}
   \caption{$u^k=\frac{5}{6}$}
    \label{fig:my_label}
     \end{subfigure}
     \caption{Visual examples why $B=B_1^k \cup B_2^k$}
        \label{fig:three graphs}
        
\end{figure}

\begin{lemma}\label{29juli1}
    Let $B^k=\{(c_l,c_r):  c_r \in [1-u^k,u^k] \ \text{and} \ c_l\in [1-u^k,c_r]\}$, $B_1^k=\{(c_l,c_r):  c_r \in [\frac{1}{2},u^k] \ \text{and} \ c_l\in [1-u^k,c_r]\}$ and $B_2^k=\{(c_l,c_r):  c_r \in [1-u^k,u^k] \ \text{and} \ c_l\in [1-u^k,\min(c_r,1-c_r)]\}$, where $u^k \in [0,\frac{1}{2}]$. Then $B^k=B_1^k \cup B_2^k$ 
\end{lemma}
\begin{proof}
Let $(c_l,c_r) \in B_1^k \cup B_2^k$. We should then prove that $(c_l,c_r) \in B^k$. It is straightforward to see that $B_1^k \subseteq B^k$ and $B_2^k \subseteq B^k$. Hence, we have that  $(c_l,c_r) \in B$.
\newline
\newline
Next, let $(c_l,c_r) \in B^k$. We should then prove that $(c_l,c_r) \in B_1^k \cup B_2^k$. If $c_r \in [1-u^k,\frac{1}{2}]$, then $\min(c_r,1-c_r)=c_r$. Then we have that $(c_l,c_r) \in B_2^k$. Conversely, let $c_r \in [\frac{1}{2},u^k]$. Then $\min(c_r,1-c_r)=1-c_r$. Hence, if $c_l \in [1-c_r,c_r]$, then $(c_l,c_r) \in B_1^k$ and if $c_l \in [1-u^k,1-c_r]$, then $(c_l,c_r) \in B_2^k$.
\end{proof}

\begin{lemma}\label{29juli2}
Let $F(B_1^k)=\{(1-c_r,1-c_l) \ \text{for all} \ (c_l,c_r) \in B_1^k \}$. Then $F(B_1^k)=B_2^k$
\end{lemma}
\begin{proof}
As  $B_1^k=\{(c_l,c_r):  c_r \in [\frac{1}{2},u^k] \ \text{and} \ c_l\in [1-u^k,c_r]\}$, we have that $\frac{1}{2} \leq c_r \leq u^k$ and $1-c_r \leq c_l \leq c_r$. This can be rewritten to $1-u^k \leq 1-c_r \leq \frac{1}{2}$ and $1-c_r \leq 1-c_l \leq c_r$. Let $c_l^*=1-c_r$ and $c_r^*=1-c_l$. Then we can rewrite the previous two inequalities to $1-u^k \leq c_l^* \leq \frac{1}{2}$ and $c_l^* \leq c_r^* \leq 1-c_l^*$. Hence, we have that $F(B_1^k)=\{(c_l,c_r):  c_l \in [1-u^k,\frac{1}{2}] \ \text{and} \ c_r \in [c_l,1-c_l] \}$. 
\newline
\newline
The range of possible $c_r$'s is largest when $c_l=1-u^k$, which leads to $c_r \in [1-u^k,u^k]$. Furthermore, from $F(B_1^k)$ we observe that $c_l \geq 1-u^k$ and we note that $c_l \leq c_r \leq 1-c_l$ implies $c_l \leq c_r$ and $c_l \leq 1-c_r$. Hence, we can rewrite $F(B_1^k)$ to $F(B_1^k)=\{(c_l,c_r):  c_r \in [1-u^k,u^k] \ \text{and} \ c_l\in [1-u^k,\min(c_r,1-c_r)]\}=B_2^k$
\end{proof}
\begin{lemma}\label{29juli3}
Let $P_l(k)=P_r(k)=[1-u^k,u^k]$. Let $C_1$ denote the set of choices that is optimal in $P_i^k$ for a belief in $B_1^k$. Let $C_1'=\{1-c_i \ \text{for all} \ c_i \in C_1\}$. Then $P_i(k+1)= C_1 \cup C_1'$
\end{lemma}
\begin{proof}
We have that the choices in $P_i(k+1)$ should be optimal in $P_i(k)$ for a belief in $B^k$. Lemma \ref{29juli1} implies that the choices in $P_i(k+1)$ should be optimal in $P_i(k)$ for a belief in $B_1^k \cup B_2^k$. Lemma \ref{29juli2} then implies that the choices in $P_i(k+1)$ should be optimal in $P_i(k)$ for a belief in $B_1^k \cup F(B_1^k)$.  
\newline
\newline
By symmetry, a choice $c_i$ is optimal in $P_i(k)$ for some $(c_l,c_r)$ if and only if the choice $1-c_i$ is optimal in $P_i(k)$ for $(1-c_r,1-c_l)$. Hence, the choices that are optimal for a belief in $F(B_1^k)$ are exactly the choices in $C_1'$.
\end{proof}

\begin{lemma}\label{kohlberg}
 Consider the belief $(c_l,c_r)$ with $c_l \leq c_r$. Then the optimal choice $c_i$ of firm $i$ in $[0,c_l]$ given $(c_l,c_r)$ satisfies $c_i=\min [x_1,c_l]$. 
\end{lemma}
\begin{proof}
\citeA{kohlberg1983equilibrium} proved Lemma \ref{marc1} for three firms.    
\end{proof}
Lemma \ref{x2cl},\ref{firm2marketshare} and \ref{leftmorethanright} are used to prove Lemma \ref{LEFTGOAT}. This Lemma shows that whenever $c_l \leq c_r$ and $c_l > 1-c_r$, then the optimal choice in $[0,c_l]$ yields a strictly greater utility than the optimal choice in $[c_r,1]$. 

\begin{lemma}\label{x2cl}
Consider $(c_l,c_r)$ and let $c_l \geq 1-c_r$ and let $x_1(c_l)\geq c_l$. Then $x_2(c_l)< c_r$.
\end{lemma}
\begin{proof}
As $c_i=c_l$, we have that $ax_1(c_l)=a(x_2(c_l)-x_1(c_l))$, which implies that $x_2(c_l)=2x_1(c_l)$. By contradiction, assume that $x_2(c_l) \geq c_r$. The consumer at $x_2(c_l)$ is indifferent between firm $i$ and firm $r$, which implies that $x_2(c_l)-c_i+a(x_2(c_l)-x_1(c_l))=x_2(c_l)-c_r+a(1-x_2(c_l))$. As $x_2(c_l)=2x_1(c_l)$, this can be rewritten to 
$3ax_1=a-(c_r-c_l)$. This leads to values of $x_1(c_l)$ and $x_2(c_l)$ of 
$$x_1(c_l)=\frac{1}{3}-\frac{c_r-c_l}{3a}$$
and
$$x_2(c_l)=\frac{2}{3}-\frac{2(c_r-c_l)}{3a}.$$
This implies that $x_1(c_l)<\frac{1}{3}$. Hence, we have that $c_l \leq x_1(c_l)< \frac{1}{3}$. As $c_l\geq 1-c_r$, this implies that $c_r>\frac{2}{3}$. But then we have that $x_2(c_l)=\frac{2}{3}-\frac{2(c_r-c_l)}{3a} < \frac{2}{3} < c_r$, which is a contradiction to $x_2(c_l) \geq c_r$.  
\end{proof}

\begin{lemma}\label{firm2marketshare}
Consider firm $i$. Let $c_l \leq c_r$ and $c_l > 1-c_r$. Consider firm $i$'s market share $x_1$ when he chooses his optimal $c_i$ in $[0,c_l]$, given his belief $(c_l,c_r)$. Next consider the belief $(1-c_r,1-c_l)$, for which $c_i'$ is optimal in $[0,1-c_r]$ with a corresponding market share of $x_1'$. If  $x_1' \geq x_1$, then $x_2'-x_1' \geq x_2-x_1$. 
\end{lemma}

\begin{proof}
From Lemma \ref{kohlberg}, we have that $c_i$ satisfies $c_i=\min[x_1,c_l]$ and $c_i'$ satisfies $c_i'=\min[x_1',1-c_r]$. Hence, we have four possible cases to consider, which we will analyze separately. 
\newline
\newline
Case 1: $c_i=x_1$ and $c_i'=x_1'$. For the belief $(c_l,c_r)$, the consumer at $x_1$ is indifferent between buying from firm $i$ and $l$, which implies that $ax_1=c_l-x_1+a(x_2-x_1)$, which can be rewritten to 
\begin{equation}\label{case1x1indifferent}
a(x_2-x_1)=ax_1-(c_l-x_1)    
\end{equation}
For the belief $(1-c_r,1-c_l)$, the consumer at $x_1$ should have a greater cost of buying from firm $l$, as $x_1' \geq x_1$ by assumption. This implies that $1-c_r - x_1+a(x_2'-x_1') \geq x_1'-x_1+ax_1'$, which can be rewritten to 
\begin{equation}\label{case1x1'indifferent}
a(x_2'-x_1') \geq ax_1'-(1-c_r-x_1').   
\end{equation}
But then we have 
$$a(x_2'-x_1') \underset{(\ref{case1x1'indifferent})}{\geq} ax_1'-(1-c_r-x_1') > ax_1-(c_l-x_1) \underset{(\ref{case1x1indifferent})}{=}a(x_2-x_1).$$
As already written down below the inequality/equality symbols above, the first inequality holds because of equation (\ref{case1x1indifferent}) and the third inequality holds because of inequality (\ref{case1x1indifferent}). The second inequality follows from $x_1'\geq x_1$ and $c_l > 1-c_r$. Hence, $x_2'-x_1'>x_2-x_1$.
\newline
\newline
Case 2: $c_i=x_1$ and $c_i'=1-c_r$. For the belief $(c_l,c_r)$, the consumer $x_1$ is indifferent between buying from firm $i$ and $l$ and hence $ax_1=c_l-x_1+a(x_2-x_1)$, which can be rewritten to equation (\ref{case1x1indifferent}). For the belief $(1-c_r,1-c_l)$, firm $i$ and $l$ choose the same location, leading to the firms having the same market shares, which implies that 
\begin{equation}\label{case2x1'indifferent}
a(x_2'-x_1')=ax_1'.   
\end{equation}
Then we have 
$$a(x_2'-x_1')\underset{(\ref{case2x1'indifferent})}{=}ax_1'\geq ax_1-(c_l-x_1)\underset{(\ref{case1x1indifferent})}{=}a(x_2-x_1),$$
and hence $x_2'-x_1'\geq x_2-x_1$.
\newline
\newline
Case 3: $c_i=c_l$ and $c_i'=x_1'$. But then $x_1 \geq c_l > 1-c_r \geq x_1'$ and by assumption $x_1' \geq x_1$. So this case cannot happen.
\newline
\newline
Case 4: $c_i=c_l$ and $c_i'=1-c_r$. For the belief $(c_l,c_r)$, both firm $i$ and $l$ choose the same location, and hence
\begin{equation}\label{case4x1indifferent}
ax_1=a(x_2-x_1).   
\end{equation}
For the belief $(1-c_r,1-c_l)$, firm $i$ and $l$ also choose the same location, which implies that
\begin{equation}\label{case4x1'indifferent}
a(x_2'-x_1')=ax_1'.   
\end{equation}
Then we have 
$$a(x_2'-x_1')\underset{(\ref{case4x1'indifferent})}{=}ax_1'\geq ax_1\underset{(\ref{case4x1indifferent})}{=}a(x_2-x_1),$$
which implies that $x_2'-x_1' \geq x_2-x_1$.
\end{proof}

\begin{lemma}\label{leftmorethanright}
 Let $c_l \leq c_r$ and $c_l > 1-c_r$. Consider firm $i$ with belief $(c_l,c_r)$ and let $c_i$ be the optimal choice in $[0,c_l]$, leading to indifferent consumers $x_1$ and $x_2$ . Also, consider the belief $(1-c_r,1-c_l)$ and let $c_i'$ be the optimal choice in $[0,1-c_r]$. Then $s_i(c_i,c_l,c_r)=x_1 > x_1'=s_i(c_i',1-c_r,1-c_l)$
\end{lemma}
\begin{proof}
By contradiction, assume that $x_1' \geq x_1$. Lemma \ref{firm2marketshare} then implies that $x_2'-x_1' \geq x_2-x_1$. Because the market shares of the firms add up to one, it must be that $1-x_2' \leq 1-x_2$, which implies that $x_2' \geq x_2$. We will now show that $1-x_2'> 1-x_2$, leading to a contradiction. We distinguish several cases, depending on where $x_2$ is positioned on the line.
\newline
\newline
Case 1: $x_2 \geq c_r$. Consider the belief $(c_l,c_r)$. The consumer at $x_2$ is indifferent between firm $l$ and $r$, which implies that $x_2-c_l+a(x_2-x_1)=x_2-c_r+a(1-x_2)$, which can be rewritten to
\begin{equation}\label{case1x2indifferent}
a(1-x_2)=c_r-c_l+a(x_2-x_1).   
\end{equation}
Now consider the belief $(1-c_r,1-c_l)$ and consider again the consumer at $x_2$. As $x_2'\geq x_2$, this consumer should have a higher cost from buying from firm $r$ than from firm $l$, which implies that $x_2-(1-c_l)+ a(1-x_2') \geq x_2-(1-c_r)+a(x_2'-x_1')$, which can be rewritten to
\begin{equation}\label{case1x2'indifferent}
a(1-x_2') \geq c_r-c_l+a(x_2'-x_1').   
\end{equation}
Hence, we obtain
$$a(1-x_2') \underset{(\ref{case1x2'indifferent})}{\geq} c_r-c_l+a(x_2'-x_1') \geq c_r-c_l+a(x_2-x_1) \underset{(\ref{case1x2indifferent})}{=}a(1-x_2)$$
As we previously determined that $1-x_2' \leq 1-x_2$, it then must be that $1-x_2'=1-x_2$, which implies that $x_2'=x_2$. This in turn implies that $x_1'=x_1$. If $x_1 \geq c_l$, then Lemma \ref{kohlberg} implies that the optimal choice of firm $i$ in $[0,c_l]$ is $c_i=c_l$. Lemma \ref{x2cl} then implies that $x_2 < c_r$, which would be a contradiction. Hence, we assume that $x_1 < c_l$. Lemma \ref{kohlberg} then implies that the optimal choice of firm $i$ in $[0,c_l]$ is $c_i=x_1 < c_l$. Consider the consumer at $x_1$, who is indifferent between firm $i$ and firm $l$, which implies that 
\begin{equation}\label{case13}
ax_1=c_l-x_1+a(x_2-x_1).   
\end{equation}
Lemma \ref{kohlberg} implies that for the belief $(1-c_r,1-c_l)$, the optimal choice $c_i'$ in $[0,1-c_r]$ satisfies $c_i'=\min (x_1',1-c_r)=\min (x_1,1-c_r)$. Let $x_1 < 1-c_r$. Then $c_i'=x_1$. The consumer at $x_1'$ is indifferent between firm $i$ and $l$, which implies that $ax_1'=1-c_r-x_1'+a(x_2'-x_1')$. Filling in that $x_1'=x_1$ and $x_2'=x_2$ leads to
\begin{equation}\label{case14}
ax_1=1-c_r-x_1+a(x_2-x_1).   
\end{equation}
However, then we have that 
$$ax_1\underset{(\ref{case14})}{=}1-c_r-x_1+a(x_2-x_1)< c_l-x_1+a(x_2-x_1)\underset{(\ref{case13})}{=}ax_1,$$
which is a contradiction. Finally, let $x_1 \geq 1-c_r$. Then $c_i'=1-c_r$, which implies that $ax_1'=a(x_2'-x_1')$, which in turn implies that $ax_1=a(x_2-x_1).$
However, then we have that 
$$ax_1=a(x_2-x_1) < c_l-x_1+a(x_2-x_1)\underset{(\ref{case13})}{=}ax_1,$$
which is a contradiction.
\newline
\newline
Case 2: $x_2 \leq c_l$. Consider the belief $(c_l,c_r)$. The consumer at $x_2$ is indifferent between firm $l$ and $r$, which implies that $c_l-x_2+a(x_2-x_1)=c_r-x_2+a(1-x_2)$, which can be rewritten to
\begin{equation}\label{case2x2indifferent}
a(1-x_2)=a(x_2-x_1)-(c_r-c_l).   
\end{equation}
Now consider the belief $(1-c_r,1-c_l)$ and consider again the consumer at $x_2$. As $x_2' \geq x_2$, this consumer should have a higher cost from buying from firm $r$ than from firm $l$, which implies that $\lvert 1-c_l-x_2 \rvert+a(1-x_2')\geq \lvert 1-c_r-x_2 \rvert+a(x_2'-x_1')$, which can be rewritten to 
\begin{equation}\label{case2x2'indifferent}
a(1-x_2')\geq a(x_2'-x_1') - (\lvert 1-c_l-x_2 \rvert -\lvert 1-c_r-x_2 \rvert ).   
\end{equation}
Hence, we obtain
$$a(1-x_2') \underset{(\ref{case2x2'indifferent})}{\geq}  a(x_2'-x_1') - (\lvert 1-c_l-x_2 \rvert -\lvert 1-c_r-x_2 \rvert ) \geq  a(x_2'-x_1') - (c_r-c_l)$$
$$> a(x_2-x_1) - (c_r-c_l) \underset{(\ref{case2x2indifferent})}{=} a(1-x_2).$$
As we previously determined that $1-x_2' \leq 1-x_2$, it then must be that $1-x_2'=1-x_2$, which implies that $x_2'=x_2$. This in turn implies that $x_1'=x_1$. As $x_2 \leq c_l$, we have that $x_1 < c_l$. Lemma \ref{kohlberg} then implies that $c_i=x_1 < c_l$. Consider the consumer at $x_1$, who is indifferent between firm $i$ and firm $l$, which implies that $ax_1=c_l-x_1+a(x_2-x_1)$, which again implies equation (\ref{case13}). 
\newline
\newline
Lemma \ref{kohlberg} implies that for the belief $(1-c_r,1-c_l)$, the optimal choice $c_i'$ in $[0,1-c_r]$ satisfies $c_i'=\min (x_1',1-c_r)=\min (x_1,1-c_r)$. Let $x_1 < 1-c_r$. Then $c_i'=x_1 < 1-c_r$. The consumer at $x_1'$ is indifferent between firm $i$ and $l$, which implies that $ax_1'=1-c_r-x_1'+a(x_2'-x_1')$. Filling in that $x_1'=x_1$ and $x_2'=x_2$ leads to equation (\ref{case14}). However, then we have that 
$$ax_1\underset{(\ref{case14})}{=}1-c_r-x_1+a(x_2-x_1)< c_l-x_1+a(x_2-x_1)\underset{(\ref{case13})}{=}ax_1,$$
which is a contradiction. Next, let $x_1 \geq 1-c_r$. Then $c_i'=1-c_r$, which implies that $ax_1'=a(x_2'-x_1')$, which in turn implies that $ax_1=a(x_2-x_1).$
However, then we have that 
$$ax_1=a(x_2-x_1) < c_l-x_1+a(x_2-x_1)\underset{(\ref{case13})}{=}ax_1,$$
which is a contradiction.
\newline
\newline
Case 3.1: $c_l < x_2< c_r$ and $1-c_l \geq x_2$. Consider the belief $(c_l,c_r)$. The consumer at $x_2$ is indifferent between firm $l$ and $r$, which implies that $x_2-c_l+a(x_2-x_1)=c_r-x_2+a(1-x_2)$, which can be rewritten to 
\begin{equation}\label{case3x2indifferent}
a(1-x_2)=2x_2-c_l-c_r+a(x_2-x_1).   
\end{equation}
Now consider the belief $(1-c_r,1-c_l)$ and again the consumer at $x_2$.  As $x_2' \geq x_2$, this consumer should have a higher cost from buying from firm $r$ than from firm $l$, which implies that $1-c_l-x_2 +a(1-x_2') \geq x_2-(1-c_r)+a(x_2'-x_1')$, which can be rewritten to 
\begin{equation}\label{case3x2'indifferent}
a(1-x_2') \geq 2x_2- (1-c_r) - (1-c_l)+a(x_2'-x_1')   
\end{equation}
Hence, we obtain
$$a(1-x_2')  \underset{(\ref{case3x2'indifferent})}{\geq}  2x_2-(1-c_r)-(1-c_l)+a(x_2'-x_1') > 2x_2-c_l-c_r+a(x_2-x_1) \underset{(\ref{case3x2indifferent})}{=}a(1-x_2),$$
which implies $1-x_2'>1-x_2$, which is a contradiction. Here the second inequality follows from $c_l > 1-c_r$.
\newline
\newline
Case 3.2: $c_l < x_2< c_r$ and $1-c_l < x_2$. Consider the belief $(c_l,c_r)$. The consumer at $x_2$ is indifferent between firm $l$ and $r$, which implies that $x_2-c_l+a(x_2-x_1)=c_r-x_2+a(1-x_2)$, which can be rewritten to equation (\ref{case3x2indifferent}). Now consider the belief $(1-c_r,1-c_l)$ and again the consumer at $x_2$.  As $x_2' \geq x_2$, this consumer should have a higher cost from buying from firm $r$ than from firm $l$, which implies that $x_2-(1-c_l) +a(1-x_2') \geq x_2-(1-c_r)+a(x_2'-x_1')$, which can be rewritten to
\begin{equation}\label{case4x2'indifferent}
a(1-x_2') \geq c_r-c_l+a(x_2'-x_1').
\end{equation}
Hence, we obtain
$$a(1-x_2')  \underset{(\ref{case4x2'indifferent})}{\geq} c_r-c_l+a(x_2'-x_1') > 2x_2-c_l-c_r+a(x_2-x_1) \underset{(\ref{case3x2indifferent})}{=} a(1-x_2).$$
As a result, we have that $1-x_2'> 1-x_2$, which is a contradiction. The second inequality follows from $x_2 < c_r$.
\end{proof}

\begin{lemma}\label{LEFTGOAT}
Let $c_l \leq c_r$ and $c_l > 1-c_r$. Consider firm $i$ with belief $(c_l,c_r)$ and let $c_i^L$ be the optimal choice in $[0,c_l]$. Let $c_i^R$ be the optimal choice in $[c_r,1]$. Then $s_i(c_i^L,c_l,c_r) > s_i(c_i^R,c_l,c_r)$    
\end{lemma}
\begin{proof}
Let $c_i^{1-R}$ denote the optimal choice of firm $i$ in $[0,1-c_r]$ for the belief $(1-c_r,1-c_l)$. Symmetry implies that $s_i(c_i^R,c_l,c_r)=s_i(c_i^{1-R},1-c_r,1-c_l)$. As Lemma \ref{leftmorethanright} implies that $s_i(c_i^L,c_l,c_r) > s_i(c_i^{1-R},1-c_r,1-c_l)$, we have proved the statement.     
\end{proof}

\begin{lemma}\label{x1increasing}
Consider firm $i$ with belief $(c_l,c_r)$ and $c_l < c_r$ and let $c_i \in [c_l,c_r)$, with $x_1(c_i) \leq c_i$. Now consider $c_i' \in (c_i,c_r]$. Then it must be that $x_1(c_i') > x_1(c_i)$.
\end{lemma}

\begin{proof}
We will prove the statement by contradiction and assume that $x_1' \leq  x_1$. When firm $i$ chooses $c_i$, the consumer at $x_1$ is indifferent between firm $i$ and firm $l$, meaning that $\lvert c_l-x_1 \rvert + ax_1 = c_i -x_1 + a(x_2-x_1)$, which can be rewritten to either
\begin{equation}\label{middlex1increasing1}
a(x_2-x_1)=\lvert c_l-x_1 \rvert -(c_i-x_1) + ax_1
\end{equation}
or
\begin{equation}\label{middlex1increasing2}
ax_1= c_i-x_1 - \lvert c_l-x_1 \rvert+ a(x_2-x_1).
\end{equation}
Similarly, the consumer at $x_2$ is indifferent between firm $i$ and $r$, meaning that $\lvert c_i-x_2 \rvert + a(x_2-x_1)= \lvert c_r-x_2 \rvert +a(1-x_2)$, which can be rewritten to 
\begin{equation}\label{middlex1increasing3}
a(1-x_2)=\lvert c_i-x_2 \rvert - \lvert c_r-x_2 \rvert +a(x_2-x_1).
\end{equation}
Now consider the choice $c_i'$. As by assumption $x_1' \leq x_1$, a consumer at $x_1$ prefers firm $i$ over firm $l$, which implies that $a(x_2'-x_1') + c_i'-x_1 \leq \lvert c_l-x_1 \rvert + ax_1'$, which can be rewritten to 
\begin{equation}\label{middlex1increasing4}
a(x_2'-x_1') \leq \lvert c_l - x_1 \rvert -(c_i'-x_1)+ ax_1'.
\end{equation}
But then we have 
$$a(x_2'-x_1') \underset{(\ref{middlex1increasing4})}{\leq} \lvert c_l - x_1 \rvert -(c_i'-x_1)+ ax_1' < \lvert c_l - x_1 \rvert -(c_i-x_1)+ ax_1 \underset{(\ref{middlex1increasing1})}{=}a(x_2-x_1).$$
This implies that $x_2'-x_1'<x_2-x_1$, and as $x_1'\leq x_1$ by assumption, we have $x_2'<x_2$ and $1-x_2'>1-x_2$.
\newline
\newline
For the choice $c_i'$, a consumer at $x_2$ prefers to buy from firm $r$ over $i$, as $x_2'<x_2$. This implies that $\lvert c_i'-x_2 \rvert + a(x_2'-x_1') \geq \lvert c_r- x_2 \rvert + a(1-x_2')$, which can be rewritten to
\begin{equation}\label{middlex1increasing5}
a(x_2'-x_1') \geq  \lvert c_r- x_2 \rvert - \lvert c_i'-x_2 \rvert + a(1-x_2').
\end{equation}
Inequality (\ref{middlex1increasing5}) and (\ref{middlex1increasing4}) together imply that  $\lvert c_r- x_2 \rvert - \lvert c_i'-x_2 \rvert + a(1-x_2') \leq \lvert c_l - x_1 \rvert -(c_i'-x_1)+ ax_1'$, which can be rewritten to 
\begin{equation}\label{middlex1increasing6}
a(1-x_2') \leq \lvert c_i'-x_2 \rvert - \lvert c_r- x_2 \rvert + \lvert c_l - x_1 \rvert - (c_i'-x_1) + ax_1'.
\end{equation}
But then we have 
$$a(1-x_2') \underset{(\ref{middlex1increasing6})}{\leq} \lvert c_i'-x_2 \rvert - \lvert c_r- x_2 \rvert + \lvert c_l - x_1 \rvert - (c_i'-x_1) + ax_1'$$
$$ \leq  \lvert c_i'-x_2 \rvert - \lvert c_r- x_2 \rvert + \lvert c_l - x_1 \rvert - (c_i'-x_1) + ax_1$$
$$ \underset{(\ref{middlex1increasing2})}{=} \lvert c_i'-x_2 \rvert - \lvert c_r- x_2 \rvert + \lvert c_l - x_1 \rvert - (c_i'-x_1) +  c_i-x_1 - \lvert c_l-x_1 \rvert+ a(x_2-x_1)$$
$$ = \lvert c_i'-x_2 \rvert - \lvert c_r- x_2 \rvert - (c_i'-c_i) + a(x_2-x_1)$$
$$ \leq \lvert c_i-x_2 \rvert - \lvert c_r- x_2 \rvert + a(x_2-x_1) \underset{(\ref{middlex1increasing3})}{=} a(1-x_2).$$
Hence, this implies that $x_2' \geq x_2$, which is a contradiction.
\end{proof}

\begin{lemma}\label{x2decreasing}
Consider firm $i$ with belief $(c_l,c_r)$ and $c_l < c_r$. Let $c_i \in [c_l,c_r)$ and let $ x_2(c_i) \leq c_i$. Now consider $c_i' \in (c_i,c_r]$. Then $x_2(c_i') < x_2(c_i)$.
\end{lemma}
\begin{proof}
By Lemma \ref{x1increasing}, we know that $x_1' > x_1$ and by contradiction, assume that $x_2' \geq x_2$. Note that this implies that $1-x_2' \leq 1-x_2$. For the choice $c_i$, the consumer at $x_1$ is indifferent between firm $i$ and $l$, which implies that $c_i-x_1+a(x_2-x_1)=\lvert c_l-x_1 \rvert +ax_1$, which can be rewritten to
\begin{equation}\label{middlex2decreasing1}
ax_1=c_i-x_1-\lvert c_l-x_1 \rvert + a(x_2-x_1).
\end{equation}
Similarly, the consumer at $x_2$ is indifferent between firm $i$ and firm $r$, which implies that $c_i-x_2+a(x_2-x_1)=c_r-x_2+a(1-x_2)$, which can be rewritten to
\begin{equation}\label{middlex2decreasing2}
a(1-x_2)=a(x_2-x_1)-(c_r-c_i).
\end{equation}
Now consider the choice $c_i'$. As $x_1'>x_1$, the consumer at $x_1$ prefers to buy from firm $l$ over firm $i$. This implies that $c_i'-x_1+a(x_2'-x_1') \geq \lvert c_l-x_1 \rvert +ax_1'$, which can be rewritten to 
\begin{equation}\label{middlex2decreasing3}
a(x_2'-x_1') \geq \lvert c_l-x_1 \rvert-(c_i'-x_1)+ax_1'.
\end{equation}
As $x_2'\geq x_2$, the consumer at $x_2$ prefers to buy from firm $i$ over $r$, which implies that $c_i'-x_2+a(x_2'-x_1') \leq c_r-x_2+a(1-x_2')$, which can be rewritten to 
\begin{equation}\label{middlex2decreasing4}
a(x_2'-x_1') \leq c_r-c_i'+a(1-x_2').
\end{equation}
Inequality (\ref{middlex2decreasing3}) and (\ref{middlex2decreasing4}) imply that $c_r-c_i'+a(1-x_2') \geq \lvert c_l-x_1 \rvert-(c_i'-x_1)+ax_1'$, which can be rewritten to  
\begin{equation}\label{middlex2decreasing5}
a(1-x_2') \geq \lvert c_l-x_1 \rvert -(c_r-x_1)+ax_1'.
\end{equation}
But then we have 
$$a(1-x_2') \underset{(\ref{middlex2decreasing5})}{\geq}  \lvert c_l-x_1 \rvert -(c_r-x_1)+ax_1' > \lvert c_l-x_1 \rvert -(c_r-x_1)+ax_1$$
$$\underset{(\ref{middlex2decreasing1})}{=} \lvert c_l-x_1 \rvert -(c_r-x_1) + c_i-x_1-\lvert c_l-x_1 \rvert + a(x_2-x_1) =a(x_2-x_1)-(c_r-c_i)\underset{(\ref{middlex2decreasing2})}{=}a(1-x_2).$$
This implies that $1-x_2' > 1-x_2$, which implies that $x_2' < x_2$, which is a contradiction.
\end{proof}

\begin{lemma}\label{x2increasingcrlarge}
Consider firm $i$ with belief $(c_l,c_r)$ and $c_l < c_r$. Let $c_i \in [c_l,c_r)$ and let $x_1(c_i) \leq c_i$ and $x_2(c_i) \geq c_r$. Now consider $c_i' \in (c_i,c_r]$. Then $x_2(c_i') > x_2(c_i)$.
\end{lemma}
\begin{proof}
By lemma \ref{x1increasing} we know that $x_1'>x_1$ and by contradiction, assume that $x_2' \leq x_2$. For the choice $c_i$, the consumer at $x_2$ is indifferent between firm $i$ and firm $r$, which implies that $x_2-c_i+a(x_2-x_1)=x_2-c_r +a(1-x_2)$, which can be rewritten to 
\begin{equation}\label{middlex2increasingcrlarge1}
a(1-x_2)=c_r-c_i + a(x_2-x_1).
\end{equation}
For the choice $c_i'$, as $x_2' \leq x_2$, the consumer at $x_2$ prefers to buy from firm $r$ over firm $i$, which implies that $x_2-c_r + a(1-x_2') \leq x_2-c_i'+a(x_2'-x_1')$, which can be rewritten to 
\begin{equation}\label{middlex2increasingcrlarge2}
a(1-x_2') \leq c_r-c_i'+a(x_2'-x_1').
\end{equation}
Note that as $x_1'> x_1$ and $x_2'\leq x_2$, it must be that $x_2'-x_1'< x_2-x_1$. However, we have that 
$$a(1-x_2') \underset{(\ref{middlex2increasingcrlarge2})}{\leq}c_r-c_i'+a(x_2'-x_1') < c_r-c_i+a(x_2-x_1)\underset{(\ref{middlex2increasingcrlarge1})}{=}a(1-x_2)$$
This implies that $1-x_2' < 1-x_2$, which implies that $x_2'>x_2$, which is a contradiction.
\end{proof}

\begin{lemma}\label{handig}
Consider firm $i$ with belief $(c_l,c_r)$ and $c_l < c_r$. Let $c_i \in [c_l,c_r)$, where $x_1(c_i) \leq c_i$ and $c_i < x_2(c_i)< c_r$. Now consider a $c_i' \in (c_i,c_r]$, where $c_i'$ satisfies $c_i' \leq x_2(c_i')$. Then $c_i''< x_2(c_i'')$ for all $c_i'' \in [c_i,c_i')$.
\end{lemma}

\begin{proof}
 By contradiction, assume that there exists a $c_i'' \in [c_i,c_i')$ such that $c_i'' \geq x_2''$. Because of continuity of $x_2$, $c_i < x_2$  and $c_i'' \geq x_2''$, there exists a $c_i^* \in (c_i,c_i'']$ such that $c_i^*=x_2^*$. Lemma \ref{x2decreasing} then implies that $x_2$ is decreasing for all choices of firm $i$ in $[c_i^*,c_r]$. However, then we have $c_i^*=x_2^*>x_2'\geq c_i'$, which is a contradiction.
\end{proof}

\begin{lemma}\label{x2increasingsmall}
Consider firm $i$ with belief $(c_l,c_r)$ and $c_l < c_r$. Let $c_i \in [c_l,c_r)$ , $x_1(c_i) \leq c_i$ and $c_i < x_2(c_i)< c_r$. Now consider a $c_i' \in (c_i,c_r]$, where $c_i'$ satisfies $c_i' \leq x_2(c_i')$. Then $x_2$ is increasing for all choices in $[c_i,c_i']$.
\end{lemma}

\begin{proof}
We will first prove that $x_2'> x_2$. By lemma \ref{x1increasing} we know that $x_1'>x_1$ and by contradiction, assume that $x_2' \leq x_2$. For the choice $c_i$, the consumer at $x_2$ is indifferent between firm $i$ and firm $r$, which implies that $x_2-c_i+a(x_2-x_1)=\lvert c_r-x_2 \rvert +a(1-x_2)$, which can be rewritten to 
\begin{equation}\label{middlex2increasing1}
a(1-x_2)=2x_2-c_i-c_r + a(x_2-x_1).
\end{equation}
For the choice $c_i'$, as $x_2' \leq x_2$, the consumer at $x_2$ prefers to buy from firm $r$ over firm $i$, which implies that $\lvert c_r-x_2 \rvert + a(1-x_2') \leq x_2-c_i'+a(x_2'-x_1')$, which can be rewritten to 
\begin{equation}\label{middlex2increasing2}
a(1-x_2') \leq 2x_2-c_i'-c_r+a(x_2'-x_1').
\end{equation}
Note that as $x_1'> x_1$ and $x_2'\leq x_2$, it must be that $x_2'-x_1'< x_2-x_1$. However, we have that 
$$a(1-x_2') \underset{(\ref{middlex2increasing2})}{\leq} 2x_2-c_i'-c_r+a(x_2'-x_1') < 2x_2-c_i-c_r + a(x_2-x_1)\underset{(\ref{middlex2increasing1})}{=} a(1-x_2).$$
This implies that $1-x_2' < 1-x_2$, which implies that $x_2'>x_2$, which is a contradiction.\newline
\newline
Finally, as we have shown that $x_2'>x_2$, Lemma \ref{handig} implies that $c_i''< x_2''$ for all $c_i'' \in [c_i,c_i')$. This in turn implies that $x_2$ is increasing in the interval $[c_i,c_i']$.
\end{proof}

\begin{lemma}\label{c1'morethanx_1'}
Consider firm $i$ with belief $(c_l,c_r)$ and $c_l < c_r$. Let $c_i \in [c_l,c_r)$ and let $x_1(c_i) \leq c_i$. Now consider $c_i' \in (c_i,c_r]$. Then $c_i' > x_1(c_i')$.  
\end{lemma}

\begin{proof}
If $c_i \geq x_2$, then Lemma \ref{x2decreasing} implies that $x_2'<x_2$. Hence, we have $c_i'>c_i \geq x_2 > x_2'>x_1'$ for all $c_i' \in (c_i,c_r]$.
\newline
\newline
Next, let $c_i < x_2$ and let $c_i' \leq x_2'$. Lemma \ref{x1increasing} and \ref{x2increasingsmall} then imply that $x_1'>x_1$ and $x_2'> x_2$. By contradiction, assume that $c_i' \leq x_1'$. Consider the choice $c_i$ and the consumer at $x_1'$. As $x_1'>x_1$, this consumer prefers to buy from firm $i$ over firm $l$. This implies that $x_1'-c_l+ax_1 \geq x_1'-c_i+a(x_2-x_1)$, which can be rewritten to 
\begin{equation}\label{middlec1morethanx11}
a(x_2-x_1)\leq ax_1+(c_i-c_l).
\end{equation}
Now consider the consumer at $x_2'$. As $x_2'>x_2$, this consumer prefers to buy from firm $r$ over firm $i$. This implies that $x_2'-c_i+a(x_2-x_1) \geq \lvert x_2'-c_r \rvert +a(1-x_2)$, which can be rewritten to 
\begin{equation}\label{middlec1morethanx12}
a(x_2-x_1) \geq \lvert x_2'-c_r \rvert - (x_2'-c_i)+a(1-x_2).
\end{equation}
Now consider the choice $c_i'$ and the consumer at $x_1'$. This consumer is indifferent between firm $i$ and $l$, which implies that $x_1'-c_i'+a(x_2'-x_1')=x_1'-c_l+ax_1'$, which can be rewritten to . 
\begin{equation}\label{middlec1morethanx13}
ax_1'= a(x_2'-x_1')- (c_i'-c_l) .
\end{equation}
Similarly, the consumer at $x_2'$ is indifferent between $i$ and $r$, which implies that $x_2'-c_i'+a(x_2'-x_1')= \lvert x_2'-c_r \rvert + a(1-x_2')$, which can be rewritten to 
\begin{equation}\label{middlec1morethanx14}
 a(1-x_2')=x_2'-c_i'- \lvert x_2'-c_r \rvert+a(x_2'-x_1').
\end{equation}
We have that (\ref{middlec1morethanx11}) and (\ref{middlec1morethanx12}) imply that $ax_1+(c_i-c_l)\geq \lvert x_2'-c_r \rvert - (x_2'-c_i)+a(1-x_2)$, which can be rewritten to 
\begin{equation}\label{middlec1morethanx15}
 ax_1 \geq \lvert x_2' - c_r \rvert - (x_2'-c_l)+a(1-x_2).
\end{equation}
But then we have that 
$$ax_1 \underset{(\ref{middlec1morethanx15})}{\geq} \lvert x_2' - c_r \rvert - (x_2'-c_l)+a(1-x_2) > \lvert x_2' - c_r \rvert - (x_2'-c_l)+a(1-x_2')$$
$$\underset{(\ref{middlec1morethanx14})}{=} \lvert x_2' - c_r \rvert - (x_2'-c_l)+ (x_2'-c_i')-\lvert x_2'-c_r \lvert + a(x_2'-x_1')=a(x_2'-x_1')-(c_i'-c_l)\underset{(\ref{middlec1morethanx13})}{=} ax_1'.$$
This implies that $x_1 > x_1'$, which is a contradiction.
\newline
\newline
Next, consider the case where $c_i < x_2$ and $c_i' > x_2'$. Because of the continuity of $x_2$, there exists a $c_i^* \in (c_i,c_i')$ such that $c_i^*=x_2(c_i^*)$. Lemma \ref{x2decreasing} implies that $x_2$ is decreasing for all choices in $[c_i^*,c_r]$. But then we have that $c_i'>c_i^*=x_2(c_i^*)> x_2'>x_1'$, which proves the result.
\end{proof}

\begin{lemma}\label{farconstantutility}
Let $c_i \in [c_l,c_r)$ and let $x_1(c_i) < c_l$ and $x_2(c_i)\geq c_r$. Let firm $i$ position marginally to some position to the right $c_i' \in (c_i,c_i^R]$, where $c_i^R$ satisfies $x_1(c_i) < x_1(c_i^R) \leq c_l$. Then $s_i(c_i')=x_2(c_i') -x_1(c_i')=x_2(c_i)-x_1(c_i)=s_i(c_i)$
\end{lemma}
\begin{proof}
We will prove that $x_2'-x_2-(x_1'-x_1)=0$, which proves the lemma. For the choice $c_i$, a consumer at $x_1$ is indifferent between firm $l$ and $i$, which implies that $c_i-x_1+a(x_2-x_1)=c_l-x_1+ax_1$, which can be rewritten to 
\begin{equation}\label{utilityconstantfar1}
ax_1=a(x_2-x_1)+(c_i-c_l).
\end{equation}
For the choice $c_i'$, the consumer at $x_1'$ is indifferent between $l$ and $i$, which implies that $c_i'-x_1'+a(x_2'-x_1')=c_l-x_1'+ax_1'$, which can be rewritten to 
\begin{equation}\label{utilityconstantfar2}
ax_1'=a(x_2'-x_1')+(c_i'-c_l).
\end{equation}
Substracting (\ref{utilityconstantfar1}) from \ref{utilityconstantfar2} leads to $a(x_1'-x_1)=a(x_2'-x_1')-a(x_2-x_1)+c_i'-c_i$, which can be rewritten to
\begin{equation}\label{utilityconstantfar3}
2a(x_1'-x_1)=a(x_2'-x_2)+c_i'-c_i.
\end{equation}
For the choice $c_i$, a consumer at $x_2$ is indifferent between firm $r$ and $i$, which implies that $x_2-c_r+a(1-x_2)=x_2-c_i+a(x_2-x_1)$, which can be rewritten to 
\begin{equation}\label{utilityconstantfar4}
a(1-x_2)=c_r-c_i+a(x_2-x_1).
\end{equation}
For the choice $c_i'$, the consumer at $x_2'$ is indifferent between $r$ and $i$. Furthermore, Lemma \ref{x2increasingcrlarge} implies that $x_2'>x_2$. Hence, we have that 
$x_2'-c_r+a(1-x_2')=x_2'-c_i'+a(x_2'-x_1')$, which can be rewritten to 
\begin{equation}\label{utilityconstantfar5}
a(1-x_2')=c_r-c_i'+a(x_2'-x_1').
\end{equation}
Substracting (\ref{utilityconstantfar5}) from \ref{utilityconstantfar4} leads to 
$a(x_2'-x_2)=a(x_2-x_1)-a(x_2'-x_1')+c_i'-c_i$, which can be rewritten to 
\begin{equation}\label{utilityconstantfar6}
2a(x_2'-x_2)=a(x_1'-x_1)+c_i'-c_i.
\end{equation}
 Finally, substracting (\ref{utilityconstantfar3}) from (\ref{utilityconstantfar6}) implies that $2a(x_2'-x_2-(x_1'-x_1))=a(x_1'-x_1)-a(x_2'-x_2)$, which can be rewritten to $3a(x_2'-x_2-(x_1'-x_1))=0$, which implies that $x_2'-x_2-(x_1'-x_1)=0$, which completes the proof. 
\end{proof}

\begin{lemma}\label{closeconstantutility}
Let $c_i \in [c_l,c_r)$ and let $c_l \leq x_1(c_i) \leq c_i$ and $c_i< x_2(c_i)<c_r$. Let firm i position marginally to the right $c_i' \in (c_i,c_i^R]$, where $c_i^R$ satisfies $c_i^R \leq x_2(c_i^R) \leq c_r$. Then $s_i(c_i')=x_2(c_i') -x_1(c_i')=x_2(c_i)-x_1(c_i)$
\end{lemma}

\begin{proof}
We will prove that $x_2' -x_2 - (x_1' -x_1)=0$, which proves the lemma. For the choice $c_i$, a consumer at $x_1$ is indifferent between firm $l$ and $i$, which implies that $x_1-c_l +ax_1=c_i-x_1+a(x_2-x_1)$, which can be rewritten to  
\begin{equation}\label{utilityconstantclose1}
ax_1=c_i+c_l-2x_1+a(x_2-x_1).
\end{equation}
For the choice $c_i'$, a consumer at $x_1'$ is again indifferent between firm $l$ and $i$. Lemma \ref{x1increasing} implies that $x_1'>x_1$ and Lemma \ref{c1'morethanx_1'} implies that $c_i'>x_1$. Hence, we have that $x_1'-c_l +ax_1'=c_i'-x_1'+a(x_2'-x_1')$, which can be rewritten to $ax_1'=c_i'+c_l-2x_1'+a(x_2'-x_1')$. 
\begin{equation}\label{utilityconstantclose2}
ax_1'=c_i'+c_l-2x_1'+a(x_2'-x_1').
\end{equation}
Substracting (\ref{utilityconstantclose1}) from (\ref{utilityconstantclose2}) implies that $ax_1'-ax_1=c_i'-c_i+2x_1-2x_1'+a(x_2'-x_1')-a(x_2-x_1)$, which can be rewritten to 
\begin{equation}\label{utilityconstantclose3}
2a(x_1'-x_1)=c_i'-c_i+2x_1-2x_1'+a(x_2'-x_2).
\end{equation}
For the choice $c_i$, a consumer at $x_2$ is indifferent between firm $r$ and firm $i$, which implies that $c_r-x_2+a(1-x_2)=x_2-c_i+a(x_2-x_1)$, which can rewritten to  \begin{equation}\label{utilityconstantclose4}
a(1-x_2)=2x_2-c_i-c_r+a(x_2-x_1).
\end{equation}
For the choice $c_i'$, a consumer at $x_2'$ is indifferent between firm $r$ and firm $i$. Furthermore, Lemma \ref{x2increasingsmall} implies that $x_2'> x_2$. We have that $c_r-x_2'+a(1-x_2')=x_2'-c_i'+a(x_2'-c_i')$, which can be rewritten to $a(1-x_2')=2x_2'-c_i'-c_r+a(x_2'-x_1')$. 
\begin{equation}\label{utilityconstantclose5}
a(1-x_2')=2x_2'-c_i'-c_r+a(x_2'-x_1').
\end{equation}
Subtracting (\ref{utilityconstantclose5}) from (\ref{utilityconstantclose4}) implies that $a(x_2'-x_2)=2x_2-2x_2'+c_i'-c_i+a(x_2-x_1)-a(x_2'-x_1')$, which can be rewritten to 
\begin{equation}\label{utilityconstantclose6}
2a(x_2'-x_2)=2x_2-2x_2'+c_i'-c_i+a(x_1'-x_1).
\end{equation}
Finally, subtracting (\ref{utilityconstantclose3}) from (\ref{utilityconstantclose6}) implies that $2a(x_2'-x_2-(x_1'-x_1))=2x_2-x_2'-2x_1+2x_1'+a(x_1'-x_1)-a(x_2'-x_2)$, which can be rewritten to $(x_2'-x_2-(x_1'-x_1))(3a+2)=0$. This implies that $x_2'-x_2-(x_1'-x_1)=0$.
\end{proof}

\begin{lemma}\label{increasingutility}
Let $c_i \in [c_l,c_r)$ and let $c_l \leq x_1(c_i) \leq c_i$ and let $x_2(c_i) \geq c_r$. Then $s_i(c_i')=x_2(c_i') -x_1(c_i')> x_2(c_i)-x_1(c_i)$ for all $c_i' \in (c_i,c_r]$   
\end{lemma}
\begin{proof}
We will prove that $x_2'-x_2-(x_1'-x_1)>0$, proving the Lemma. Consider the choice $c_i$ and a consumer at $x_1$. This consumer is indifferent between $l$ and $i$, which implies that $x_1-c_l+ax_1=c_i-x_1+a(x_2-x_1)$, which can be rewritten to  
\begin{equation}\label{utilityincreasing1}
ax_1=c_l+c_i-2x_1+a(x_2-x_1).
\end{equation}
Now consider $c_i'$ and a consumer at $x_1'$. Lemma \ref{x1increasing} implies that $x_1' > x_1$ and Lemma \ref{c1'morethanx_1'} implies that $c_i' > x_1'$. A consumer at $x_1'$ is indifferent between firm $l$ and $i$, which implies that $x_1'-c_l+ax_1'=c_i'-x_1'+a(x_2'-x_1')$, which can be rewritten to  
\begin{equation}\label{utilityincreasing2}
ax_1'=c_l+c_i'-2x_1'+a(x_2'-x_1').
\end{equation}
Subtracting (\ref{utilityincreasing1}) from (\ref{utilityincreasing2}) implies that $a(x_1'-x_1)=c_i'-c_i+2x_1-2x_1'+a(x_2'-x_1')-a(x_2-x_1)$, which can be rewritten to $2a(x_1'-x_1)=c_i'-c_i+2x_1-2x-1'+a(x_2'-x_2)$.
\begin{equation}\label{utilityincreasing3}
2a(x_1'-x_1)=c_i'-c_i+2x_1-2x-1'+a(x_2'-x_2).
\end{equation}
Now consider $c_i$ and a consumer at $x_2$, who is indifferent between firm $r$ and firm $i$, which implies that $x_2-c_r+a(1-x_2)=x_2-c_i+a(x_2-x_1)$, which can be rewritten to $a(1-x_2)=c_r-c_i+a(x_2-x_1)$. 
\begin{equation}\label{utilityincreasing4}
a(1-x_2)=c_r-c_i+a(x_2-x_1).
\end{equation}
Now consider $c_i'$ and a consumer at $x_2'$. By Lemma \ref{x2increasingcrlarge} we have that $x_2'>x_2$. As a consumer at $x_2'$ is indifferent between firm $r$ and $i$, we have that $x_2'-c_r+a(1-x_2')=x_2'-c_i'+a(x_2'-x_1')$, which can be rewritten to $a(1-x_2')=c_r-c_i'+a(x_2'-x_1')$. 
\begin{equation}\label{utilityincreasing5}
a(1-x_2')=c_r-c_i'+a(x_2'-x_1').
\end{equation}
Subtracting (\ref{utilityincreasing5}) from (\ref{utilityincreasing5}) leads to $a(x_2'-x_2)=c_i'-c_i+a(x_2-x_1)-a(x_2'-x_1')$, which can be rewritten to $2a(x_2'-x_2)=c_i'-c_i+a(x_1'-x_1)$. 
\begin{equation}\label{utilityincreasing6}
2a(x_2'-x_2)=c_i'-c_i+a(x_1'-x_1).
\end{equation}
Finally, subtracting (\ref{utilityincreasing3}) from (\ref{utilityincreasing6}) implies that $2a(x_2'-x_2-(x_1'-x_1))=a(x_1'-x_1)-a(x_2'-x_2)+2x_1'-2x_1$, which can be rewritten to $x_2'-x_2-(x_1'-x_1)=\frac{2(x_1'-x_1)}{3a}>0$, which completes the proof.
\end{proof}

\begin{lemma}\label{decreasingutility}
Consider $c_i \in [c_l,c_r)$ and let $x_1(c_i) < c_l$ and let $c_i < x_2(c_i) < c_r$. Let firm i position marginally to the right $c_i' \in (c_i,c_i^R)$, where $c_i^R$ satisfies that $x_1(c_i^R) \leq c_l$ and $c_i^R \leq x_2(c_i^R) \leq c_r$. Then $s_i(c_i')=x_2(c_i') -x_1(c_i') < x_2(c_i)-x_1(c_i)$.
\end{lemma}
\begin{proof}
We will prove that $x_2'-x_2 - (x_1'-x_1)<0$, which proves the Lemma. Consider the choice $c_i$ and a consumer at $x_1$. This consumer is indifferent between $l$ and $i$, which implies that $c_l-x_1+ax_1=c_i-x_1+a(x_2-x_1)$, which can be rewritten to  
\begin{equation}\label{utilitydecreasing1}
ax_1=c_i-c_l+a(x_2-x_1).
\end{equation}
Now consider $c_i'$ and a consumer at $x_1'$. Lemma \ref{x1increasing} implies that $x_1' > x_1$. A consumer at $x_1'$ is indifferent between firm $l$ and $i$, which implies that $c_l-x_1'+ax_1'=c_i'-x_1'+a(x_2'-x_1')$, which can be rewritten to 
\begin{equation}\label{utilitydecreasing2}
ax_1'=c_i'-c_l+a(x_2'-x_1').
\end{equation}
Subtracting (\ref{utilitydecreasing1}) from (\ref{utilitydecreasing2}) $a(x_1'-x_1)=c_i'-c_i+a(x_2'-x_1')-a(x_2-x_1)$, which can be rewritten to 
\begin{equation}\label{utilitydecreasing3}
2a(x_1'-x_1)=c_i'-c_i+a(x_2'-x_2).
\end{equation}
Now consider $c_i$ and a consumer at $x_2$, who is indifferent between firm $r$ and firm $i$, which implies that $c_r-x_2+a(1-x_2)=x_2-c_i+a(x_2-x_1)$, which can be rewritten to
\begin{equation}\label{utilitydecreasing4}
a(1-x_2)=2x_2-c_i-c_r+a(x_2-x_1).
\end{equation}
Now consider $c_i'$ and a consumer at $x_2'$. By Lemma \ref{x2increasingsmall} we have that $x_2'>x_2$. As a consumer at $x_2'$ is indifferent between firm $r$ and $i$, we have that $c_r-x_2'+a(1-x_2')=x_2'-c_i'+a(x_2'-x_1')$, which can be rewritten to 
\begin{equation}\label{utilitydecreasing5}
a(1-x_2')=2x_2'-c_i'-c_r+a(x_2'-x_1').
\end{equation}
Subtracting (\ref{utilitydecreasing5}) from (\ref{utilitydecreasing4}) leads to $a(x_2'-x_2)=c_i'-c_i+2x_2-2x_2'+a(x_2-x_1)-a(x_2'-x_1')$, which can be rewritten to 
\begin{equation}\label{utilitydecreasing6}
2a(x_2'-x_2)=c_i'-c_i+2x_2-2x_2'+a(x_1'-x_1).
\end{equation}
Finally, subtracting (\ref{utilitydecreasing3}) from (\ref{utilitydecreasing6}) leads to $2a(x_2'-x_2-(x_1'-x_1))=2x_2-2x_2'+a(x_1'-x_1)-a(x_2'-x_2)$, which can be rewritten to $x_2'-x_2-(x_1'-x_1)=\frac{2(x_2-x_2')}{3a}<0$, which completes the proof. 
\end{proof}

\begin{theorem}\label{Leftdirection}
Consider firm $i$ with belief $(c_l,c_r)$, where $c_l \leq c_r$ and $c_l \geq 1-c_r$. If $x_1(c_l) < c_l$, then firm $i's$ unique optimal choice satisfies $c_i^{OPT}=x_1(c_i^{OPT})<c_l$. If $x_1(c_l)=c_l$, then $c_l$ is an optimal choice for firm $i$. 
\end{theorem}

\begin{proof}
Lemma \ref{kohlberg} implies that $c_i^{OPT}$ yields the greatest market share in $[0,c_l]$. Note that if $x_1(c_l) < c_l$, this choice is positioned strictly to the left of $c_l$, and if $x_1(c_l)=c_l$, then $c_i^{OPT}=c_l$. Furthermore, as $x_1(c_l) \leq c_l$, Lemma \ref{x1increasing} implies that $x_1$ is increasing for all choices of firm $i$ in $[c_l,c_r]$. Lemma \ref{c1'morethanx_1'} then implies that $c_i > x_1(c_i)$ for all $c_i \in [c_l,c_r]$. Consider the following 5 cases:
\begin{enumerate}
    \item Let $x_2(c_i) \leq c_i$.
    \item Let $x_2(c_i) > c_i$ and
    \begin{enumerate}
        \item $x_1(c_i) \geq c_l$ and $x_2(c_i) \geq c_r$.
        \item $x_1(c_i) \geq c_l$ and $x_2(c_i) < c_r$.
        \item $x_1(c_i) < c_l$ and $x_2(c_i) \geq c_r$.
        \item $x_1(c_i) < c_l$ and $x_2(c_i) < c_r$.
    \end{enumerate}
\end{enumerate}
First, we assume that $x_1(c_l) < c_l$. We observe that for $c_i=c_l$, we cannot be in case 2(a), as it is assumed that $x_1(c_l)< c_l$. Now assume that there does not exist any $c_i \in [c_l,c_r)$ such that $c_i$ satisfies the conditions of case 2(a). Then it must be that for all $c_i \in [c_l,c_r)$, $c_i$ satisfies one of the other four cases. Note that for any of the other four cases, a marginal increase to the right either leads to a constant utility, or lower utility. To elaborate on this, if $c_i$ satisfies the conditions of case 1, then Lemma \ref{x2decreasing} implies that a marginal increase to the right $c_i'>c_i$ leads to $s_i(c_i')<s_i(c_i)$. If $c_i$ satisfies the conditions of 2(b) or 2(c), then Lemma \ref{closeconstantutility} and \ref{farconstantutility} imply that a marginal increase $c_i'>c_i$ leads to $s_i(c_i')=s_i(c_i)$. Lastly, if it satisfies the conditions of case 2(d), then Lemma \ref{decreasingutility} implies that a marginal increase to the right $c_i'>c_i$ leads to $s_i(c_i')<s_i(c_i)$. Hence, for all $c_i \in [c_l,c_r)$, a marginal increase to the right leads to a lower or same utility as before. This implies that an optimal choice in $[c_i,c_r]$ is given by $c_i=c_l$, which yields a lower utility than $c_i^{OPT}$. 
\newline
\newline
Finally, assume that there exists a $c_i \in (c_l,c_r)$ such that it satisfies the conditions of case 2(a). Let $c_i^*$ denote the smallest choice in $(c_l,c_r)$ that satisfies the conditions of case 2(a). Then all choices in $[c_l,c_i^*)$ do not satisfy the conditions of case 2(a), and similarly to the paragraph before, this implies that any marginal increase to the right leads to a lower or constant utility. Hence, an optimal choice in $[c_l,c_i^*)$ is the choice $c_l$. Now consider the choice $c_i^*$. Note that any choice in $[c_i^*,c_r)$ then satisfies case 2(a), as $x_1$ is increasing on the interval $[c_l,c_r]$ and lemma \ref{x2increasingcrlarge} implies that $x_2$ is increasing on $[c_i^*,c_r]$. Lemma \ref{increasingutility} then implies that for all $c_i' \in [c_i^*,c_r)$ and $c_i'' \in (c_i',c_r]$, we have that $s_i(c_i'')> s_i(c_i')$. This implies that the optimal choice in $[c_i^*,c_r]$ is $c_i=c_r$. Hence, if $c_l > 1-c_r$, then we have that $c_i^{OPT}$ yields a greater market share than $c_l$, and Lemma \ref{LEFTGOAT} implies that $c_i^{OPT}$ yields a greater utility than any choice in $[c_r,1]$. Hence, the unique optimal choice is $c_i^{OPT}$. If $c_l=1-c_r$, then by symmetry, positioning at $c_l$ and $c_r$ yield the same market share. As $c_i^{OPT}$ yields a greater market share than $c_l$, the optimal choice of firm $i$ is $c_i^{OPT}$.
\newline
\newline
Now assume that $x_1(c_l)=c_l=c_i^{OPT}$. Similarly to the case before, we note that if there does not exist any $c_i \in [c_l,c_r)$ such that $c_i$ satisfies the conditions of case 2(a), then the choice $c_i=c_l$ is an optimal choice as the market share of firm $i$ will not increase on the interval $[c_l,c_r]$.
\newline
\newline
Now assume that there exists a $c_i \in (c_l,c_r)$ such that it satisfies the conditions of case 2(a). Note that as $c_l=x_1(c_l)$, Lemma \ref{x2cl} implies that $x_2(c_l)< c_r$ and we are hence in case 2(b). Let $c_i^*$ denote the smallest choice in $(c_l,c_r)$ that satisfies the conditions of case 2(a). Then all choices in $[c_l,c_i^*)$ do not satisfy the conditions of case 2(a), and similarly to the paragraph before, this implies that any marginal increase to the right leads to a lower or constant utility. Hence, an optimal choice in $[c_l,c_i^*)$ is the choice $c_l$. Now consider the choice $c_i^*$. Note that any choice in $[c_i^*,c_r)$ then satisfies case 2(a), as $x_1$ is increasing on the interval $[c_l,c_r]$ and lemma \ref{x2increasingcrlarge} implies that $x_2$ is increasing on $[c_i^*,c_r]$. Lemma \ref{increasingutility} then implies that for all $c_i' \in [c_i^*,c_r)$ and $c_i'' \in (c_i',c_r]$, we have that $s_i(c_i'')> s_i(c_i')$. This implies that the optimal choice in $[c_i^*,c_r]$ is $c_i=c_r$. Hence, if $c_l > 1-c_r$, then Lemma \ref{LEFTGOAT} implies that $c_i^{OPT}$ yields a greater utility than any choice in $[c_r,1]$. Hence, an optimal choice is $c_i^{OPT}$. If $c_l=1-c_r$, then by symmetry, positioning at $c_l$ and $c_r$ yield the same market share, and both are an optimal choice.
\end{proof}

\begin{lemma}\label{nodig}
Consider $(c_l,c_r)$ and let $c_l \geq 1-c_r$ and let $x_1(c_l) \geq c_l$. Then $x_1(c_r)> c_l$ and $x_2(c_r) \leq c_r$ .
\end{lemma}
\begin{proof}
Lemma \ref{x2cl} implies that $x_2(c_l) \leq c_r$. Hence, we have that $x_1(c_l) \geq c_l$ and $x_2(c_l) \leq c_r$. We can then compute the position of the indifferent consumers $x_1(c_l)$ and $x_2(c_l)$. As $c_i=c_l$, we have that $ax_1(c_l)=a(x_2(c_l)-x_1(c_l))$, which implies that $x_2(c_l)=2x_1(c_l)$. Furthermore, a consumer at $x_2(c_l)$ is indifferent between $r$ and $l$, which implies that $x_2(c_l)-c_l+a(x_1(c_l))=c_r-x_2(c_l)+a(1-x_2(c_l))$. We can now use that $x_2(c_l)=2x_1(c_l)$ and solve for $x_1(c_l)$ and $x_2(c_l)$, which leads to 
\begin{equation}\label{nodig1}
x_1(c_l)=\frac{c_l+c_r+a}{4+3a} \geq c_l
\end{equation}
and 
\begin{equation}\label{nodig2}
x_2(c_l)=\frac{2c_l+2c_r+2a}{4+3a} \leq c_r
\end{equation}
Instead of proving that $x_2(c_r) \leq c_r$ and $x_1(c_r)>c_l$, we will consider the belief $(1-c_r,1-c_l)$, and prove that $x_1'(1-c_r)\geq 1-c_r$ and $x_2'(1-c_r)<1-c_l$. Note that when $x_1'(1-c_r)\geq 1-c_r$ and $x_2'(1-c_r)<1-c_l$, we can can denote the exact position of these two indifferent consumers using the same procedure as before. This leads to 
\begin{equation}\label{nodig3}
x_1'(1-c_r)=\frac{2-c_l-c_r+a}{4+3a}
\end{equation}
and
\begin{equation}\label{nodig4}
x_2'(1-c_r)=\frac{4-2c_l-2c_r+2a}{4+3a}.
\end{equation}
We will now prove that $\frac{2-c_l-c_r+a}{4+3a} \geq 1-c_r$, which is equivalent to proving that $c_l \leq (1+a)(3c_r-2)$. We will do this by proving the following:
$$1-c_r \leq c_l \leq \frac{c_r+a}{3+3a} \leq (1+a)(3c_r-2)$$
Note that the first inequality is true by assumption, and (\ref{nodig1}) implies that the second inequality is true. As $1-c_r \leq c_l \leq \frac{c_r+a}{3a+3}$, it must be that $1-c_r \leq  \frac{c_r+a}{3a+3}$, which implies that 
\begin{equation}\label{nodig5}
c_r \geq \frac{2a+3}{3a+4}.
\end{equation}
We will now prove the third inequality. To prove this, if
$$\frac{c_r+a}{3a+3} \leq  (1+a)(3c_r-2),$$ then this can be rewritten to 
$$c_r+a \leq 3(1+a)^2(3c_r-2),$$ which yields
$$c_r+a \leq 9a^2c_r+18ac_r+9c_r-6a^2-12a-6,$$ which implies that
$$9a^2c_r+18ac_r+8c_r \geq 6a^2+13a+6,$$ which finally leads to 
$$c_r \geq \frac{ 6a^2+13a+6}{9a^2+18a+8}=\frac{(2a+3)(3a+2)}{(3a+4)(3a+2)}=\frac{2a+3}{3a+4}.$$ As (\ref{nodig5}) shows that $c_r \geq \frac{2a+3}{3a+4}$, this proves the third inequality. 
\newline
\newline
We still need to prove that $x_2'(1-c_r)=\frac{4-2c_l-2c_r+2a}{4+3a} < 1-c_l$, which can be rewritten to having to prove that $c_l < \frac{a+2c_r}{3a+2}$. However, note that $c_l \leq \frac{c_r+a}{3a+3} < \frac{a+2c_r}{3a+2}$, which completes the proof.

\end{proof}

\begin{lemma}\label{x1lowerthanc_iproperties}
Consider $(c_l,c_r)$ with $c_l \geq 1-c_r$. Let $c_i \in [c_l,x_1(c_i))$ and consider a $c_i'$ that satisfies $c_i' \in (c_i,x_1(c_i')]$ . Then $x_1(c_i') < x_1(c_i)$ and $x_2(c_i') > x_2(c_i)$.
\end{lemma}
\begin{proof}
By contradiction, assume that $x_1(c_i') \geq x_1(c_i)$. For ease of reading, let $x_1(c_i)=x_1$ and $x_1(c_i')=x_1'$ Consider the choice $c_i$, and the consumer at $x_1'$. As $x_1' \geq x_1$, this consumer prefers $i$ over $l$, which implies that $x_1'-c_i+a(x_2-x_1) \leq x_1'-c_l+ax_1$, which can be rewritten to 
\begin{equation}\label{lowerthanx1line1}
a(x_2-x_1) \leq c_i-c_l+ax_1 
\end{equation}
Now consider the choice $c_i'$ and the consumer at $x_1'$. A consumer at $x_1'$ is indifferent between firm $i$ and $l$, which implies that $x_1'-c_i'+a(x_2'-x_1')=x_1'-c_l+ax_1'$ and can be rewritten to 
\begin{equation}\label{lowerthanx1line2}
a(x_2'-x_1')=ax_1'+(c_i'-c_l) 
\end{equation}
But then we have that 
$$a(x_2'-x_1')\underset{(\ref{lowerthanx1line2})}{=}ax_1'+(c_i'-c_l) > ax_1+(c_i-c_l) \underset{(\ref{lowerthanx1line1})}{\geq} a(x_2-x_1),$$
which implies that $x_2'>x_2$, which in turn implies that $1-x_2'< 1-x_2$. Now consider the choice $c_i$ and the consumer at $x_2'$. As $x_2'>x_2$, a consumer at $x_2'$ prefers firm $r$ over firm $i$, which implies that $\lvert x_2'-c_i \rvert + a(1-x_2) \leq  x_2'-c_i+a(x_2-x_1),$ which can be rewritten to 
\begin{equation}\label{lowerthanx1line3}
a(1-x_2) \leq x_2'-c_i-\lvert x_2'-c_i \rvert +a(x_2-x_1)
\end{equation}
Now consider the choice $c_i'$ and a consumer at $x_2'$, who is indifferent between firm $i$ and firm $r$, which implies that $\lvert x_2'-c_r \rvert +a(1-x_2') = x_2'-c_i'+a(x_2'-x_1')$, which can be rewritten to
\begin{equation}\label{lowerthanx1line4}
a(1-x_2')= x_2'-c_i'-\lvert x_2'-c_r \rvert +a(x_2'-x_1')
\end{equation}
But then we have that
$$a(1-x_2')\underset{(\ref{lowerthanx1line4})}{=}x_2'-c_i'-\lvert x_2'-c_r \rvert +a(x_2'-x_1') > x_2'-c_i-\lvert x_2'-c_i \rvert +a(x_2-x_1) \underset{(\ref{lowerthanx1line3})}{\geq} a(1-x_2),$$
which is a contradiction. Hence, we have proven that $x_1' < x_1$. We will now prove that $x_2' > x_2$. By contradiction, assume that $x_2' \leq x_2$, which implies that $1-x_2' \geq 1-x_2$. Consider the choice $c_i$ and the consumer at $x_1$, who is indifferent between $i$ and $l$. Hence for this consumer we have $x_1-c_i+a(x_2-x_1)=x_1-c_l+ax_1$, which can be rewritten to 
\begin{equation}\label{lowerthanx1line5}
ax_1=a(x_2-x_1)-(c_i-c_l)
\end{equation}
Similarly, for the choice $c_i$, for the consumer $x_2$ we have that $x_2-c_i+a(x_2-x_1)=\lvert c_r-x_2 \rvert +a(1-x_2)$, which can be rewritten to 
\begin{equation}\label{lowerthanx1line6}
a(x_2-x_1)=\lvert c_r-x_2 \rvert -(x_2-c_i) +a(1-x_2)
\end{equation}
Now consider the choice $c_i'$, as $x_1'<x_1$, a consumer at $x_1$ prefers to buy from firm $i$ over $l$. This leads to $x_1-c_i'+a(x_2'-x_1') \leq x_1-c_l+ax_1'$, which can be rewritten to 
\begin{equation}\label{lowerthanx1line6}
a(x_2'-x_1')\leq c_i'-c_l+ax_1'
\end{equation}
For the choice $c_i'$, as $x_2'\leq x_2$ by assumption, a consumer at $x_2$ prefers $r$ over $i$. This leads to $x_2-c_i'+a(x_2'-x_1') \geq \lvert c_r-x_2 \rvert +a(1-x_2')$, which can be rewritten to 
\begin{equation}\label{lowerthanx1line7}
a(x_2'-x_1') \geq \lvert c_r-x_2 \rvert -(x_2-c_i')+a(1-x_2')
\end{equation}
We have that (\ref{lowerthanx1line7}) and (\ref{lowerthanx1line6}) together imply that $\lvert c_r-x_2 \rvert -(x_2-c_i')+a(1-x_2') \leq c_i'-c_l+ax_1'$, which can be rewritten to
$$a(1-x_2') \leq x_2-c_l -\lvert c_r-x_2 \rvert +ax_1' <  x_2-c_l -\lvert c_r-x_2 \rvert +ax_1 \underset{(\ref{lowerthanx1line5})}{=} x_2-c_i - \lvert c_r-x_2 \rvert + a(x_2-x_1)$$
$$\underset{(\ref{lowerthanx1line6})}{=}  x_2-c_i - \lvert c_r-x_2 \rvert + \lvert c_r-x_2 \rvert -(x_2-c_i) +a(1-x_2) =a(1-x_2),$$
which implies that $x_2' > x_2$. This contradicts our initial assumption of $x_2' \leq x_2$, and completes the proof. 
\end{proof}

\begin{lemma}\label{middleoptimalchoices}
Consider $(c_l,c_r)$ with $c_l \geq 1-c_r$ and $x_1(c_l) \geq c_l$. Then there exists a $c_i^* \in [c_l,c_r)$ that satisfies $x_1(c_i^*)=c_i^*$ and a $c_i^{**} \in (c_i^*,c_r]$ that satisfies $x_2(c_i^{**})=c_i^{**}$ . 
\end{lemma}

\begin{proof}
By Lemma \ref{x2cl}, we have that $x_2(c_l) < c_r$, which implies that $x_1(c_l) < c_r$. If $x_1(c_l)= c_l$, then we have that $c_i^*=c_l$. If $x_1(c_l)> c_l$, Lemma \ref{x1lowerthanc_iproperties} and continuity of $x_1$ implies that there exists a $c_i^* \in (c_l,x_1(c_l))$ that satisfies $x_1(c_i^*)=c_i^*$ . 
\newline
\newline
By Lemma \ref{x2cl}, we have that $x_2(c_l) < c_r$. Similarly, Lemma \ref{nodig} implies that $x_2(c_r) \leq c_r$. If $x_2(c_r)=c_r$, then we have $c_i^{**}=c_r$. Now assume that $x_2(c_r)< c_r$. By contradiction, assume that $x_2(c_i) \geq c_r$ for some $c_i \in (c_l,c_r)$. As $x_2(c_l) < c_r$ and $x_2(c_r) < c_r$, continuity of the market shares then implies that there exists a $c_i^{c_r} \in (c_l,c_r)$ that satisfies $x_2(c_i^{c_r})=c_r$. If $x_1(c_i^{c_r}) \leq c_i^{c_r}$, Lemma \ref{x2increasingcrlarge} implies that $x_2$ is increasing for all choices in $[c_i^{c_r},c_r]$. Similarly, if $x_1(c_i^{c_r}) > c_i^{c_r}$, then Lemma \ref{x1lowerthanc_iproperties} implies that $x_2$ is increasing on $[x_1(c_i^{c_r}),c_i^*]$. Hence we have that $x_2(c_r)> x_2(c_i^{c_r}) \geq c_r$, which is a contradiction, since we assume that $x_2(c_r)<c_r$. As we have proven that  $x_2(c_i) < c_r$ for all $c_i \in (c_l,c_r)$ and $x_2(c_i^*)>x_1(c_i^*)=c_i^*$, continuity of the market shares implies that there exists a $c_i^{**} \in (c_i^*,c_r)$ such that $x_2(c_i^{**})=c_i^{**}$.
\end{proof}

\begin{theorem}\label{lastdirection}
Consider $(c_l,c_r)$ with $c_l \geq 1-c_r$ and $x_1(c_l)\geq c_l$. Then the optimal choices of firm $i$ satisfy that $c_i^{OPT} \in [c_i^*,c_i^{**}]$, where $c_i^* \in [c_l,c_r)$ satisfies that $x_1(c_i^*)=c_i^*$ and $c_i^{**} \in (c_l,c_r]$ satisfies that $x_2(c_i^{**})=c_i^{**}$.
\end{theorem}

\begin{proof}
Lemma \ref{x1lowerthanc_iproperties} implies $x_1$ is decreasing and $x_2$ is increasing on the interval $[c_l,c_i^*]$. Hence, firm $i'$'s market share is increasing on this interval as well. Lemma \ref{closeconstantutility} implies that firm 1's market share is constant on the interval $[c_i^*,c_i^{**}]$. Finally, Lemma \ref{x2decreasing} and \ref{x1increasing} implies that firm $i$'s market share is decreasing on the interval $[c_i^{**},c_r]$. Hence, the optimal choices of firm $i$ are the choices in the interval $[c_i^*,c_i^{**}]$. 
\end{proof}

\begin{lemma}\label{zondag}
Let $x_1(c_l) \geq c_l$, then $x_2(c_i) < c_r$ for all $c_i \in [c_l,c_r)$
\end{lemma}

\begin{proof}
By Lemma \ref{x2cl}, we have that $x_2(c_l) < c_r$. Similarly, Lemma \ref{nodig} implies that $x_2(c_r) \leq c_r$. By contradiction, assume that $x_2(c_i) \geq c_r$ for some $c_i \in (c_l,c_r)$. As $x_2(c_l) < c_r$ and $x_2(c_r) \leq c_r$, continuity of the market shares then implies that there exists a $c_i^{*} \in (c_l,c_r)$ that satisfies $x_2(c_i^{*})=c_r$. If $x_1(c_i^{*}) \leq c_i^{*}$, Lemma \ref{x2increasingcrlarge} implies that $x_2$ is increasing for all choices in $[c_i^{*},c_r]$. Similarly, if $x_1(c_i^{*}) > c_i^{*}$, then Lemma \ref{x1lowerthanc_iproperties} implies that $x_2$ is increasing on $[x_1(c_i^{*}),c_i^*]$. Hence we have that $x_2(c_r)> x_2(c_i^{*}) \geq c_r$, which is a contradiction, since we assume that $x_2(c_r)\leq c_r$.
\end{proof}

\subsection*{Optimal choices}
In this section, whenever we consider a belief $(c_l,c_r)$ we assume that $c_l \leq c_r$ and $c_l \geq 1-c_r$. Note that by symmetry, for any choice $c_i$ that is optimal for a belief $(c_l,c_r)$, the choice $1-c_i$ is optimal for the belief $(1-c_r,1-c_l)$.
\newline
\newline
First, consider any belief $(c_l,c_r)$ such that $x_1(c_l) \leq c_l$. Theorem \ref{Leftdirection} then implies that the optimal choice of firm $i$ is the choice $c_i$ that satisfies $x_1(c_i)=c_i$. The market share of firm $i$ can be computed if we know the position of $x_2(c_i)$. Note that there are three cases to distinguish. We consider $x_2(c_i) \leq c_l$, $c_l \leq x_2(c_i) \leq c_r$ and $x_2(c_i) \geq c_r$. From now on for ease of notation, we will write $x_1=x_1(c_i)$ and $x_2=x_2(c_i)$. 
\subsubsection*{Case 1: $x_1 \leq c_l$ and $x_2 \leq c_l$}
Consider the first case, where $x_2 \leq c_l$. We can then compute the position of the indifferent consumers $x_1$ and $x_2$ by solving 
$$ax_1=c_l-x_1+a(x_2-x_1) \ \text{and} \ c_l-x_2+a(x_2-x_1)=c_r-x_2+a(1-x_2).$$
This leads to $x_1=\frac{a+c_l+c_r}{2+3a}$ and $x_2=\frac{2ac_r+c_r-c_l-ac_l+a+2a^2}{2a+3a^2}$. So far, we have the following inequalities that should hold:
$$c_l \leq c_r, \  c_l \geq 1-c_r, \  c_l \geq x_1,\  c_l \geq x_2 \ \text{and} \ c_r \geq x_2.$$
As $c_l \geq x_2$ implies that $ c_l \geq x_1$ and $ c_r \geq x_2$, and filling in our found values of $x_2$ and solving for $c_l$ leads to 
$$c_l \leq c_r, \ c_l \geq \frac{2ac_r+c_r+a+2a^2}{1+3a+3a^2} \ \text{and} \ c_l \geq 1-c_r$$
As a result, it should hold that 
$$\max(1-c_r,\frac{2ac_r+c_r+a+2a^2}{1+3a+3a^2}) \leq c_l \leq c_r$$
We consider two cases, which are $\max(1-c_r,\frac{2ac_r+c_r+a+2a^2}{1+3a+3a^2})=1-c_r$ and $\max(1-c_r,\frac{2ac_r+c_r+a+2a^2}{1+3a+3a^2})=\frac{2ac_r+c_r+a+2a^2}{1+3a+3a^2}$. We have that $\max(1-c_r,\frac{2ac_r+c_r+a+2a^2}{1+3a+3a^2})=1-c_r$ if $c_r \leq \frac{1+a}{2+3a}$. Hence, for $c_r \leq \frac{1+a}{2+3a}$ we have that $1-c_r \leq c_l \leq c_r$. However, non-emptiness requires that $c_r \geq 1-c_r$, which implies that $c_r \geq \frac{1}{2} >  \frac{1+a}{2+3a}$. As a result, there is no feasible region when $c_r \leq \frac{1+a}{2+3a}$. Next, consider $c_r \geq \frac{1+a}{2+3a}$. Then we have that $\max(1-c_r,\frac{2ac_r+c_r+a+2a^2}{1+3a+3a^2})=\frac{2ac_r+c_r+a+2a^2}{1+3a+3a^2}$ and hence $\frac{2ac_r+c_r+a+2a^2}{1+3a+3a^2}\leq c_l \leq c_r$. Non-emptiness now requires that $c_r \geq \frac{2ac_r+c_r+a+2a^2}{1+3a+3a^2}$, which implies that $c_r \geq \frac{1+2a}{1+3a}$.
Hence, in this first case, we have that $c_r \in [\frac{1+2a}{1+3a},1]$ and $c_l \in [\frac{2ac_r+c_r+a+2a^2}{1+3a+3a^2},c_r]$ leads to an optimal choice of $\frac{a+c_l+c_r}{2+3a}$. 
\subsubsection*{Case 2: $x_1 \leq c_l$ and $c_l \leq x_2 \leq c_r$}
Consider the second case, where $c_l \leq x_2 \leq c_r$. We can then compute the position of the indifferent consumers $x_1$ and $x_2$ by solving 
$$ax_1=c_l-x_1+a(x_2-x_1) \ \text{and} \ x_2-c_l+a(x_2-x_1)=c_r-x_2+a(1-x_2).$$
This leads to $x_1=\frac{3ac_l+2c_l+ac_r+a^2}{2+6a+3a^2}$ and $x_2=\frac{3ac_l+c_l+2ac_r+c_r+a+2a^2}{2+6a+3a^2}$. So far, we have the following inequalities that should hold:
$$c_l \geq x_1, \ c_l \geq 1-c_r, \  x_2 \geq c_l, \ c_r \geq x_2, \ \text{and} \ c_l\leq c_r$$
Filling in the values for $x_1$ and $x_2$ and solving for $c_l$ leads to 
$$c_l \geq \frac{a+c_r}{3+3a}, \ c_l \geq 1-c_r, \ c_l \leq \frac{2 a c_r+c_r+a + 2a^2 }{1 + 3 a + 3 a^2},\ c_l \leq \frac{3a^2c_r+4ac_r+c_r-a-2a^2}{1+3a}, \ \text{and} \ c_l \leq c_r.$$
As a result, it should hold that
\begin{equation}\label{tweedecase}
 \max( \frac{a+c_r}{3+3a},1-c_r) \leq c_l \leq \min(\frac{2 a c_r+c_r+a + 2a^2 }{1 + 3 a + 3 a^2},\frac{3a^2c_r+4ac_r+c_r-a-2a^2}{1+3a},c_r)   
\end{equation}
We have that $\max( \frac{a+c_r}{3+3a},1-c_r)= 1-c_r$ if $c_r \leq \frac{3+2a}{4+3a}$ and 
$\max( \frac{a+c_r}{3+3a},1-c_r)= \frac{a+c_r}{3+3a}$ if $c_r \geq \frac{3+2a}{4+3a}$. Similarly, we have that $\min(\frac{2 a c_r+c_r+a + 2a^2 }{1 + 3 a + 3 a^2},\frac{3a^2c_r+4ac_r+c_r-a-2a^2}{1+3a},c_r)=\frac{3a^2c_r+4ac_r+c_r-a-2a^2}{1+3a}$ if $c_r \leq \frac{1+2a}{1+3a}$ and $\min(\frac{2 a c_r+c_r+a + 2a^2 }{1 + 3 a + 3 a^2},\frac{3a^2c_r+4ac_r+c_r-a-2a^2}{1+3a},c_r)=\frac{2 a c_r+c_r+a + 2a^2 }{1 + 3 a + 3 a^2}$ if $c_r \geq \frac{1+2a}{1+3a}$. It can be verified that $\frac{3+2a}{4+3a} < \frac{1+2a}{1+3a}$. Hence, we will now consider the cases $c_r \leq \frac{3+2a}{4+3a}$,$ \frac{3+2a}{4+3a} \leq c_r \leq \frac{1+2a}{1+3a}$ and $ \frac{1+2a}{1+3a} \leq c_r \leq 1$.
\newline
\newline
If $c_r \leq \frac{3+2a}{4+3a}$, then $\max(\frac{a+c_r}{3+3a},1-c_r)=1-c_r$ and $\min (\frac{3a^2c_r+4ac_r+c_r-a-2a^2}{1+3a},\frac{2 a c_r+c_r+a + 2a^2 }{1 + 3 a + 3 a^2},c_r)=\frac{3a^2c_r+4ac_r+c_r-a-2a^2}{1+3a}$. Hence, (\ref{tweedecase}) simplifies to $1-c_r \leq c_l \leq \frac{3a^2c_r+4ac_r+c_r-a-2a^2}{1+3a}$. Non emptiness then implies that $c_r \geq \frac{1+4a+2a^2}{2+7a+3a^2}$. Hence, for $c_r \in [\frac{1+4a+2a^2}{2+7a+3a^2},\frac{3+2a}{4+3a}]$ and $c_l \in [1-c_r,\frac{3a^2c_r+4ac_r+c_r-a-2a^2}{1+3a}]$, firm $i$'s optimal choice is $ \frac{3ac_l+2c_l+ac_r+a^2}{2+6a+3a^2}$.
\newline
\newline
If $ \frac{3+2a}{4+3a} \leq c_r \leq \frac{1+2a}{1+3a}$, then $\max(\frac{a+c_r}{3+3a},1-c_r)=\frac{a+c_r}{3+3a}$ and $\min (\frac{3a^2c_r+4ac_r+c_r-a-2a^2}{1+3a},\frac{2 a c_r+c_r+a + 2a^2 }{1 + 3 a + 3 a^2},c_r)=\frac{3a^2c_r+4ac_r+c_r-a-2a^2}{1+3a}$. Hence, (\ref{tweedecase}) simplifies to $\frac{a+c_r}{3+3a} \leq c_l \leq \frac{3a^2c_r+4ac_r+c_r-a-2a^2}{1+3a}$. Non emptiness then implies that $c_r \geq \frac{2a}{1+3a}$, which is satisfied since $c_r \geq \frac{3+2a}{4+3a}$ and $\frac{3+2a}{4+3a} > \frac{2a}{1+3a}$. Hence, for $c_r \in [\frac{3+2a}{4+3a},\frac{1+2a}{1+3a}]$ and $c_l \in [\frac{a+c_r}{3+3a} ,\frac{3a^2c_r+4ac_r+c_r-a-2a^2}{1+3a}]$, firm $i$'s optimal choice is $ \frac{3ac_l+2c_l+ac_r+a^2}{2+6a+3a^2}$.
\newline
\newline
If $ \frac{1+2a}{1+3a} \leq c_r \leq 1$, then $\max(\frac{a+c_r}{3+3a},1-c_r)=\frac{a+c_r}{3+3a}$ and $\min (\frac{3a^2c_r+4ac_r+c_r-a-2a^2}{1+3a},\frac{2 a c_r+c_r+a + 2a^2 }{1 + 3 a + 3 a^2},c_r)=\frac{2 a c_r+c_r+a + 2a^2 }{1 + 3 a + 3 a^2}$. Hence, (\ref{tweedecase}) simplifies to $\frac{a+c_r}{3+3a} \leq c_l \leq \frac{2 a c_r+c_r+a + 2a^2 }{1 + 3 a + 3 a^2}$. Non emptiness then implies that $c_r \geq -a$, which is satisfied since $c_r \geq 0$. Hence, for $c_r \in [\frac{1+2a}{1+3a},1]$ and $c_l \in [\frac{a+c_r}{3+3a} ,\frac{2 a c_r+c_r+a + 2a^2 }{1 + 3 a + 3 a^2}]$, firm $i$'s optimal choice is $ \frac{3ac_l+2c_l+ac_r+a^2}{2+6a+3a^2}$.
\subsubsection*{Case 3: $x_1 \leq c_l$ and $x_2 \geq c_r$}
Consider the third case, where $x_2 \geq c_r$. We can then compute the position of the indifferent consumers $x_1$ and $x_2$ by solving 
$$ax_1=c_l-x_1+a(x_2-x_1) \ \text{and} \ x_2-c_l+a(x_2-x_1)=x_2-c_r+a(1-x_2).$$
This leads to $x_1=\frac{a+3c_l-c_r}{2+3a}$ and $x_2=\frac{3ac_l+c_l-c_r-2ac_r+a+2a^2}{2a+3a^2}$. So far, we have the following inequalities that should hold:
$$c_l \geq 1-c_r, \ c_r \leq x_2 \ c_l \geq x_1, \ c_l \leq x_2, \ \text{and} \ c_l \leq c_r.$$
We have that $c_r \leq x_2$ implies that $ c_l \leq x_2$. Furthermore, how the inequality $c_l \geq x_1$ can be solved depends on the value of $a$. If $a < \frac{1}{3}$, filling in our found values of $x_1$ and $x_2$ and solving for $c_l$ then leads to
$$c_l \geq 1-c_r, \ c_l \geq \frac{3a^2c_r+4ac_r+c_r-a-2a^2}{1+3a}, \ c_l \leq \frac{c_r-a}{1-3a}, \ \text{and} \ c_l \leq c_r.$$
Similarly, if $a=\frac{1}{3}$, then we obtain that $c_l \geq x_1$ is equivalent to $c_l \geq \frac{\frac{1}{3}+3c_l-c_r}{3}$, which can be rewritten to $c_r \geq \frac{1}{3}$. As $1-c_r \leq c_l \leq c_r$, it must be that $c_r \geq \frac{1}{2}$, so the inequality $c_l \geq x_1$ holds. The remaining inequalities to check are 
$$c_l \geq 1-c_r, \ c_l \geq \frac{3a^2c_r+4ac_r+c_r-a-2a^2}{1+3a}, \ \text{and} \ c_l \leq c_r.$$
Lastly, if $a > \frac{1}{3}$, then
$$c_l \geq 1-c_r, \ c_l \geq \frac{3a^2c_r+4ac_r+c_r-a-2a^2}{1+3a}, \ c_l \geq \frac{a-c_r}{3a-1}, \ \text{and} \ c_l \leq c_r.$$
Hence, if $a < \frac{1}{3}$, it should hold that
\begin{equation}\label{lastigecase1}
\max(1-c_r,\frac{3a^2c_r+4ac_r+c_r-a-2a^2}{1+3a}) \leq c_l \leq \min(\frac{c_r-a}{1-3a},c_r).
\end{equation}
Similarly, if $a=\frac{1}{3}$, it should hold that 
\begin{equation}\label{lastigecase2}
\max(1-c_r,\frac{3a^2c_r+4ac_r+c_r-a-2a^2}{1+3a}) \leq c_l \leq c_r.    
\end{equation}
Lastly, if  $a > \frac{1}{3}$, then
\begin{equation}\label{lastigecase3}
\max(\frac{a-c_r}{3a-1},1-c_r,\frac{3a^2c_r+4ac_r+c_r-a-2a^2}{1+3a}) \leq c_l \leq c_r.
\end{equation}
Note that if $a< \frac{1}{3}$, then $\min( \frac{c_r-a}{1-3a},c_r)=c_r$ if $c_r \geq \frac{1}{3}$. As $1-c_r \leq c_l \leq c_r$, it must be that $c_r \geq \frac{1}{2}$, So we have that $\min( \frac{c_r-a}{1-3a},c_r)=c_r$ in this case. Similarly, it can be verified that for $a> \frac{1}{3}$ and $c_r \geq \frac{1}{2}$, that $\max (\frac{a-c_r}{3a-1},1-c_r,\frac{3a^2c_r+4ac_r+c_r-a-2a^2}{1+3a})=\max (1-c_r,\frac{3a^2c_r+4ac_r+c_r-a-2a^2}{1+3a})$. Hence, for any $a$ it should hold that 
$$\max(1-c_r,\frac{3a^2c_r+4ac_r+c_r-a-2a^2}{1+3a}) \leq c_l \leq c_r.$$
We have that $\max (1-c_r,\frac{3a^2c_r+4ac_r+c_r-a-2a^2}{1+3a})=1-c_r$ if $c_r \leq \frac{1+4a+2a^2}{2+7a+3a^2}$. Hence, let $c_r\leq \frac{1+4a+2a^2}{2+7a+3a^2}$, then we have that $1-c_r \leq c_l \leq c_r$. Non-emptiness implies that $c_r \geq \frac{1}{2}$. Hence, for $c_r \in [\frac{1}{2},\frac{1+4a+2a^2}{2+7a+3a^2}]$ and $c_l \in [1-c_r,c_r]$, the optimal choice of firm $i$ is $\frac{a+3c_l-c_r}{2+3a}$. Similarly, if $c_r \geq \frac{1+4a+2a^2}{2+7a+3a^2}$, then we have that $\frac{3a^2c_r+4ac_r+c_r-a-2a^2}{1+3a} \leq c_l \leq c_r$. Non-emptiness implies that $c_r \leq \frac{1+2a}{1+3a}$.  Hence, for $c_r \in [ \frac{1+4a+2a^2}{2+7a+3a^2}, \frac{1+2a}{1+3a}]$ and $c_l \in [\frac{3a^2c_r+4ac_r+c_r-a-2a^2}{1+3a},c_r]$, the optimal choice of firm $i$ is $\frac{a+3c_l-c_r}{2+3a}$.

\subsubsection*{Case 4: $x_1(c_l) \geq c_l$}
Now, we assume that $x_1(c_l) \geq c_l$. We will first compute the region of values of $c_l$ and $c_r$ for which it is satisfied that $x_1(c_l) \geq c_l$. After that, we will compute the region of optimal choices. As $x_1(c_l) \geq c_l$, Lemma \ref{x2cl} then implies that $x_2(c_l) < c_r$. Also, we note that $x_1(c_l) \geq c_l$ implies that $x_2(c_l) > c_l$. For ease of notation, let $x_1(c_l)=x_1$ and $x_2(c_l)=x_2$. We can compute the position of the indifferent consumers $x_1$ and $x_2$ by solving
$$x_1-c_l+ax_1=x_1-c_l+a(x_2-x_1) \ \text{and} \ x_2-c_l+a(x_2-x_1)=c_r-x_2+a(1-x_2)$$
This leads to $x_1=\frac{a+c_l+c_r}{4+3a}$ and $x_2=\frac{2(a+c_l+c_r)}{4+3a}$. Hence, it should hold that 
$$c_l \geq 1-c_r, c_l \leq c_r, x_1 \geq c_l, \ \text{and} \ c_r \geq x_2.$$
As Lemma \ref{x2cl} shows that $x_1 \geq c_l$ implies that $x_2 > c_r$, we disregard the inequality $c_r \geq x_2$ above. Filling in the values of $x_1$ and solving for $c_l$ leads to 
$$c_l \geq 1-c_r, c_l \leq c_r \ \text{and} \ c_l \leq \frac{a+c_r}{3+3a}$$
As a result, it should hold that 
\begin{equation}\label{middlecrclcases}
 1-c_r \leq c_l \leq \min(\frac{a+c_r}{3+3a},c_r)
\end{equation}
We have that $\min(\frac{a+c_r}{3+3a},c_r)=\frac{a+c_r}{3+3a}$ if $c_r \geq \frac{a}{2+3a}$. As $c_r \geq \frac{1}{2}$, (\ref{middlecrclcases}) can be simplified to 
$$1-c_r \leq c_l \leq \frac{a+c_r}{3+3a}.$$
Non emptiness then implies that $c_r \geq \frac{3+2a}{4+3a}$. Hence, for $c_r \in [\frac{3+2a}{4+3a},1]$ and $c_l \in [1-c_r,\frac{a+c_r}{3+3a}]$, it holds that $x_1(c_l) \geq c_l$.
\newline
\newline
Next, we compute the region of optimal choices. As $x_1(c_l) \geq c_l$, Theorem \ref{lastdirection} then implies that any choice $c_i \in [c_i^*,c_i^{**}]$ is an optimal choice. Here $c_i^* \in [c_l,c_r)$ satisfies that $x_1(c_i^*)=c_i^*$ and $c_i^{**} \in (c_l,c_r]$ satisfies that $x_2(c_i^{**})=c_i^{**}$. We consider $c_i^*$. We can compute the position of the indifferent consumers $x_1=x_1(c_i^*)$ and $x_2=x_2(c_i^*)$ by solving 
$$x_1-c_l+ax_1=a(x_2-x_1) \ \text{and} \ x_2-x_1+a(x_2-x_1)=c_r-x_2+a(1-x_2).$$
As $c_i^* \in [c_l,c_r)$, we have that $x_1=x_1(c_i^*)=c_i^* \geq c_l$ and Lemma \ref{zondag} implies that $x_2 \leq c_r$.
This leads to $x_1=c_i^*=\frac{2c_l+2ac_l+ac_r+a^2}{2+5a+3a^2}$ and $x_2=\frac{a + 2 a^2 + c_l + a c_l + c_r + 2 a c_r}{2 + 5 a + 3 a^2}$. 
\newline
\newline
Similarly, for the choice $c_i^{**}=x_2$, we compute the position of the indifferent consumers $x_1=x_1(c_i^{**})$ and $x_2(c_i^{**})$ by solving 
$$x_1-c_l+ax_1=x_2-x_1+a(x_2-x_1) \ \text{and} \ a(x_2-x_1)=c_r-x_2+a(1-x_2).$$
This leads to $x_1=\frac{a + a^2 + c_l + 2 a c_l + c_r + a c_r}{2+5a+3a^2}$ and $x_2=c_i^{**}=\frac{ac_l+2ac_r+2c_r+2a+2a^2}{2+5a+3a^2}$. To summarize, for  $c_r \in [\frac{3+2a}{4+3a},1]$ and $c_l \in [1-c_r, \frac{a+c_r}{3+3a}]$, the region of optimal choices of firm $i$ is given by $[c_i^*,c_i^{**}]=[\frac{2c_l+2ac_l+ac_r+a^2}{2+5a+3a^2},\frac{ac_l+2ac_r+2c_r+2a+2a^2}{2+5a+3a^2}]$.

\subsubsection*{Point Rationalizable choices of three firms}
\begin{figure}[ht]
    \centering
        \begin{tikzpicture} [scale=0.9]
        
        \fill[black!0!white]  (10,10) -- (10,8.57) -- (7.5,7.5);
        \fill[black!10!white]  (7.5,7.5) -- (10,8.57) -- (10,3.33333)--(7.143,2.857)--(5.83,4.167);
        \fill[black!30!white]  (5,5) -- (7.5,7.5) -- (5.83,4.167);
        \fill[black!50!white]  (7.143,2.857) -- (10,3.333333) -- (10,0);
        
        \draw[thick,->] (0,0) -- (10.5,0) node[anchor=north west] {$c_r$};
        \draw[thick,->] (0,0) -- (0,10.5) node[anchor= south east] {$c_l$};
        \draw[thick] (10,0) -- (5,5);
        \draw[thick] (5,5) -- (10,10);
        \draw[thick] (10,0) -- (10,10);
        \draw[thick] (0,10) -- (10,10);
        \draw[dashed] (5,5) -- (5,0);
        \draw[dashed] (5,5) -- (0,5);

        \draw[thick] (7.5,7.5) -- (10,8.57);
        \draw[thick] (5.83,4.167) -- (7.5,7.5);
        \draw[thick] (7.143,2.857) -- (10,3.333333);

        \node at (0,0) [anchor=north] {0};
        \node at (10,0) [anchor=north] {1};
        \node at (0,10) [anchor=east] {1};
        \node at (5,0) [anchor=north] {$\frac{1}{2}$};
        \node at (0,5) [anchor=east] {$\frac{1}{2}$};

        \node at (9.5,9) {\textbf{1}};
        \node at (8.5,6) {\textbf{2}};
        \node at (5.5,5) {\textbf{3}};
        \node at (9,2) {\textbf{4}};

        \draw[dashed] (7.5,7.5) -- (7.5,10) node[anchor=south] {$\frac{1+2a}{1+3a}$};
        \draw[dashed] (7.5,7.5) -- (0,7.5) node[anchor=east] {$\frac{1+2a}{1+3a}$};
        \draw[dashed] (5.83,4.167) -- (5.83,0) node[anchor=north] {$\frac{1+4a+2a^2}{2+7a+3a^2}$};
         \draw[dashed] (5.83,4.167) -- (0,4.167) node[anchor=east] {$\frac{1+3a+a^2}{2+7a+3a^2}$};
         \draw[dashed] (7.143,2.857) -- (7.143,0) node[anchor=north] {$\frac{3+2a}{4+3a}$};
         \draw[dashed] (7.143,2.857) -- (0,2.857) node[anchor=east] {$\frac{1+a}{4+3a}$};
          \node at (10,8.57) [anchor=west] {$\frac{1+3a+2a^2}{1+3a+3a^2}$};
           \node at (10,3.33333) [anchor=west] {$\frac{1}{3}$};

        \end{tikzpicture}
\caption{$a=1$}
\label{a=1}
\begin{tabular}{r@{: }l r@{: }l}
\textbf{1} & $c_i=\frac{c_l+c_r+a}{2+3a}$ & \textbf{2} & $  c_i=\frac{3ac_l+2c_l+ac_r+a^2}{2+6a+3a^2}$\\
\textbf{3}& $c_i= \frac{a+3c_l-c_r}{2+3a}$ & \textbf{4}& $c_i \in [\frac{2c_l+2ac_l+ac_r+a^2}{2+5a+3a^2},\frac{ac_l+2ac_r+2c_r+2a+2a^2}{2+5a+3a^2}]$
\end{tabular}
\label{legend a=1}
\end{figure}
This leads us to the reaction correspondence of firm $i$ for all $(c_l,c_r)$, where $c_l \leq c_r$ and $c_l \geq 1-c_r$:
\[   
c_i(c_l,c_r)=
     \begin{cases}
      \frac{a+c_l+c_r}{2+3a}  & \text{if} \ c_r \in [\frac{1+2a}{1+3a},1] \ \text{and} \ c_l \in [\frac{2ac_r+c_r+a+2a^2}{1+3a+3a^2},c_r]   \\
       \frac{3ac_l+2c_l+ac_r+a^2}{2+6a+3a^2}  & \text{if} \ c_r \in [\frac{1+4a+2a^2}{2+7a+3a^2},\frac{3+2a}{4+3a}] \ \text{and} \ c_l \in [1-c_r,\frac{4ac_r+3a^2c_r+c_r-a-2a^2}{1+3a}]   \\
     \frac{3ac_l+2c_l+ac_r+a^2}{2+6a+3a^2}  & \text{if} \ c_r \in [\frac{3+2a}{4+3a},\frac{1+2a}{1+3a}] \ \text{and} \ c_l \in [\frac{a+c_r}{3+3a},\frac{4ac_r+3a^2c_r+c_r-a-2a^2}{1+3a}]   \\
     \frac{3ac_l+2c_l+ac_r+a^2}{2+6a+3a^2}  & \text{if} \ c_r \in [\frac{1+2a}{1+3a},1] \ \text{and} \ c_l \in [\frac{a+c_r}{3+3a},\frac{c_r+2ac_r+a+2a^2}{1+3a+3a^2}]   \\
      \frac{a+3c_l-c_r}{2+3a}  & \text{if} \ c_r \in [\frac{1}{2},\frac{1+4a+2a^2}{2+7a+3a^2}] \ \text{and} \ c_l \in [1-c_r,c_r]   \\
        \frac{a+3c_l-c_r}{2+3a}  & \text{if} \ c_r \in [\frac{1+4a+2a^2}{2+7a+3a^2},\frac{1+2a}{1+3a}] \ \text{and} \ c_l \in [\frac{4ac_r+3a^2c_r+c_r-a-2a^2}{1+3a},c_r]   \\
        [\frac{2c_l+2ac_l+ac_r+a^2}{2+5a+3a^2},\frac{ac_l+2ac_r+2c_r+2a+2a^2}{2+5a+3a^2}]  & \text{if} \ c_r \in [\frac{3+2a}{4+3a},1] \ \text{and} \ c_l \in [1-c_r,\frac{a+c_r}{3+3a}]   \\
    
     \end{cases}
\]
Figure 8 shows a visual representation of this reaction correspondence of a firm $i$. The first three regions correspond to beliefs such that it is optimal for firm $i$ to position to the left of firm $l$. Region 4 correspond to the beliefs such that it is optimal for firm $i$ to position in between firm $l$ and $r$.
Note that by symmetry, for every choice $c_i$ that is optimal for some belief $(c_l,c_r)$, where $c_l \leq c_r$ and $c_l \geq 1-c_r$, the choice $1-c_i$ is optimal for the belief $(1-c_r,1-c_l)$. This leads us to our first result. 

\begin{lemma}\label{smallestchoice}
Consider $(c_l,c_r)$ with $c_l \leq c_r$ and $c_l \geq 1-c_r$. Then the smallest optimal choice for firm $i$ is $\frac{1+a}{4+3a}$, where $c_l= \frac{1+a}{4+3a}$ and $c_r=\frac{3+2a}{4+3a}$.
\end{lemma}
\begin{proof}
We will prove this by analyzing the reaction correspondence of firm $i$ and computing the smallest for each of the seven cases. Consider the first case where $c_r \in [\frac{1+2a}{1+3a},1]$ and $c_l \in [\frac{2ac_r+c_r+a+2a^2}{1+3a+3a^2},c_r]$. Then firm $i$'s optimal choice is $\frac{a+c_l+c_r}{2+3a}$. The lowest possible location is then when $c_r=\frac{1+2a}{1+3a}$ and $c_l=\frac{2a\frac{1+2a}{1+3a}+\frac{1+2a}{1+3a}+a+2a^2}{1+3a+3a^2}=\frac{1+2a}{1+3a}$, which leads to a location of $\frac{a+c_l+c_r}{2+3a}=\frac{a+\frac{1+2a}{1+3a}+\frac{1+2a}{1+3a}}{2+3a}=\frac{1+a}{1+3a}$.
\newline
\newline
Consider the second case where $c_r \in [\frac{1+4a+2a^2}{2+7a+3a^2},\frac{3+2a}{4+3a}]$ and $c_l \in [1-c_r,\frac{4ac_r+3a^2c_r+c_r-a-2a^2}{1+3a}]$. Then firm $i$'s optimal choice is $\frac{3ac_l+2c_l+ac_r+a^2}{2+6a+3a^2}$. The lowest possible location is then when $c_r=\frac{3+2a}{4+3a}$ and $c_l=1-c_r=\frac{1+a}{4+3a}$, leading to a choice of  $\frac{3a\frac{1+a}{4+3a}+2\frac{1+a}{4+3a}+a\frac{3+2a}{4+3a}+a^2}{2+6a+3a^2}=\frac{1+a}{4+3a}< \frac{1+a}{1+3a}$. Similarly the third case where $ c_r \in [\frac{3+2a}{4+3a},\frac{1+2a}{1+3a}]$ and $ c_l \in [\frac{a+c_r}{3+3a},\frac{4ac_r+3a^2c_r+c_r-a-2a^2}{1+3a}]$, the optimal location of firm $i$ is $\frac{3ac_l+2c_l+ac_r+a^2}{2+6a+3a^2}$. The lowest possible location is then when $c_r=\frac{3+2a}{4+3a}$ and $c_l=\frac{1+a}{4+3a}$, leading to a location of $\frac{1+a}{4+3a}$. For the fourth case, where $ c_r \in [\frac{1+2a}{1+3a},1]$ and $c_l \in [\frac{a+c_r}{3+3a},\frac{c_r+2ac_r+a+2a^2}{1+3a+3a^2}]$, the optimal choice is $\frac{3ac_l+2c_l+ac_r+a^2}{2+6a+3a^2}$. As $c_r \geq \frac{1+2a}{1+3a} > \frac{3+2a}{4+3a}$ and $c_l > \frac{1+a}{4+3a}$, we have a lowest optimal location strictly greater than $\frac{1+a}{4+3a}$.
\newline
\newline
For the fifth case, consider $c_r \in [\frac{1}{2},\frac{1+4a+2a^2}{2+7a+3a^2}]$ and $c_l \in [1-c_r,c_r]$, which leads to an optimal choice of $ \frac{a+3c_l-c_r}{2+3a}$. The lowest possible location is then when $c_r=\frac{1+4a+2a^2}{2+7a+3a^2}$ and $c_l=1-c_r=\frac{1+3a+a^2}{2+7a+3a^2}$, leading to a choice of  $\frac{a+3(\frac{1+3a+a^2}{2+7a+3a^2})-\frac{1+4a+2a^2}{2+7a+3a^2}}{2+3a}= \frac{(a+1)^2}{2+7a+3a^2} > \frac{1+a}{4+3a}$. Similarly, for the sixth case, where $ c_r \in [\frac{1+4a+2a^2}{2+7a+3a^2},\frac{1+2a}{1+3a}]$ and $ c_l \in [\frac{4ac_r+3a^2c_r+c_r-a-2a^2}{1+3a},c_r]$, the optimal choice is $ \frac{a+3c_l-c_r}{2+3a}$. The lowest possible location is then again when $c_r=\frac{1+4a+2a^2}{2+7a+3a^2}$ and $c_l=\frac{4ac_r+3a^2c_r+c_r-a-2a^2}{1+3a}=\frac{1+3a+a^2}{2+7a+3a^2}$, leading to the same location as the fifth case.
\newline
\newline
For the last case, consider $ c_r \in [\frac{3+2a}{4+3a},1]$ and $c_l \in [1-c_r,\frac{a+c_r}{3+3a}]$, which leads to a set of optimal choices of $ [\frac{2c_l+2ac_l+ac_r+a^2}{2+5a+3a^2},\frac{ac_l+2ac_r+2c_r+2a+2a^2}{2+5a+3a^2}]$. The lowest optimal choice is then when $c_r=\frac{3+2a}{4+3a}$ and $c_l=1-c_r=\frac{1+a}{4+3a}$, which leads to a lowest optimal choice of $\frac{2\frac{1+a}{4+3a}+2a\frac{1+a}{4+3a}+a\frac{3+2a}{4+3a}+a^2}{2+5a+3a^2}=\frac{1+a}{4+3a}$. 
\end{proof}

\subsubsection*{Proof of Theorem 2}
\begin{proof}
We will first prove by induction that for any round $k$ of the iterative procedure, we have that $P_i^k=[1-U^k,U^k]$, where $U^k=\frac{(3+2a)(a+U^{k-1})}{3(1+a)^2}$ and $U^0=1$. We will also prove that $U^k > \frac{3+2a}{4+3a}$ for all $k \geq 0$. Note that we have that $P_i^0=[0,1]$, which is our base case.
\newline
\newline
Next, let the iterative procedure be in round $k \geq 1$, and assume that $P_i^{k-1}=[1-U^{k-1},U^{k-1}]$ and assume that $U^{k-1} > \frac{3+2a}{4+3a}$. We will then prove that $P_i^k=[1-\frac{(3+2a)(a+U^{k-1})}{3(1+a)^2},\frac{(3+2a)(a+U^{k-1})}{3(1+a)^2}]$ and that $U^k > \frac{3+2a}{4+3a}$.  
\newline
\newline
As $U^{k-1}>  \frac{3+2a}{4+3a}$, we have that $[\frac{1+a}{4+3a},\frac{3+2a}{4+3a}] \subseteq P_i^{k-1}$. Lemma \ref{smallestchoice} then implies that for any belief where $c_l \leq c_r$ and $c_l \geq 1-c_r$, the smallest optimal choice of firm $i$ is $\frac{1+a}{4+3a}$, where $c_l=\frac{1+a}{4+3a} \in P_i^{k-1}$ and $c_r=\frac{3+2a}{4+3a} \in P_i^{k-1}$. Hence, by symmetry, for the belief $(1-c_r,1-c_l)$, the greatest optimal choice is $1-\frac{1+a}{4+3a}=\frac{3+2a}{4+3a}$. 
\newline
\newline
We will now show that for the belief $(c_l,c_r)$, where $c_l \leq c_r$ and $c_l \geq 1-c_r$, the greatest optimal choice of firm $i$ is $\frac{(3+2a)(a+U^{k-1})}{3(1+a)^2}$. Consider the reaction correspondence of firm $i$. The first possible optimal choice is $ \frac{a+c_l+c_r}{2+3a}$, which is increasing in $c_l$ and $c_r$. Hence, this choice can be at most $\frac{a+2U^{k-1}}{2+3a}$. It can be verified that $\frac{a+2U^{k-1}}{2+3a} < \frac{(3+2a)(a+U^{k-1})}{3(1+a)^2}$. Next, consider the optimal choice $ \frac{3ac_l+2c_l+ac_r+a^2}{2+6a+3a^2}$, which is increasing in $c_l$ and $c_r$. Hence, this choice can be at most $\frac{3a\frac{U^{k-1}+2aU^{k-1}+a+2a^2}{1+3a+3a^2}+2\frac{U^{k-1}+2aU^{k-1}+a+2a^2}{1+3a+3a^2}+aU^{k-1}+a^2}{2+6a+3a^2}=\frac{(1+a)(a+U^{k-1})}{1+3a+3a^2}$. It can be verified that $\frac{(1+a)(a+U^{k-1})}{1+3a+3a^2} < \frac{(3+2a)(a+U^{k-1})}{3(1+a)^2}$. Next, for the choice $\frac{a+3c_l-c_r}{2+3a}$ we observe that $\frac{a+3c_l-c_r}{2+3a} \leq \frac{a+2c_l}{2+3a} \leq \frac{a+2U^{k-1}}{2+3a} < \frac{(3+2a)(a+U^{k-1})}{3(1+a)^2}$. 
\newline
\newline
Lastly, consider the interval of optimal choices $[\frac{2c_l+2ac_l+ac_r+a^2}{2+5a+3a^2},\frac{ac_l+2ac_r+2c_r+2a+2a^2}{2+5a+3a^2}]$, corresponding to $c_r \in [\frac{3+2a}{4+3a},U^{k-1}]$ and $c_l \in [1-c_r, \frac{a+c_r}{3+3a}]$, and as $U^{k-1}> \frac{3+2a}{4+3a}$, the latter two sets are nonempty. The largest choice is then when $c_r=U^{k-1}$ and $c_l=\frac{a+U^{k-1}}{3+3a}$. This leads to an interval of choices of \newline $[\frac{2\frac{a+U^{k-1}}{3+3a}+2a\frac{a+U^{k-1}}{3+3a}+aU^{k-1}+a^2}{2+5a+3a^2},\frac{a\frac{a+U^{k-1}}{3+3a}+2aU^{k-1}+2U^{k-1}+2a+2a^2}{2+5a+3a^2}]=[\frac{a+U^{k-1}}{3+3a},\frac{(3+2a)(a+U^{k-1})}{3(1+a)^2}]$.  Hence, for this belief, any choice in $[\frac{a+U^{k-1}}{3+3a},\frac{(3+2a)(a+U^{k-1})}{3(1+a)^2}]$ survives round $k$ of the iterative procedure. By symmetry, for the belief $(1-c_r,1-c_l)$, the choices $[1-\frac{(3+2a)(a+U^{k-1})}{3(1+a)^2},\frac{3+2a-U^{k-1}}{3+3a}]$ are optimal. Hence, we have that $P_i^k=[1-\frac{(3+2a)(a+U^{k-1})}{3(1+a)^2},\frac{(3+2a)(a+U^{k-1})}{3(1+a)^2}]$.
\newline
\newline
Next, we will verify that $\frac{(3+2a)(a+U^{k-1})}{3(1+a)^2} > \frac{3+2a}{4+3a}$. By assumption, we have that $U^{k-1} > \frac{3+2a}{4+3a}$. Hence, we have that $U^k=\frac{(3+2a)(a+U^{k-1})}{3(1+a)^2} > \frac{(3+2a)(a+\frac{3+2a}{4+3a})}{3(1+a)^2} = \frac{3+2a}{4+3a}$  
\newline
\newline
Lastly, we are interested in the set of point rationalizable choices of each firm. As $U^k$ and $U^{k-1}$ converge to the same point $p$, it should hold that $p=\frac{(3+2a)(a+p)}{3(1+a)^2}$, which implies that $p=\frac{3+2a}{4+3a}$. Hence, the set of point rationalizable choices is given by $P_i=[1-p,p]=[\frac{1+a}{4+3a},\frac{3+2a}{4+3a}]$
\end{proof}
\newpage
\bibliographystyle{apacite}
\bibliography{references}

\begin{thebibliography}{}

\bibitem [\protect \citeauthoryear {%
Aumann%
\ \BBA {} Brandenburger%
}{%
Aumann%
\ \BBA {} Brandenburger%
}{%
{\protect \APACyear {1995}}%
}]{%
aumann1995epistemic}
\APACinsertmetastar {%
aumann1995epistemic}%
\begin{APACrefauthors}%
Aumann, R.%
\BCBT {}\ \BBA {} Brandenburger, A.%
\end{APACrefauthors}%
\unskip\
\newblock
\APACrefYearMonthDay{1995}{}{}.
\newblock
{\BBOQ}\APACrefatitle {Epistemic conditions for Nash equilibrium} {Epistemic conditions for nash equilibrium}.{\BBCQ}
\newblock
\APACjournalVolNumPages{Econometrica: Journal of the Econometric Society}{}{}{1161--1180}.
\PrintBackRefs{\CurrentBib}

\bibitem [\protect \citeauthoryear {%
Bernheim%
}{%
Bernheim%
}{%
{\protect \APACyear {1984}}%
}]{%
bernheim1984rationalizable}
\APACinsertmetastar {%
bernheim1984rationalizable}%
\begin{APACrefauthors}%
Bernheim, B\BPBI D.%
\end{APACrefauthors}%
\unskip\
\newblock
\APACrefYearMonthDay{1984}{}{}.
\newblock
{\BBOQ}\APACrefatitle {Rationalizable strategic behavior} {Rationalizable strategic behavior}.{\BBCQ}
\newblock
\APACjournalVolNumPages{Econometrica: Journal of the Econometric Society}{}{}{1007--1028}.
\PrintBackRefs{\CurrentBib}

\bibitem [\protect \citeauthoryear {%
Downs%
}{%
Downs%
}{%
{\protect \APACyear {1957}}%
}]{%
downs1957economic}
\APACinsertmetastar {%
downs1957economic}%
\begin{APACrefauthors}%
Downs, A.%
\end{APACrefauthors}%
\unskip\
\newblock
\APACrefYearMonthDay{1957}{}{}.
\newblock
{\BBOQ}\APACrefatitle {An economic theory of political action in a democracy} {An economic theory of political action in a democracy}.{\BBCQ}
\newblock
\APACjournalVolNumPages{Journal of political economy}{65}{2}{135--150}.
\PrintBackRefs{\CurrentBib}

\bibitem [\protect \citeauthoryear {%
Hotelling%
}{%
Hotelling%
}{%
{\protect \APACyear {1929}}%
}]{%
hotbllino1929stability}
\APACinsertmetastar {%
hotbllino1929stability}%
\begin{APACrefauthors}%
Hotelling, H.%
\end{APACrefauthors}%
\unskip\
\newblock
\APACrefYearMonthDay{1929}{}{}.
\newblock
{\BBOQ}\APACrefatitle {Stability in competition} {Stability in competition}.{\BBCQ}
\newblock
\APACjournalVolNumPages{The economic journal}{39}{153}{41--57}.
\PrintBackRefs{\CurrentBib}

\bibitem [\protect \citeauthoryear {%
Kohlberg%
}{%
Kohlberg%
}{%
{\protect \APACyear {1983}}%
}]{%
kohlberg1983equilibrium}
\APACinsertmetastar {%
kohlberg1983equilibrium}%
\begin{APACrefauthors}%
Kohlberg, E.%
\end{APACrefauthors}%
\unskip\
\newblock
\APACrefYearMonthDay{1983}{}{}.
\newblock
{\BBOQ}\APACrefatitle {Equilibrium store locations when consumers minimize travel time plus waiting time} {Equilibrium store locations when consumers minimize travel time plus waiting time}.{\BBCQ}
\newblock
\APACjournalVolNumPages{Economics Letters}{11}{3}{211--216}.
\PrintBackRefs{\CurrentBib}

\bibitem [\protect \citeauthoryear {%
Pearce%
}{%
Pearce%
}{%
{\protect \APACyear {1984}}%
}]{%
pearce1984rationalizable}
\APACinsertmetastar {%
pearce1984rationalizable}%
\begin{APACrefauthors}%
Pearce, D\BPBI G.%
\end{APACrefauthors}%
\unskip\
\newblock
\APACrefYearMonthDay{1984}{}{}.
\newblock
{\BBOQ}\APACrefatitle {Rationalizable strategic behavior and the problem of perfection} {Rationalizable strategic behavior and the problem of perfection}.{\BBCQ}
\newblock
\APACjournalVolNumPages{Econometrica: Journal of the Econometric Society}{}{}{1029--1050}.
\PrintBackRefs{\CurrentBib}

\bibitem [\protect \citeauthoryear {%
Perea%
}{%
Perea%
}{%
{\protect \APACyear {2007}}%
}]{%
perea2007one}
\APACinsertmetastar {%
perea2007one}%
\begin{APACrefauthors}%
Perea, A.%
\end{APACrefauthors}%
\unskip\
\newblock
\APACrefYearMonthDay{2007}{}{}.
\newblock
{\BBOQ}\APACrefatitle {A one-person doxastic characterization of Nash strategies} {A one-person doxastic characterization of nash strategies}.{\BBCQ}
\newblock
\APACjournalVolNumPages{Synthese}{158}{2}{251--271}.
\PrintBackRefs{\CurrentBib}

\bibitem [\protect \citeauthoryear {%
Peters%
, Schr{\"o}der%
\BCBL {}\ \BBA {} Vermeulen%
}{%
Peters%
\ \protect \BOthers {.}}{%
{\protect \APACyear {2018}}%
}]{%
peters2018hotelling}
\APACinsertmetastar {%
peters2018hotelling}%
\begin{APACrefauthors}%
Peters, H.%
, Schr{\"o}der, M.%
\BCBL {}\ \BBA {} Vermeulen, D.%
\end{APACrefauthors}%
\unskip\
\newblock
\APACrefYearMonthDay{2018}{}{}.
\newblock
{\BBOQ}\APACrefatitle {Hotelling’s location model with negative network externalities} {Hotelling’s location model with negative network externalities}.{\BBCQ}
\newblock
\APACjournalVolNumPages{International Journal of Game Theory}{47}{}{811--837}.
\PrintBackRefs{\CurrentBib}

\bibitem [\protect \citeauthoryear {%
Spohn%
}{%
Spohn%
}{%
{\protect \APACyear {1982}}%
}]{%
spohn1982make}
\APACinsertmetastar {%
spohn1982make}%
\begin{APACrefauthors}%
Spohn, W.%
\end{APACrefauthors}%
\unskip\
\newblock
\APACrefYearMonthDay{1982}{}{}.
\newblock
{\BBOQ}\APACrefatitle {How to make sense of game theory} {How to make sense of game theory}.{\BBCQ}
\newblock
\BIn{} \APACrefbtitle {Philosophy of economics: Proceedings, Munich, July 1981} {Philosophy of economics: Proceedings, munich, july 1981}\ (\BPGS\ 239--270).
\newblock
\APACaddressPublisher{}{Springer}.
\PrintBackRefs{\CurrentBib}

\bibitem [\protect \citeauthoryear {%
van Sloun%
}{%
van Sloun%
}{%
{\protect \APACyear {2023}}%
}]{%
van2023rationalizable}
\APACinsertmetastar {%
van2023rationalizable}%
\begin{APACrefauthors}%
van Sloun, J.%
\end{APACrefauthors}%
\unskip\
\newblock
\APACrefYearMonthDay{2023}{}{}.
\newblock
{\BBOQ}\APACrefatitle {Rationalizable behavior in the Hotelling--Downs model of spatial competition} {Rationalizable behavior in the hotelling--downs model of spatial competition}.{\BBCQ}
\newblock
\APACjournalVolNumPages{Theory and Decision}{95}{2}{309--335}.
\PrintBackRefs{\CurrentBib}

\end{thebibliography}

\end{document}